\documentclass{article}
\usepackage[utf8]{inputenc}
\usepackage[english]{babel}

\usepackage{lmodern} 
\usepackage{amsmath}
\usepackage{amssymb}
\usepackage{amsthm}
\usepackage{amsfonts}       %
\usepackage{url}            %
\usepackage{blindtext}
\usepackage{booktabs}       %
\usepackage{enumitem}
\usepackage[margin=3.0cm,a4paper]{geometry}
\usepackage{graphicx}
\graphicspath{{images/}{../images/}}
\usepackage{microtype}      %
\usepackage{xcolor}         %
\usepackage{subcaption}
\usepackage{tikz}
\usepackage{todonotes}
\usepackage[square, numbers]{natbib}
\usepackage{bm}
\usepackage{braket}

\usepackage{hyperref}       %
\hypersetup{
    colorlinks,
	linkcolor={blue!80!black},
    citecolor={red!80!black},
    urlcolor={blue!80!black}    
}

\usepackage{cleveref}
\usepackage{tikz}
\usetikzlibrary{arrows.meta}

\newtheorem{theorem}{Theorem}
\newtheorem{proposition}{Proposition}
\newtheorem{lemma}{Lemma}
\newtheorem{corollary}{Corollary}

\newtheorem{conjecture}[theorem]{Conjecture}

\newtheorem{definition}[theorem]{Definition}
\newtheorem{example}[theorem]{Example}

\newtheorem{remark}[theorem]{Remark}

\usepackage{dsfont} %

\newcommand{\e}[1]{ {\mathrm{e}}^{ #1 } }

\newcommand{\expectation}[1]{ \mathbb{E} [ #1 ] }
\newcommand{\expectationWrt}[2]{ \mathbb{E}_{#2} [ #1 ] }
\newcommand{\expectationbig}[1]{ \mathbb{E} \bigl[ #1 \bigr] }
\newcommand{\expectationbigWrt}[2]{ \mathbb{E}_{#2} \bigl[ #1 \bigr] }
\newcommand{\expectationBig}[1]{ \mathbb{E} \Bigl[ #1 \Bigr] }
\newcommand{\expectationBigWrt}[2]{ \mathbb{E}_{#2} \Bigl[ #1 \Bigr] }

\newcommand{\pnorm}[2]{ \| #1 \|{}_{#2} }
\newcommand{\pnormbig}[2]{ \big\| #1 \big\|{}_{#2} }
\newcommand{\pnormBig}[2]{ \Bigl\| #1 \Bigr\|{}_{#2} }
\newcommand{\probability}[1]{ \mathbb{P} [ #1 ] }
\newcommand{\probabilitybig}[1]{ \mathbb{P} \bigl[ #1 \bigr] }
\newcommand{\probabilityBig}[1]{ \mathbb{P} \Bigl[ #1 \Bigr] }

\newcommand{\refLemma}[1]{{\textrm{Lemma~\ref{#1}}}}

\def\eqcom#1{\overset{\textnormal{(#1)}}}

\def\({{\Bigl(}}
\def\){{\Bigr)}}

\newcommand{\ba}{\begin{array}}
\newcommand{\ea}{\end{array}}

\newcommand{\xdeleted}[1]{\deleted{}} %

\newcommand{\N}{\mathbb{N}}
\newcommand{\Z}{\mathbb{Z}}

\newcommand{\R}{\mathbb{R}}
\newcommand{\C}{\mathbb{C}}

\newcommand{\abs}[1]{\left|#1\right|}

\newcommand{\commentout}[1]{}

\newcommand{\fnormbig}[1]{ \bigl\| {#1} \bigr\|_{\mathrm{F}}}

\newcommand{\opnorm}[1]{\pnorm{ #1}{\mathrm{op}}}
\newcommand{\opnormbig}[1]{\pnormbig{ #1}{\mathrm{op}}}

\newcommand{\tr}[1]{\mathrm{Tr}[ #1 ] }
\newcommand{\trbig}[1]{\mathrm{Tr} \bigl[ #1 \bigr] }
\newcommand{\trBig}[1]{\mathrm{Tr} \Bigl[ #1 \Bigr] }

\newcommand{\bsa}[1]{B^{\mathrm{sa}}({#1})}
\newcommand{\nnbsa}[1]{B_{^{\geq }}^{\mathrm{sa}}({#1})}
\newcommand{\pbsa}[1]{B_{^{>}}^{\mathrm{sa}}({#1})}
\newcommand{\id}[1]{\mathds{1}_{{#1}}}

\DeclareRobustCommand{\stirling}{\genfrac\{\}{0pt}{}}

\newcommand{\BigO}[1]{ O\Bigl(#1\Bigr) }

\newcommand{\ci}{\mathrm{i}}
\newcommand{\ceil}[1]{\lceil {#1} \rceil}
\newcommand{\floor}[1]{\lfloor {#1} \rfloor}
\renewcommand\bra[1]{{\langle{#1}|}}
\makeatletter
\renewcommand\ket[1]{%
  \@ifnextchar\bra{\k@t{#1}\!}{\k@t{#1}}%
}
\newcommand\k@t[1]{{|{#1}\rangle}}
\makeatother

\usepackage[acronym]{glossaries}

\newacronym{POVM}{POVM}{Positive Operator-Valued Measure}
\newacronym{PM}{PM}{Projective Measurement}
\newacronym{ONB}{ONB}{Orthonormal Basis}
\makeglossaries

\title{The suboptimality ratio of projective measurements restricted to low-rank subspaces}

\author{
  Albert Senen--Cerda\\ ~ \\ 
   \small{IRIT, LAAS--CNRS, and Universit\'e Toulouse Paul Sabatier, France}\\
 \small{\texttt{albert.senen-cerda@irit.fr}}  \\
  }
\date{}

\begin{document}

\maketitle

\begin{abstract}
Limitations in measurement instruments can hinder the implementation of some quantum algorithms.
Understanding the suboptimality of such measurements with restrictions may then lead to more efficient measurement policies.
In this paper, we theoretically examine the suboptimality arising from a Procrustes problem for minimizing the average distance between two fixed quantum states when one of the states has been measured by a \gls{PM}.
Specifically, we compare optima when we can only use \glspl{PM} that are aligned with a low-rank subspace where the quantum states are supported, and when we can measure with the full set of \glspl{PM}.
For this problem, we show that the suboptimality ratio is independent of the dimension of the full space, and is at most polylogarithmic in the dimension of the low-rank subspace.
In the proof of this result, we use a probabilistic approach and the main techniques include trace inequalities related to projective measurements, and operator norm bounds for equipartitions of Parseval frames, which are of independent interest.

\end{abstract}

\glsresetall 
 
\section{Introduction}

\glspl{PM}, and the larger set of \glspl{POVM} play a key role in quantum tomography \cite{bisio2009optimal}, state discrimination \cite{chefles2000quantum}, quantum key distribution \cite{kalev2013choice, liu2014choice}, and measurement-based quantum computing \cite{briegel2009measurement}.
In practical applications, however, due to restrictions in measurement instruments and resources, not all optimal \glspl{PM} or \glspl{POVM} can be accurately implemented.
For example, this occurs when only a limited selection of high-fidelity \glspl{PM} are available to the practitioner or when no ancilla states are possible; those typically used to implement \glspl{POVM}.
From a theoretical perspective, these restrictions are circumvented by using a smaller set of measurements plus additional resources, e.g., \glspl{POVM} can be instead simulated from \glspl{PM} with additional postselection \cite{oszmaniec2019simulating}.

If we focus on optimization problems with measurements, it is unclear which problems have \glspl{PM} as optima, and more precisely how suboptimal they become if additional restrictions exist.
In this paper, we focus on this latter issue and bound the suboptimality arising from a Procrustes problem with \glspl{PM} that are aligned with a low-rank structure of the measured states.\\

Upon measuring a quantum state with an observable on an ambient space $\mathcal{H}$, the probability that a certain outcome is observed depends on the alignment between the measured state and the state that would correspond to that outcome.
A projective measurement $\mathcal{P} = \{ P_1, \ldots, P_n \}$ consists of $n$ orthogonal projectors that characterize the outcomes of a measurement.
After measurement of a mixed quantum state $\rho$ by $\mathcal{P}$, the random outcome of the measurement $X$ is $P_i$ with probability $\tr{\rho P_i} \geq 0$ for all $i \in [n]$.
The Procrustes problem for \glspl{PM} can be formulated as finding a \gls{PM} $\mathcal{P}$ of $\rho$ that in expectation brings the random outcome state $X = X(\mathcal{P}, \rho)$ closest to another mixed quantum state $\tau$.
In Frobenious norm, the problem is equivalent to finding
\begin{equation}
\min_{\mathcal{P} \in P(\mathcal{H})} \expectationBig{\fnormbig{\tau - X(\mathcal{P}, \rho)}^{2}},
\label{eqn:intro_procrustes}
\end{equation}
where $P(\mathcal{H})$ is the set of \glspl{PM}.
Up to constants, \eqref{eqn:intro_procrustes} is equal to the following maximization problem over \glspl{PM}
\begin{equation}
\max_{\mathcal{P} \in P(\mathcal{H})} \sum_{i=1}^{n} \trbig{\tau P_i \rho P_i}.
\label{eqn:intro_maximization}
\end{equation}
This optimization problem is nonconvex, and for our setting we consider additional restrictions. 
Firstly, we assume that $\tau$, and $\rho$ have a low rank structure and are only composed of states supported in a subspace $\mathcal{S} \subset \mathcal{H}$ of dimension $r \ll n$.
Secondly, we model constraints of the measurement instrument by considering only \glspl{PM} $\mathcal{Q}$ that are aligned with this subspace $\mathcal{S}$.
Formally, this means that $r$ out of $n$ orthogonal projectors generate also a projective measurement on $\mathcal{S}$.
The set of such \glspl{PM} will generate a strict subset $P(\mathcal{S}) \subset P(\mathcal{H})$, and we can thus also define a constrained maximization problem analogous to \eqref{eqn:intro_maximization}.
A relevant question is how suboptimal is the restriction to measurements over the low-rank structure compared to using the full set of \glspl{PM}.
We can define the suboptimality ratio as the smallest constant $K_{n, r} \geq 1$ such that for any quantum states $\rho$ and $\tau$ with support in a subspace $\mathcal{S} \subset \mathcal{H}$ of dimension $r$,
\begin{equation}
\max_{\mathcal{P} \in P(\mathcal{H})} \sum_{i=1}^{n} \trbig{\tau P_i \rho P_i} \leq K_{n, r}
\max_{\mathcal{Q} \in P(\mathcal{S})} \sum_{i=1}^{r} \trbig{\tau Q_i \rho Q_i}.
\label{eqn:intro_krauss_maximization}
\end{equation}
The constant $K_{n,r}$ encodes the suboptimality of restricting to measurements in the low-rank subspace, and a low value would imply that we can use the smaller set of measurements aligned with the subspace of quantum states to approximate \eqref{eqn:intro_maximization}.
Constants that encode information on maximal suboptimality are common in quantum information.
We can namely mention Tsirelson's bound that is perhaps the most well-known and is fundamental in bounding violations of Bell's inequalities \cite{cirel1980quantum}.
For the maximization in \eqref{eqn:intro_maximization}, there are some results in the context of quantum control \cite{pechen}.
In particular, for the constant $K_{n,r}$ , when $\tau$ and $\rho$ are pure states ($r=2$) then it is known that $K_{n,2} = 1$ for all $n \geq 2$ \cite{wang2011quantum}.
For general mixed states ($r \geq 2$), the magnitude of $K_{n,r}$ is however unknown.\\

Our main contribution is to show that there exists $c>0$ such that for any $n \geq r \geq 2$,
\begin{equation}
K_{n,r} \leq c \log(r)^{4}.
\label{thm:intro_main}
\end{equation}
The bound \eqref{thm:intro_main} indicates that the suboptimality does not depend on the dimension $n$ of the ambient space $\mathcal{H}$ and is at most polylogarithmic in the dimension $r$ of the low-rank subspace $\mathcal{S}$.
As a consequence, we may approximate the solution of \eqref{eqn:intro_maximization} up to a polylogarithmic factor by optimizing over the smaller space $\mathcal{S}$ where both $\rho$ and $\tau$ are supported.

Our second contribution is technical and relates to the methods used in the proof of \eqref{thm:intro_main}. 
We use a probabilistic approach to upper bound the trace of certain Kraus maps by randomized operators while preserving their \gls{PM} structure.
Restricted to $\mathcal{S}$, the projectors $P_i$ induce vectors $\ket{v_i} \in \mathcal{S}$ that generate a Parseval frame and satisfy
\begin{equation}
\sum_{i=1}^{n} \ket{v_i}\bra{v_i} = \id{\mathcal{S}}.
\label{eqn:intro_frame_property}
\end{equation}
The measurement $\mathcal{P} \in P(\mathcal{H})$ in \eqref{eqn:intro_maximization} can be understood on $\mathcal{S}$ as sum of a \gls{PM} $\mathcal{Q} \in P(\mathcal{S})$ that is weighted by $T = \ceil{n/r}$ matrices $L_t$ obtained from some equipartition of the frame vectors
\begin{equation}
\sum_{t=1}^{T} \trBig{\tau L_t \mathcal{Q}\bigl[ L_t^{\dagger} \rho\ L_t \bigr] L_t^{\dagger}} \quad \text{ where } \quad  \mathcal{Q}[a] = \sum_{i=1}^{r} Q_i a Q_i.
\label{eqn:intro_sum_weighted_operators}
\end{equation} 
Similar to \eqref{eqn:intro_frame_property}, the weight matrices satisfy 
\begin{equation}
\sum_{t=1}^{T} L_t L_t^{\dagger} = \id{\mathcal{S}}.
\label{eqn:intro_partition_identity}
\end{equation}
The difficulty of the analysis lies in decoupling the dimension $n$ from \eqref{eqn:intro_sum_weighted_operators} while preserving the \glspl{PM} structure of the measurements involved.
To tackle this problem we use a probabilistic approach and find a single random measurement that in expectation behaves as the sum of the measurements in \eqref{eqn:intro_sum_weighted_operators}.
Specifically, we consider Gaussian random weight matrices of the type
\begin{equation}
\hat{L} = \sum_{t=1}^{T} g_t L_t, \quad  \text{where}\quad g_t \sim \mathrm{N}(0,1) \quad \text{ are i.i.d.\ for all } t \in [T].
\label{eqn:intro_randomization}
\end{equation}
By using this random weight in a measurement, we transfer the dependence on the dimension $n$ from the sum of weighted \gls{PM} in \eqref{eqn:intro_sum_weighted_operators} to the variance of \eqref{eqn:intro_randomization}.

Additional key steps in the proof include showing the existence of a partition $\pi$ of the frames that satisfies variance constraints \emph{independent} of $n$. 
The probabilistic method from combinatorics is used in this step to bound the operator norm of a frame variance matrix. 
In this way, with high probability we obtain
\begin{equation}
\opnorm{\hat{L}} < c \log(r).
\end{equation}
We also introduce new trace inequalities to extract the weights $L_t$ from the traces in \eqref{eqn:intro_sum_weighted_operators} while preserving the \glspl{PM} structure of the measurements.
These inequalities are based on complex interpolation results.

Remarkably, our proof method directly relates the polylogarithmic term in \eqref{thm:intro_main} with the additional dimensional factor that appears in well-known random matrix concentration inequalities \cite{tropp2015introduction} compared to their univariate counterparts.
This fact suggests that our method cannot be further improved without using additional information about the optima of \eqref{eqn:intro_maximization}.\\

\subsection{Related Literature}

Procrustes problems have been previously studied in the literature.
We can namely mention Procrustes problems for orthogonal matrices \cite{schonemann1966generalized}, weighted versions with quadratic constraints \cite{gower2004procrustes}, conic constraints \cite{baghel2022non, fulova2023solving}, and problems with manifold constraints \cite{bhatia2019procrustes, absil2008optimization}.

The maximization problem in \eqref{eqn:intro_maximization} appears in the context of quantum control with nonselective measurements in \cite{pechen}, and for general Kraus maps in \cite{pechen2008control}. Special cases are studied in \cite{wang2011quantum}, including the case $r=2$ and $n \geq 2$. In particular $K_{n,2} = 1$ is shown using the first order stationary conditions; see \eqref{eqn:critical_point_condition} in Section~\ref{sec:procrustes_problems}.
By relaxing in \eqref{eqn:intro_maximization} the optimization to general Kraus maps, no local maxima exist for $n=2$ \cite{pechen2008control}, but this may not be the case for projective measurements.

From a theoretical perspective, the analysis of the suboptimality in optimization problems is commonly conducted to approximate solutions of computationally hard problems.
Related to the quantum setting, for example, we can mention semi-definite program relaxations for generalized orthogonal Procrustes problems \cite{bandeira2016approximating, Won2018OrthogonalTM} that are related to the well-known Grothendieck's inequality \cite{blei2014grothendieck}.
Grothendieck's constant \cite{krivine1978constantes} characterizes maximal inequality bounds and can be interpreted as a suboptimality ratio between optimization problems with different constraints on the dimensions \cite{briet2014grothendieck}, which is similar to \eqref{eqn:intro_krauss_maximization}.
We remark, however, that these constants are universally bounded and do not grow indefinitely with the intrinsic dimension of the problem---the dimension of $\mathcal{S}$ in our setting---except if there are additional constraints, such as those induced by a graph structure \cite{alon2006quadratic}.

We use key properties of Parseval frames \cite{waldron2018introduction} such as \eqref{eqn:intro_frame_property}.
In particular, we partition the frame into equal-sized sets and try to impose an orthogonal structure on each partition.
Orthogonalization has been studied before in \cite{de2021orthogonalization} for \glspl{POVM}, where it is shown that if a \gls{POVM} is close to being orthogonal, then there exists a \gls{PM} that is actually close to the \gls{POVM}.
Differently in our problem, the frame in \eqref{eqn:intro_frame_property} may be far from orthogonal, and we have additional constraints such as the dimension of the \glspl{PM} that we seek.

In the proof of \eqref{thm:intro_main}, we use the probabilistic method \cite{alon2016probabilistic} to prove the existence of operators with certain properties by showing that they occur with positive probability.
In quantum information, for example, this method can be used to find low-rank approximations of Kraus maps \cite{lancien2024approximating} by using random Kraus operators.
In our case, we use the random matrix \eqref{eqn:intro_randomization} to randomize, and upper bound \eqref{eqn:intro_sum_weighted_operators} while preserving the \gls{PM} structure of the operators. 
Concentration inequalities for Gaussian random matrices from \cite{tropp2015introduction} are then used to estimate the operator norm provided that its variance can be controlled.
Operator norm bounds for random matrices have been used in the quantum setting to understand quantum expanders \cite{lancien2024optimal, hastings2007random}, and other random quantum channels \cite{gonzalez2018spectral, lancien2024limiting}.
In our setting, to show that the variance of \eqref{eqn:intro_randomization} is small we consider a random partition $\pi$ and bound the operator norm of a random variance matrix of $\hat{L}$ induced by this partition. 
We use the method of moments to estimate the operator norm \cite{tao2012topics}, but differently from the typical concentration bounds that assume independence, in our case the entries of the matrix have strong dependencies induced by the frame.

Finally, we also use operator trace inequalities such as trace convexity \cite{lieb1973convex}, and in Section~\ref{sec:interpolation_bound} we show new trace inequalities involving projective measurements using complex interpolation techniques.
We can mention \cite{sutter2017multivariate} where interpolation inequalities for traces of matrix products are shown using a similar approach.

\section{Preliminaries}
\label{sec:preliminaries}

Let $\mathcal{H}$ be a complex finite-dimensional Hilbert space with inner product $\langle \cdot | \cdot \rangle : \mathcal{H} \times \mathcal{H} \to \C$, and such that $\dim(\mathcal{H}) = n$. 
Let $B(\mathcal{H})$ be the space of linear maps of $\mathcal{H}$, and let 
\begin{align*}
\bsa{\mathcal{H}} & =\bigl\{ A \in B(\mathcal{H}) \mid A = A^{\dagger} \bigr\},
\end{align*}
be the set of self-adjoint operators of $\mathcal{H}$.
We also let $\nnbsa{\mathcal{H}}$, and $\pbsa{\mathcal{H}}$ be the sets of nonnegative selfadjoint, and positive selfadjoint operators, respectively.
The set of density matrices of $\mathcal{H}$ is defined by
\begin{equation}
D(\mathcal{H}) = \big\{  \rho \in \nnbsa{\mathcal{H}}, \tr{\rho} = 1\big\}.
\end{equation}

The set of (nondegenerate) \glspl{PM} in $\mathcal{H}$ is composed of sets of $n$ orthogonal rank-one projectors
\begin{equation}
P(\mathcal{H}) = \big\{ (P_1, \ldots, P_n) \in B(\mathcal{H})^{\times n}  : P_i = P_i^{\dagger},\ P_i P_j = \delta_{ij}P_i,\ \mathrm{rk}(P_i)=1 \text{ for } i \in [n] \big\}.
\label{eqn:def_PMs}
\end{equation}
For any $\mathcal{P} \in P(\mathcal{H})$, the following partition of unity property for the identity operator $\id{\mathcal{H}} \in B(\mathcal{H})$ holds
\begin{equation}
\sum_{i=1}^{n}  P_i = \id{\mathcal{H}}.
\end{equation}

When a mixed quantum state encoded by a density matrix $\rho \in D(\mathcal{H})$ is measured by a \gls{PM} $\mathcal{P} \in P(\mathcal{H})$, the outcome state $X$ is $P_i \in D(\mathcal{H})$ with probability $\tr{P_i \rho} \geq 0$.
Let $\mu(\rho, \mathcal{P})$ be the distribution of the random density matrix $X$ describing the outcome state after measuring $\rho$ with $\mathcal{P} \in P(\mathcal{H})$.
The expected state after measurement is then given by 
\begin{equation}
\mathcal{P}[\rho] = \expectationWrt{X}{X \sim \mu(\rho, \mathcal{P})} = \sum_{i=1}^{n} \trbig{\rho P_i} P_i = \sum_{i=1}^{n} P_i \rho P_i.
\end{equation}

\subsection{Procrustes problem for projective measurements}
\label{sec:procrustes_problems}
Given a target quantum state $\tau \in D(\mathcal{H})$, we can ask what is the expected distance between $\tau$, and the outcome state of the measurement of $\rho$ by $\mathcal{P}$. 
The Procrustes problem for \gls{PM} consists of finding for which measurement $\mathcal{P} \in P(\mathcal{H})$ the expected distance is minimized. 
In the Frobenious norm, $\fnormbig{A}^{2} = \tr{AA^{\dagger}}$, the Procrustes problem is
\begin{align}
\min_{\mathcal{P} \in P(\mathcal{H})} \expectationBigWrt{\fnormbig{\tau - X}^{2}}{X \sim \mu(\rho, \mathcal{P})} &= \min_{\mathcal{P} \in P(\mathcal{H})} \sum_{i=1}^{n} \tr{P_i \rho} \fnormbig{\tau - P_i}^{2} \nonumber \\
& = \min_{\mathcal{P} \in P(\mathcal{H})} - \sum_{i=1}^{n} 2\trbig{P_i \rho}\trbig{P_i \tau} + C(\rho, \tau) \nonumber \\
& = -2 \max_{\mathcal{P} \in P(\mathcal{H})} \trbig{\rho \mathcal{P}[\tau]} + C(\rho, \tau),
\label{eqn:selective_procrustes}
\end{align}
where the last equality holds since projectors are rank-one and so for all $i \in [n]$, $\tr{P_i \rho}\tr{P_i \tau} = \tr{P_i \rho P_i \tau}$.
Both $\tau$ and $\rho$ are fixed, so both \eqref{eqn:selective_procrustes} and \eqref{eqn:intro_maximization} are equivalent up to a constant.

The maximization problem \eqref{eqn:intro_maximization} is symmetric over $\rho$ and $\tau$, and can be posed as an optimization over the unitary group $U(\mathcal{H}) = \{ U \in B(\mathcal{H}) : U U^{\dagger} = U U^{\dagger} = \id{\mathcal{H}} \}$.
We let the measurement $\tilde{\mathcal{P}}(U) = \{ U \tilde{P}_1 U^{\dagger}, \ldots, U \tilde{P}_n U^{\dagger} \}$ be dependent on $U \in U(\mathcal{H})$ for some $\tilde{\mathcal{P}} \in P(\mathcal{H})$ fixed.
Then we can consider the equivalent maximization problem
\begin{equation}
\max_{U \in U(\mathcal{H})} \trbig{\rho \tilde{\mathcal{P}}(U)[\tau]}.
\label{eqn:maximization_unitary}
\end{equation}
In \cite{wang2011quantum}, the first order stationary conditions for critical points of \eqref{eqn:maximization_unitary} were computed: a critical point $U^{\star}$ with measurement $\mathcal{P}^{\star} = \mathcal{P}(U^{\star})$ must satisfy
\begin{equation}
\bigl[\rho, \mathcal{P}^{\star}[\tau] \bigr] = \bigl[\mathcal{P}^{\star}[\rho], \tau \bigr],
\label{eqn:critical_point_condition}
\end{equation}
where $[A, B] = AB - BA$ is the commutator.
This problem \eqref{eqn:maximization_unitary} is nonconvex, and to our knowledge it is unknown how complex the optimization landscape is, e.g., how many critical points exist satisfying \eqref{eqn:critical_point_condition}.
In Section~\ref{sec:numerics}, we will use manifold optimization \cite{absil2008optimization} to try to numerically find \eqref{eqn:maximization_unitary}.

\subsection{Restriction to aligned projective measurements}

\label{sec:reduction}

For the maximization problem in \eqref{eqn:intro_maximization}, we impose some restrictions in the set of measurements as well as some additional structure that we can exploit.
To model the limitations of measurement instruments, we assume that we can only use projective measurements aligned with a subspace smaller than $\mathcal{H}$.
Furthermore, we assume there is an additional low-rank structure in the quantum states that is also aligned with the aforementioned subspace.
We formalize these ideas as follows.

Firstly, we assume that both $\tau$ and $\rho$ belong to a low-rank subspace $\mathcal{S} \subset \mathcal{H}$ with $\dim(\mathcal{S}) = r$. 
Specifically, for any vector $\ket{v} \in S^{\perp}$---the orthogonal subspace in $\mathcal{H}$ of $\mathcal{S}$---$\rho \ket{v} = \tau \ket{v} = 0$, and all eigenvectors with positive eigenvalue of $\rho$ and $\tau$ belong to $\mathcal{S}$.
Secondly, to model the restriction of certain measurements, we will consider the maximization only over a subset of \glspl{PM} that commute with the unique orthogonal projector $\Pi_{\mathcal{S}}: \mathcal{H} \to \mathcal{S}$ from $\mathcal{H}$ onto $\mathcal{S}$ and is related to the set of \glspl{PM} of $\mathcal{S}$.
As in \eqref{eqn:def_PMs}, we define
\begin{equation}
P(\mathcal{S}) = \big\{ (Q_1, \ldots, Q_r) \in B(\mathcal{S})^{\times r} : Q_i = Q_i^{\dagger},\ Q_i Q_j = \delta_{ij} Q_i,\ \mathrm{rk}(Q_i)=1 \text{ for } i\in [r]\big\}.
\label{eqn:def_restricted_PMs}
\end{equation}
There is an injective map $\iota: P(\mathcal{S}) \to P(\mathcal{H})$ given by completing the set of $r$ orthogonal projectors of $\mathcal{Q} \in P(\mathcal{S})$ in $\mathcal{H}$ with orthogonal projectors from a fixed $\mathcal{Q}^{\perp} \in P(\mathcal{S}^{\perp})$.
Since $\rho$, and $\tau \in D(\mathcal{S})$ are invariant under $\Pi_{\mathcal{S}}$, we have the equality
\begin{align}
\trbig{\tau \mathcal{Q}[\rho]} &= \trbig{\tau \iota(\mathcal{Q})[\rho]} \text{ for any } \mathcal{Q} \in P(\mathcal{S}).
\label{eqn:restriction_PM}
\end{align}
We consider then the set of restricted measurements to be $P(\mathcal{S})$, which is independent of $\mathcal{H}$, and define the restricted optimization problem by the maximization
\begin{equation}
\max_{\mathcal{Q} \in P(\mathcal{S})} \trbig{\tau \mathcal{Q}[\rho]}.
\label{eqn:maximization_restricted}
\end{equation}

\section{Suboptimality ratio of projective measurements}
\label{sec:suboptimality}

On the Hilbert space $\mathcal{H}$ with $\dim(\mathcal{H}) = n$, the suboptimality ratio $K_{n,r}$ is defined as the smallest constant such that for any $\rho$, and $\tau$ supported on the subspace $\mathcal{S} \subseteq \mathcal{H}$ of fixed dimension $r \leq n$, we have
\begin{equation}
\max_{\mathcal{P} \in P(\mathcal{H})} \trbig{\tau \mathcal{P}[\rho]} \leq K_{n,r} 
\max_{\substack{\mathcal{Q} \in P(\mathcal{S})}} \trbig{\tau \mathcal{Q}[\rho]}.
\label{eqn:maximization}
\end{equation}
The constant $K_{n,r} \geq 1$ encodes how suboptimal restricted measurements are when there is an aligned low-rank structure present in the quantum states compared to using a larger set of \glspl{PM}.

Our main result in Theorem~\ref{thm:main} bounds the suboptimality ratio $K_{n,r}$ when the quantum states $\tau$, and $\rho$ are supported in $\mathcal{S}$, and the measurements are also restricted to this subspace. 
Its proof can be found in Section~\ref{sec:proof_theorem}.
\begin{theorem}
Let $\mathcal{H}$ be a complex Hilbert space with $\dim(\mathcal{H}) = n$, and let $\mathcal{S} \subseteq \mathcal{H}$ be a Hilbert subspace with $\dim(\mathcal{S}) = r \leq n$.
There exists $c>0$ such that for any $n \geq r \geq 2$, and any $\rho$, $\tau \in D(\mathcal{S})$, 
\begin{equation}
\max_{\mathcal{P} \in P(\mathcal{H})} \trbig{\tau \mathcal{P}[\rho]} \leq c \log(r)^{4}
\max_{\substack{\mathcal{Q} \in P(\mathcal{S})}} \trbig{\tau \mathcal{Q}[\rho]}.
\end{equation}
\label{thm:main}
\end{theorem}

From Theorem~\ref{thm:main}, we obtain that the bound of $K_{n,r}$ is independent of the dimension of the ambient Hilbert space $\mathcal{H}$, and at most polylogarithmic in the dimension of the subspace $\mathcal{S}$.
By choosing restricted measurements in $\mathcal{S}$, we lose at most a polylogarithmic factor in \eqref{eqn:maximization}.\\
Remarkably, as can be seen in the proof of Theorem~\ref{thm:main}, the polylogarithmic dependence in the bound of $K_{n,r}$ arises from a connection to random matrix theory.
Specifically, in random matrix concentration results there is typically an additional dimensional factor---the matrix size $r$---that appears compared to their counterparts in univariate concentration inequalities.
In our case, this factor determines the logarithmic dependence and is a consequence of our probabilistic approach to showing the result; see Remark~\ref{remark:dimensional_dependence} in Section~\ref{sec:concentration}.\\

For the Procrustes problem in \eqref{eqn:intro_procrustes} we can approximate the optimum by restricting to the subspace $\mathcal{S}$. 
Theorem~\ref{thm:main}, and \eqref{eqn:selective_procrustes} immediately yield a rough approximation inequality in Corollary~\ref{cor:approximation_inequality}.
\begin{corollary}
Suppose $\tau, \rho \in D(\mathcal{S}) \subset D(\mathcal{H})$, where $\mathcal{S} \subset \mathcal{H}$ is a $r$-dimensional subspace of $\mathcal{H}$. Then, there exists a universal constant $c>0$ such that
\begin{equation}
\min_{\mathcal{Q} \in P(\mathcal{S})} \expectationBigWrt{\fnormbig{\tau - X}^{2}}{X \sim \mu(\rho, \mathcal{Q})} - 2(c\log(r)^4 - 1) \trbig{\tau \mathcal{Q}[\rho]} \leq \min_{\mathcal{P} \in P(\mathcal{H})} \expectationBigWrt{\fnormbig{\tau - X}^{2}}{X \sim \mu(\rho, \mathcal{P})}.
\label{eqn:approximation_inequality}
\end{equation}
\label{cor:approximation_inequality}
\end{corollary}

We remark that Corollary~\ref{cor:approximation_inequality} is not sharp since the bound of $K_{n,r}$ does not include information about the relationship between \emph{absolute} values in \eqref{eqn:maximization}, only their \emph{relative} values.
In particular, there may be cases when the inequality in \eqref{eqn:approximation_inequality} is trivial since the left-hand side may be negative.
This fact raises the question if the bound of $K_{n,r}$ in Theorem~\ref{thm:main} can be furhter improved.

In Section~\ref{sec:numerics}, we use numerical simulations to examine the constant $K_{n,r}$ for different values of $r$, and $n$. 
The simulations suggest that $K_{n,r} \simeq 1$ for small values of $r$, but a more complex behavior cannot be ruled out.
\section{Proof of Theorem~\ref{thm:main}}
\label{sec:proof_theorem}

We present a sketch of the proof of Theorem~\ref{thm:main} in the following steps (i)--(vi), each one corresponding to Sections~\ref{sec:reduction_weighted}--\ref{sec:final_steps}, respectively.
\begin{itemize}
\item[(i)] \emph{Reduction to weighted \glspl{PM}.} In Section~\ref{sec:reduction_weighted}, we fix a \gls{PM} $\mathcal{P} \in P(\mathcal{H})$ and we examine the effect that the measurement has restricted to the subspace $\mathcal{S}$ where $\rho$ and $\tau$ are supported. 
We show that the restriction of $\mathcal{P}$ can be decomposed as $T = \ceil{n/r}$ `weighted' measurements over $\mathcal{S}$
\begin{equation}
\trbig{\rho \mathcal{P}[\tau]} = \sum_{i=1}^{T} \trbig{\rho \mathcal{Q}_{L^{\pi}_i}[\tau]}.
\label{eqn:decomposition_weighted_PM}
\end{equation}
The matrices $\{L^{\pi}_t\}_{t=1}^{T}$ are composed of some equipartition of vectors $\ket{v_i}_{i=1}^{n}$ from a Parseval frame induced by the orthogonal projectors in $\mathcal{P}$. 
Moreover, the equipartition of the vectors can be chosen with a permutation $\pi$ of their labels.
\item[(ii)] \emph{A good partition of Parseval frames exists.} In Section~\ref{sec:partitions_frames}, we find a partition of the Parseval frame from the previous decomposition such that not only the matrices $\{L^{\pi}_t\}_{t=1}^{T}$ satisfy the normalization condition corresponding to frames
\begin{equation}
\sum_{t=1}^{T} L^{\pi}_t(L^{\pi}_t)^{\dagger} = \id{\mathcal{S}},
\label{eqn:partition_identity}
\end{equation}
but also
\begin{equation}
\mathcal{L}^{\pi} = \sum_{t=1}^{T} (L^{\pi}_t)^{\dagger} L^{\pi}_t \simeq \id{\mathcal{S}},
\label{eqn:partition_identity_transposed}
\end{equation}
as much as possible.
We use the probabilistic method from combinatorics with $\mathcal{L}^{\pi}$, and the method of moments from random matrix theory to bound the operator norm $\opnorm{\mathcal{L}^{\pi}} = \sup_{|x| = 1} |\mathcal{L}^{\pi}x|$.
We show that there exists a permutation $\pi$ of the vectors, and its corresponding equipartition such that $\opnorm{\mathcal{L}^{\pi}} \leq c \log(r)$.
\item[(iii)] \emph{Randomization}. In order to avoid the fact that we have $T=\ceil{n/r}$ summands in \eqref{eqn:decomposition_weighted_PM}, in Section~\ref{sec:randomization} we use a random weight matrix
\begin{equation}
\hat{L} = \sum_{t=1}^{T} g_t L_t^{\pi},
\label{eqn:random_matrix}
\end{equation}
with $g_t$ for $t \in [T]$ i.i.d.\ $\mathrm{N}(0,1)$ Gaussian random variables.
Moreover, we show using operator trace inequalities that
\begin{equation}
\sum_{i=1}^{T} \trbig{\rho \mathcal{Q}_{L^{\pi}_i}[\tau]} \leq 3 \expectationBig{\trbig{\rho \mathcal{Q}_{\hat{L}}[\tau]}},
\label{eqn:upper_bound_randomization}
\end{equation}
where $\mathcal{Q}_{\hat{L}}$ is a random measurement on $\mathcal{S}$ weighted by $\hat{L}$.
\item[(iv)] \emph{Concentration.} In Section~\ref{sec:concentration} we use well-known concentration inequalities for Gaussian random matrices to show concentration of the operator norm $\opnorm{\hat{L}}$.
The variance of $\hat{L}$ crucially becomes independent of $n$ due to \eqref{eqn:partition_identity}, and \eqref{eqn:partition_identity_transposed}.

\item[(v)] \emph{Interpolation.} We extract the weights from the weighted measurement $\mathcal{Q}_{\hat{L}}$ in \eqref{eqn:upper_bound_randomization}.
Specifically, in Section~\ref{sec:interpolation_bound} we show by using complex interpolation inequalities that for any weighted measurement $\mathcal{Q}_{L}$ with $L \in B(\mathcal{H})$, there exists $\mathcal{Z} \in P(\mathcal{S})$ such that 
\begin{equation}
\trbig{\rho \mathcal{Q}_{L}[\tau]} \leq \opnorm{L}^4 \trbig{\rho \mathcal{Z}[\tau]}.
\end{equation}
\item[(vi)] \emph{Combining all steps.} Finally, in Section~\ref{sec:final_steps} the previous interpolation inequality implies that there exists a constant $C>0$ and $\mathcal{Y} \in P(\mathcal{S})$ such that
\begin{equation}
\expectationBig{\trbig{\rho \mathcal{Q}_{\hat{L}}[\tau]}} \leq C \expectationbig{ \opnorm{\hat{L}}^4 } \trbig{\rho \mathcal{Y}[\tau]},
\label{eqn:upper_bound_interpolation}
\end{equation}
where $\hat{L}$ is the random matrix in \eqref{eqn:random_matrix}.
A polylogarithmic bound in $r$ for $\expectationbig{ \opnorm{\hat{L}}^4 }$ is obtained by using the concentration properties from step (iv).
Finally, from \eqref{eqn:decomposition_weighted_PM}, \eqref{eqn:upper_bound_randomization}, and \eqref{eqn:upper_bound_interpolation}, the claim of Theorem~\ref{thm:main} follows.
\end{itemize}

\subsection{Reduction to weighted \glspl{PM}}
\label{sec:reduction_weighted}

We first show that the projections of a fixed measurement $\mathcal{P} \in P(\mathcal{H})$ in \eqref{eqn:selective_procrustes} to $\mathcal{S}$ is equivalent to a sum over \emph{weighted} \glspl{PM} in $P(\mathcal{S})$. 
\begin{definition}[Weighted \glspl{PM}]
Let $\mathcal{Q}= (Q_1, \ldots, Q_{r}) \in P(\mathcal{S})$ and $L \in B(\mathcal{S})$. A weighted \gls{PM} is a map $\mathcal{Q}_{L}: D(\mathcal{S}) \to B(\mathcal{S})$ defined for any $\rho \in D(\mathcal{S})$ by
\begin{equation}
\mathcal{Q}_{L}[\rho] = \sum_{i=1}^r L Q_i L^{\dagger} \rho L Q_i L^{\dagger} =  L \mathcal{Q}\bigl[ L^{\dagger} \tau L\bigr] L^{\dagger}.
\end{equation}
\label{def:contracted_PM}
\end{definition}
From Definition~\ref{def:contracted_PM}, given $\mathcal{Q} \in P(\mathcal{S})$ and $L \in B(\mathcal{S})$ we can write for $\tau, \rho \in D(\mathcal{H})$
\begin{equation}
\trbig{\rho \mathcal{Q}_{L}[\tau]} = \trbig{\rho L \mathcal{Q}[L^{\dagger} \tau L] L^{\dagger}}.
\label{eqn:definition_contracted_PM}
\end{equation}

The following lemma poses the optimization problem in \eqref{eqn:intro_maximization} as a maximum over a sum of weighted \glspl{PM} in $\mathcal{S}$ that are dependent. 
This is a special case of the well-known Naimark dilation theorem. 
\begin{lemma}
Let $\mathcal{S} \subseteq \mathcal{H}$ be an inclusion of finite-dimensionial Hilbert spaces with $\dim(\mathcal{S})=r$, and $\dim(\mathcal{H}) = n$.
Let $\rho, \tau \in D(\mathcal{S}) \subseteq D(\mathcal{H})$ and $\mathcal{P} \in P(\mathcal{H})$. 
There exists $\mathcal{Q} \in P(\mathcal{S})$, and $L_1, \ldots, L_{T} \in B(\mathcal{S})$ with $T = \ceil{n/r}$ such that 
\begin{equation}
\trbig{\rho \mathcal{P}[\tau]} = \sum_{i=1}^{T} \trbig{\rho \mathcal{Q}_{L_i}[\tau]},
\end{equation}
and
\begin{equation}
\sum_{t=1}^{T} L_t L_t^{\dagger}  = \mathds{1}_{\mathcal{S}}.
\label{eqn:reduction_parition_of_unity}
\end{equation}
\label{lem:reduction_sum_PMs}
\end{lemma}
\begin{proof}
We have $\mathcal{P}= (P_1, \ldots, P_{n}) \in P(\mathcal{H})$ with $P_i = \ket{u_i}\bra{u_i}$ for $i \in [n]$, and $\{\ket{u_i}\}_{i=1}^{n}$ is an \gls{ONB} of $\mathcal{H}$. 
Denote the orthogonal projection from $\mathcal{H}$ to $\mathcal{S}$ by $\Pi_{S}:\mathcal{H} \to \mathcal{S}$, and define $\ket{v_i} = \Pi_{S} \ket{u_i} \in \mathcal{S}$ for all $i \in [n]$. 
By using the fact that (i) $\rho  = \Pi_{S} \rho \Pi_{S}^{\dagger}$ and similarly for $\tau$, the following equality holds,
\begin{align}
\trbig{\rho \mathcal{P}[\tau]} &= \sum_{i=1}^{n} \trbig{\rho P_i \tau P_i} \nonumber \\
&\eqcom{i} = \sum_{i=1}^{n} \trbig{\rho \Pi_{S} P_i \Pi_{S}^{\dagger} \tau \Pi_{S} P_i \Pi_{S}^{\dagger}}, \nonumber \\
&= \sum_{i=1}^{n} \trbig{\rho \ket{v_i}\bra{v_i} \tau \ket{v_i}\bra{v_i}}.
\end{align}
Out of the $n$ vectors, we consider $T$ disjoint sets of $r$ vectors, where we complete with the null vector if necessary.
If $m = \ceil{n/r}r$ is the total number of vectors, we may chose the equipartition $\{\ket{v_i}\}_{i=1}^{r}, \{\ket{v_i}\}_{i=r+1}^{2r}, \ldots, \{\ket{v_i}\}_{i=T(r-1)+1}^{Tr}$.
Let $\{\ket{e_i}\}_{i=1}^{r}$ be a \gls{ONB} of $\mathcal{S}$ and let $\mathcal{E} = (E_1, \ldots, E_r) \in P(\mathcal{S})$ be the \gls{PM} in this basis, that is, $E_i = \ket{e_i} \bra{e_i}$ for $i \in [r]$.
We define the operators
\begin{equation}
L_t = \ket{v_{r(t-1) + 1}} \bra{e_1} + \ldots + \ket{v_{rt}} \bra{e_r} =  \sum_{l=1}^{r} \ket{v_{r(t-1) + l}} \bra{e_l} \in B(\mathcal{S}) \text{ for } t \in [T],
\label{eqn:proof_lemma_reduction_1}
\end{equation}
and note that $L_t E_l L_t^{\dagger}= \ket{v_{r(t-1) + l}}\bra{v_{r(t-1) + l}} $ for all $l \in [r]$, and $t \in [T]$.
In the notation of \eqref{eqn:definition_contracted_PM}, we have now
\begin{align}
\sum_{i=1}^{n} \trbig{\rho P_i \tau P_i} &= \sum_{t = 1}^{T} \sum_{i=1}^{r} \trbig{\rho L_t E_i L_t^{\dagger} \tau L_t E_i L_t^{\dagger}} \nonumber \\
&= \sum_{t = 1}^{T} \trbig{\rho \mathcal{E}_{L_t}[\tau]}.
\end{align}
Moreover, from the condition that $\sum_{i=1}^{n} P_i = \mathds{1}_{\mathcal{H}}$, and $\sum_{l=1}^{r} E_l = \mathds{1}_{\mathcal{S}} = \Pi_{\mathcal{S}}$ we obtain that 
\begin{align}
\mathds{1}_{\mathcal{S}} = \Pi_{S} \mathds{1}_{\mathcal{H}} \Pi^{\dagger}_{S} = \sum_{i=1}^{n} \Pi_{S} P_i \Pi_{S}^{\dagger}
= \sum_{i=1}^{m} \ket{v_i}\bra{v_i} 
= \sum_{t = 1}^{T} \sum_{i=1}^{r} L_t E_i L_t^{\dagger} 
= \sum_{t = 1}^{T} L_t L_t^{\dagger}.
\end{align}
\end{proof}

The vectors $\{\ket{v_i}\}_{i=1}^{m}$ in the proof of Lemma~\ref{lem:reduction_sum_PMs} constitute a Parseval frame of $\mathcal{S}$, since
\begin{equation}
\sum_{i=1}^{m} \ket{v_i}\bra{v_i} = \id{\mathcal{S}},
\end{equation}
which implies that $|v_i| \leq 1$, and
\begin{equation}
\sum_{i=1}^{m} |v_i|^2 = r.
\label{eqn:proof_variance_example_normalization}
\end{equation}

\subsection{Equipartitions of Parseval frames}
\label{sec:partitions_frames}

In Section~\ref{sec:reduction_weighted} we have restricted a \gls{PM} $\mathcal{P} \in P(\mathcal{H})$ to $\mathcal{S}$, and we have show in Lemma~\ref{lem:reduction_sum_PMs} that it is equivalent to a sum over weighted \glspl{PM} in $\mathcal{S}$.
To later introduce randomization in Section~\ref{sec:randomization} without adding a large variance, we additionally need the weight matrices to be approximately symmetric.
In particular, while we have the partition of unity property in \eqref{eqn:reduction_parition_of_unity}, its counterpart
\begin{align}
\mathcal{L} &= L_1^{\dagger} L_1 + \cdots + L_{T}^{\dagger} L_{T}, \quad \quad \text{ where } \nonumber \\
(\mathcal{L})_{ij} &= \sum_{t=1}^{T} \braket{v_{r(t-1) + i} \mid v_{r(t-1) + j}}
\label{eqn:def_L}
\end{align}
can be very far away from the identity $\id{\mathcal{S}}$. 
To overcome this issue, we will exploit two invariant properties of the decomposition that we have not used in the definition of $\{L_t\}_{t \in [T]}$ and Lemma~\ref{lem:reduction_sum_PMs}.
Without loss of generality, we will assume that $T = n/r \in \N$, since we can always complete the frame with the null vector to satisfy the constraint.
\begin{itemize}
\item For any permutation $\pi \in S_n$, we can define similar matrices $\{L^{\pi}_t\}_{t \in [T]}$ by using the permuted set of vectors $\{ \ket{v_{\pi(i)}} \}_{i \in m}$ which satisfy Lemma~\ref{lem:reduction_sum_PMs}, but do not leave \eqref{eqn:def_L} invariant.
\item We can add a phase modifier $\Theta = (\theta_1, \ldots, \theta_n) \in \R^{n}$ to the frame $\{\ket{v_i}\}_{i=1}^{n}$ by defining for $j \in [n]$, $\ket{v_j}(\Theta) = \exp(\ci \theta_j) \ket{v_j}$  such that Lemma~\ref{lem:reduction_sum_PMs} also holds for this frame, but \eqref{eqn:def_L} is not invariant.
\end{itemize}
For a permutation $\pi \in S_n$, and a frame modifier $\Theta = (\theta_1, \ldots, \theta_n) \in \R^{n}$, we define for $t \in [T]$
\begin{equation}
L^{\pi}_t(\Theta) = \e{\ci \theta_{\pi(r(t-1) + 1)}}\ket{v_{\pi(r(t-1) + 1)}}\bra{e_1} + \ldots + \e{\ci \pi(\theta_{rt})} \ket{v_{\pi(rt)}}\bra{e_r} \in B(\mathcal{S}).
\label{eqn:definition_L}
\end{equation}
The matrices $\{L^{\pi}_i(\Theta)\}_{i=1}^{T}$ also satisfy Lemma~\ref{lem:reduction_sum_PMs} for any $\pi \in S_n$ and $\Theta \in \R^{n}$.
The generalization of \eqref{eqn:def_L} for $\pi \in S_n$ and $\Theta \in \R^{n}$ is then
\begin{align}
\mathcal{L}^{\pi}(\Theta) &= L^{\pi}_1(\Theta)^{\dagger}L^{\pi}_1(\Theta) + \ldots + L^{\pi}_{T}(\Theta)^{\dagger}L^{\pi}_{T}(\Theta).
\label{eqn:def_L_pi_theta}
\end{align}
Our aim is to find $\pi \in S_n$ and $\Theta \in \R^{n}$ such that $\opnorm{\mathcal{L}^{\pi}(\Theta)}$ is as small as possible.
Using the probabilistic method from combinatorics we can show the following.
\begin{proposition}
Let $\mathcal{L}^{\pi}(\Theta)$ be defined in \eqref{eqn:def_L_pi_theta}. 
There exists a constant $c > 0$ such that for any $r\geq 2$ and $T \geq T(r) \geq 1$, there exists a permutation $\pi^{\star} \in S_n$, and a choice of frame modified by $\Theta^{\star} \in \R^{n}$ such that
\begin{equation}
\opnormbig{\mathcal{L}^{\pi^{\star}}(\Theta^{\star})} < c \log(r).
\end{equation}
\label{prop:variance_term_bounded}
\end{proposition}

Proposition~\ref{prop:variance_term_bounded} contains a dimension-dependent term $\log(r)$, which is small enough for our purposes and crucially holds for $T$---that is, for $n$---large enough.
To show Proposition~\ref{prop:variance_term_bounded} we consider the random matrix $\mathcal{L}^{\pi}(\Theta)$ with a random permutation $\pi \sim \mathrm{Unif}(S_m)$, and a random phase $\exp(\ci \Theta_i) \sim \mathrm{Unif}(S^{1})$ for all $i \in [n]$.
The method of moments is used to estimate the operator norm of $\mathcal{L}^{\pi}(\Theta)$, and is a well-known method in random matrix theory \cite{tao2012topics}. 
This typically involves using that for any positive-semidefinite Hermitian matrix $A$, $\opnorm{A}^{k} \leq \tr{A^{k}}$ for any $k \geq 1$.
Therefore, we can compute the expectation of the operator norm by estimating the expected trace
\begin{equation}
\expectationbig{\opnorm{\mathcal{L}^{\pi}(\Theta)}^{k}} \leq  \expectationbig{\mathrm{Tr}[\mathcal{L}^{\pi}(\Theta)^{k}]} = F(k).
\label{eqn:upper_bound_expectation_operator_norm}
\end{equation}
By showing that $F(k)$ is small enough, from \eqref{eqn:upper_bound_expectation_operator_norm} we deduce that for some permutation $\pi \in S_n$, and frame modifier $\Theta$ the operator norm is also small enough.
Differently to common results in random matrix theory, the entries of $\mathcal{L}^{\pi}(\Theta)$ are not independent, but strongly correlated depending on the frame.
Nonetheless, we are able to obtain a bound of the operator norm by leveraging the frame properties.
The following bound characterizes the expectation \eqref{eqn:upper_bound_expectation_operator_norm}.
\begin{lemma}
Under the assumptions of Proposition~\ref{prop:variance_term_bounded}, there exists $c>0$ such that for any $r \geq 2$, $T \geq 2$ and $r \geq k \geq 1$,
\begin{equation}
F(k) \leq c \left( r k^{k} +  \frac{k! k^{2k} r^{k}}{T} \right).
\label{eqn:variance_bound_term}
\end{equation}
\label{lem:variance_bound_moment_bound}
\end{lemma}
\noindent
Let us show first that Lemma~\ref{lem:variance_bound_moment_bound} implies Proposition~\ref{prop:variance_term_bounded}.\\

\noindent
\emph{Proof (of Proposition~\ref{prop:variance_term_bounded})}. Fix $k = \floor{\log(r)}$, and choose $T=T(r) \geq 1$ large enough such that
\begin{equation}
F(k) \leq 2 c r k^{k}.
\end{equation}
Then, there exists a constants $c_1, C>0$ such that for all $r \geq 2$,
\begin{align}
\expectationbig{\opnorm{\mathcal{L}^{\pi}(\Theta)}} & \leq \expectationbig{\opnorm{\mathcal{L}^{\pi}(\Theta)}^{k}}^{1/k} \nonumber \\
& \leq F(k)^{1/k} \nonumber \\
& \leq c_1^{1/\log(r)} r^{1/\log(r)} \log(r) \nonumber \\
& \leq  C \log(r).
\label{eqn:proof_variance_bound_prop5}
\end{align}
From \eqref{eqn:proof_variance_bound_prop5}, we deduce that there exists a permutation $\pi^{\star}$ and frame modifier $\Theta^{\star}$ such that
\begin{equation}
\opnorm{\mathcal{L}^{\pi^{\star}}(\Theta^{\star})} \leq C \log(r).
\label{eqn:moments_bound}
\end{equation}
\qed

The proof of Lemma~\ref{lem:variance_bound_moment_bound} is involved, and is deferred to Appendix~\ref{secappendix:proof_of_lemma}.
The intuition behind the proof is that the correlations that the different vectors $\ket{v_i}$ of the frame may have---encoded in the inner products $\langle v_i | v_j \rangle$---either are small when $T$ is large, or they are concentrated in a handful of vectors.
In the former case, for a random permutation and frame modifier, the matrix $\mathcal{L}^{\pi}(\Theta)$ will be approximately homogeneous and thus have operator norm of order $O(1)$.
In the latter case, most permutations will `separate' the significant vectors well so that we will only need to worry about the diagonal of $\mathcal{L}^{\pi}(\Theta)$. 
We can then estimate the expectation of the operator norm when significant vectors add up in the diagonal entries.
This is quantified by the leading term in \eqref{eqn:variance_bound_term}.
It also turns out that when there are only significant vectors (when $L_1 = \id{\mathcal{S}}$, for example), the order of \eqref{eqn:moments_bound} is almost sharp.
Specifically, the expected operator norm in this case is the expected maximum of a multinomial random variable with uniform probability $(1/r, \ldots, 1/r)$ and with $r$ trials, which has order $\log(r)/\log(\log(r))$ \cite{raab1998balls}.

In \eqref{eqn:proof_variance_bound_prop5} the choice of $k \simeq \log(r)$ is optimal, since for larger or smaller order of $k$, we would obtain a worse bound on $F(k)$.

\subsection{Randomization}
\label{sec:randomization}
We assume for the time being that $T$ is large enough so that by Proposition~\ref{prop:variance_term_bounded} we can choose the equipartition of the frame induced by $\mathcal{P}$ in Lemma~\ref{lem:reduction_sum_PMs} that satisfies both \eqref{eqn:reduction_parition_of_unity} and
\begin{equation}
\opnormbig{L_1^{\dagger} L_1 + \ldots + L_T^{\dagger} L_T} \leq c \log(r).
\end{equation}
We namely omit in the notation the selected permutation $\pi^{\star}$ and phase $\Theta^{\star}$ that Proposition~\ref{prop:variance_term_bounded} guarantees, that is, we will denote $L_t = L^{\pi^{\star}}_{t}(\Theta^{\star})$ for all $t \in [T]$.
Later on in Section~\ref{sec:final_steps}, we will show that we by enlarging the Hilbert space $\mathcal{H}$, the assumption that $T$ is large enough will become superfluous.

Let $g_t$ for $t \in [T]$ be i.i.d.\ $\mathrm{N}(0, 1)$ random variables.
Define the operator-valued random variable
\begin{equation}
\hat{L} = \sum_{t=1}^{T} g_t L_t.
\label{eqn:definition_L_hat}
\end{equation}

By using this random matrix as a weight, we can decouple the dimension $T$ from the sum of $T$ weighted measurments in Lemma~\ref{lem:reduction_sum_PMs}, and transfer its dependence to the variance of $\hat{L}$.
\begin{lemma}
In the same setting of Lemma~\ref{lem:reduction_sum_PMs}, let $\hat{L}$ be defined in \eqref{eqn:definition_L_hat}, and let $\mathcal{Q} \in P(\mathcal{S})$, then
\begin{equation}
\sum_{t=1}^{T} \trbig{\rho \mathcal{Q}_{L_t}[\tau]} \leq 3\expectationBig{\trbig{\rho \mathcal{\mathcal{Q}}_{\hat{L}}[\tau]}}.
\label{eqn:lower_bound_expectation}
\end{equation}
\label{lem:upper_bound_expectation}
\end{lemma}
Before showing Lemma~\ref{lem:upper_bound_expectation} at the end of this section, we require some intermediate results first.
First, we find a proxy that in expectation bounds the sum of the weighted \gls{PM} on the left side of \eqref{eqn:lower_bound_expectation}.
Let $\hat{B}$ be an independent copy of $\hat{L}$, and $\mathcal{Q} \in P(\mathcal{S})$.
We denote by $\mathcal{Q}_{\hat{L}, \hat{B}}$ the random biweighted measurement that satisfies
\begin{align}
\trbig{\mathcal{Q}_{\hat{L}, \hat{B}}[\rho, \tau]}  &= \sum_{i=1}^{r} \trBig{\rho \hat{L} Q_i \hat{L}^{\dagger}}\trBig{\tau \hat{B} Q_i \hat{B}^{\dagger}} \nonumber \\
&= \sum_{i=1}^{r} \trBig{\rho \hat{L} Q_i \hat{L}^{\dagger} \otimes \tau \hat{B} Q_i \hat{B}^{\dagger}} \nonumber \\
&= \sum_{i=1}^{r} \trBig{(\rho \otimes \tau) (\hat{L} \otimes \hat{B})  (Q_i \otimes Q_i) (\hat{L}^{\dagger} \otimes \hat{B}^{\dagger})}.
\label{eqn:candidate_variable}
\end{align}
Note that we also have $\trbig{\mathcal{Q}_{\hat{L}, \hat{B}}[\rho, \tau]}  \geq 0$. 
\begin{lemma}
Let $\rho, \tau \in D(\mathcal{S})$, $\mathcal{Q} \in P(\mathcal{S})$, and $\{L_t \}_{t \in [T]}$ be weight matrices.
Let $\hat{L}$, $\hat{B}$ be independent copies of \eqref{eqn:definition_L_hat}. Then,
\begin{equation}
\sum_{t=1}^{T} \trbig{\rho \mathcal{Q}_{L_t}[\tau]} \leq \expectationBig{ \trbig{\mathcal{Q}_{\hat{L}, \hat{B}}[\rho, \tau]} }
\end{equation}
\label{lem:lower_bound_expectation_by_objective_function}
\end{lemma}
\begin{proof}
From the fact that $\hat{L}$ and $\hat{B}$ are independent, we readily obtain from \eqref{eqn:candidate_variable} that 
\begin{align}
\expectationBig{ \trbig{\mathcal{Q}_{\hat{L}, \hat{B}}[\rho, \tau]} } &=
\sum_{i=1}^{r} \expectationBig{\trBig{(\rho \otimes \tau) (\hat{L} \otimes \hat{B})  (Q_i \otimes Q_i) (\hat{L}^{\dagger} \otimes \hat{B}^{\dagger})}} \nonumber \\
&= \sum_{i=1}^{r} \sum_{t,z = 1}^{T} \trBig{(\rho \otimes \tau) (L_t \otimes L_z)  (Q_i \otimes Q_i) (L_t^{\dagger} \otimes L_z^{\dagger})} \nonumber \\
&= \sum_{t,z=1}^{T} \sum_{l=1}^{r} \trBig{(\rho \otimes \tau) (L_t \otimes L_z)  (Q_i \otimes Q_i) (L_t^{\dagger} \otimes L_z^{\dagger})} \nonumber \\
& \geq \sum_{t=1}^{T}  \sum_{l=1}^{r} \trBig{(\rho \otimes \tau) (L_t \otimes L_t)  (Q_i \otimes Q_i) (L_t^{\dagger} \otimes L_t^{\dagger})} \tag{all terms positive} \nonumber \\
& = \sum_{t=1}^{T}  \sum_{l=1}^{r} \mathrm{Tr}[\rho L_t  Q_i L_t^{\dagger} \tau L_t  Q_i L_t^{\dagger}] \tag{$\{ Q_i \}_{i=1}^{r}$ are rank one projectors} \nonumber \\
& = \sum_{t=1}^{T} \trbig{\rho \mathcal{Q}_{L_t}[\tau]}.
\end{align}
\end{proof}

The biweighted measurement $\mathcal{Q}_{\hat{L}, \hat{B}}$ in \eqref{eqn:candidate_variable} does not have the single weighted measurement structure that we need.
To overcome this problem we use the following two results based on trace convexity to upper bound $\trbig{\mathcal{Q}_{\hat{L}, \hat{B}}[\rho, \tau]} $ by several weighted measurements.
\begin{lemma}\emph{(Adapted from \cite[Corollary 1.1]{lieb1973convex})}.
Let $A, B \in \nnbsa{\mathcal{S}}$ and let $G_1, G_2 \in B(\mathcal{S})$. For any $\lambda \in [0,1]$, let 
\begin{equation}
G = \lambda G_1 + (1- \lambda) G_2.
\end{equation} 
The following inequality holds
\begin{equation}
\trbig{ A G B  G^{\dagger} } \leq \lambda \trbig{ A G_1 B G_1^{\dagger}} + (1-\lambda) \trbig{ A G_2 B G_2^{\dagger}}.
\end{equation}
\label{lemma:trace_convexity}
\end{lemma}

For the following result we require some additional notation.
For any $J \in B(\mathcal{S})$, and $M \in \bsa{\mathcal{S} \otimes \mathcal{S}}$, we let
\begin{equation}
\mathsf{Ad}_{J}[M] = (J \otimes J) \ M\ (J \otimes J)^{\dagger}.
\label{eqn:def_adjoint}
\end{equation}
We define the set of separable operators of $\nnbsa{\mathcal{S} \otimes \mathcal{S}}$ by 
\begin{equation}
\mathrm{Sep}\bigl( \nnbsa{\mathcal{S} \otimes \mathcal{S}} \bigr)  = \Bigl\{ \sum_{i=1}^{a} \lambda_i A_1^{i} \otimes A_2^{i} \Big| a \in \N, A_1^{i}, A_2^{i} \in \nnbsa{\mathcal{S}}, \text{ and } \lambda_i \geq 0 \text{ for all } i\in [a] \Bigr\}.
\label{eqn:def_separable_operators}
\end{equation}

\begin{lemma}
Let $C, D \in \mathrm{Sep}(\nnbsa{\mathcal{S} \otimes \mathcal{S}})$ and $X, Y \in B(\mathcal{S})$, then in the notation of \eqref{eqn:def_adjoint}, 
\begin{align}
\tr{C &(X \otimes Y)  D (X \otimes Y)^{\dagger}} + \tr{C (Y \otimes X) D (Y \otimes X)^{\dagger}}  \leq  \nonumber \\
\frac{1}{2}\Bigl( &\trbig{C  \mathsf{Ad}_{X + Y}[D]} + \trbig{C \mathsf{Ad}_{X - Y}[D]}\Bigr) + \trbig{C  \mathsf{Ad}_{X}[D]} + \trbig{C  \mathsf{Ad}_{Y}[D]} .
\label{eqn:trace_quadrature}
\end{align}
\label{lem:trace_symmetrization_quadrature}
\end{lemma}
\begin{proof}
We will first change variables. 
Define
\begin{align}
\mathcal{M}_+ &= X \otimes Y +  Y \otimes X.
\end{align}
We can check that expanding the terms of $\mathcal{M}_+$ gives the following identity
\begin{align}
\trbig{C &(X \otimes Y)  D (X \otimes Y)^{\dagger}} + \trbig{C (Y \otimes X) D (Y \otimes X)^{\dagger}} = \trbig{C \mathcal{M}_+ D \mathcal{M}_+^{\dagger}} \nonumber \\
 & \phantom{aaaaaaaaaaaaaaaaaaaaaaa}  -\trbig{C(X \otimes Y) D (Y \otimes X)^{\dagger}} -  \trbig{C(Y \otimes X) D (X \otimes Y)^{\dagger}}\nonumber\\
& \leq \trbig{C \mathcal{M}_+ D \mathcal{M}_+^{\dagger}} + \big|\trbig{C(X \otimes Y) D (Y \otimes X)^{\dagger}}\big| + \big|\trbig{C(Y \otimes X) D (X \otimes Y)^{\dagger}}\big|.
\label{eqn:antisymm1}
\end{align}
We examine the term involving $\mathcal{M}_+$ in \eqref{eqn:antisymm1} closely.
Note the identity
\begin{align}
\mathcal{M}_+ &= \frac{1}{2} \Bigl( (X + Y) \otimes  (X + Y) -  (X - Y) \otimes  (X - Y) \Bigr) \nonumber  \\
&= \frac{1}{2} \Bigl( (X + Y) \otimes  (X + Y) +  \ci (X - Y) \otimes  \ci (X - Y)],
\label{eqn:antisymm_+}
\end{align}
We can use the trace convexity of Lemma~\ref{lemma:trace_convexity} with the indentity \eqref{eqn:antisymm_+} of $\mathcal{M}_+$. 
Trace convexity yields the following inequality
\begin{equation}
\trbig{C \mathcal{M}_+ D \mathcal{M}_+^{\dagger}} \leq \frac{1}{2} \Bigl( \trbig{C  \mathsf{Ad}_{X + Y}[D]} + \trbig{C  \mathsf{Ad}_{X - Y}[D]} \Bigr).
\label{eqn:antisymm2}
\end{equation}

We bound the remaining terms of \eqref{eqn:antisymm1}.
We assume for the time being that $C$, and $D$ are elementary tensors, that is, $C=C_1 \otimes C_2$, and $D = D_1 \otimes D_2$, where $C_1, C_2, D_1, D_2 \in \nnbsa{\mathcal{S}}$.
We will use (i) Cauchy-Schwartz inequality in $\tr{C_1 X D_1 Y^{\dagger}} = \tr{(C_1^{1/2} X D_1^{1/2})(C_1^{1/2} Y D_1^{1/2})^{\dagger}}$, and similarly with $\tr{C_2 Y D_2 X^{\dagger}}$.
The inequality reads
\begin{align}
\big| \trbig{C(X \otimes Y) D & (Y \otimes X)^{\dagger}} \big|  = \Big| \trbig{C_1 X D_1 Y^{\dagger}} \trbig{C_2 Y D_2 X^{\dagger}} \Big| \nonumber \\
& \eqcom{i} \leq \trbig{C_1 X D_1 X^{\dagger}}^{\frac{1}{2}} \trbig{C_1 Y D_1 Y^{\dagger}}^{\frac{1}{2}} \trbig{C_2 X D_2 X^{\dagger}}^{\frac{1}{2}} \trbig{C_2 Y D_2 Y^{\dagger}}^{\frac{1}{2}} \nonumber \\
& = \trbig{C(X \otimes X) D (X \otimes X)^{\dagger}}^{\frac{1}{2}} \trbig{C(Y \otimes Y) D (Y \otimes Y)^{\dagger}}^{\frac{1}{2}} \nonumber\\
& = \trbig{C  \mathsf{Ad}_{X}[D]}^{\frac{1}{2}} \trbig{C  \mathsf{Ad}_{Y}[D]}^{\frac{1}{2}}.
\label{eqn:antisymm3}
\end{align}
We similarly obtain \eqref{eqn:antisymm3} for the remaining term in \eqref{eqn:antisymm1}.
In particular, we have
\begin{align}
\abs{\trbig{C(X \otimes Y) D (Y \otimes X)^{\dagger}}} + \big| \trbig{C(Y \otimes X) & D (X \otimes Y)^{\dagger}} \big|  \leq 2 \trbig{C  \mathsf{Ad}_{X}[D]}^{\frac{1}{2}} \trbig{C  \mathsf{Ad}_{Y}[D]}^{\frac{1}{2}} \nonumber \\
& \leq \trbig{C \mathsf{Ad}_{X}[D]} + \trbig{C  \mathsf{Ad}_{Y}[D]}.
\label{eqn:antisymm4}
\end{align}
Finally, we can substitute the bounds of \eqref{eqn:antisymm2} and \eqref{eqn:antisymm4} into \eqref{eqn:antisymm1} to show the result in the case that $C$, and $D$ are elementary tensors.
To extend the result to all $C, D \in \mathrm{Sep}(\nnbsa{\mathcal{S} \otimes \mathcal{S}})$, note that if $C$ is separable, we can decompose $C = \sum_{i=1}^{a} \lambda_i C_1^{i} \otimes C_2^{i}$ with $C_j^{i} \in \nnbsa{\mathcal{S}}$ for $j \in [2]$, and $\lambda_i \geq 0$ for $i \in [a]$.
We can do a similar decomposition with $D$.
Now, from the bilinearity of both sides of the inequality \eqref{eqn:trace_quadrature} with respect to $C$, and $D$, we can extend \eqref{eqn:trace_quadrature} to all positive linear combinations of tensors in $\nnbsa{\mathcal{S}} \otimes \nnbsa{\mathcal{S}}$.
\end{proof}

We are now in position to show Lemma~\ref{lem:upper_bound_expectation}.\\
\noindent
\emph{Proof (of Lemma~\ref{lem:upper_bound_expectation}).}
From Lemma~\ref{lem:lower_bound_expectation_by_objective_function}, we have the inequality
\begin{equation}
\sum_{t=1}^{T} \trbig{\rho \mathcal{Q}_{L_t}[\tau]} \leq \expectationBig{ \trbig{\mathcal{Q}_{\hat{L}, \hat{B}}[\rho, \tau]} }.
\label{eqn:lower_bound_1}
\end{equation}
We symmetrize the expectation in \eqref{eqn:lower_bound_1} noting that since $\hat{L}$ and $\hat{B}$ are identically and independently distributed we have
\begin{align}
\expectationBig{ \trbig{\mathcal{Q}_{\hat{L}, \hat{B}}[\rho, \tau]} } &= \frac{1}{2} \expectationBig{ \trbig{\mathcal{Q}_{\hat{L}, \hat{B}}[\rho, \tau]} } + \frac{1}{2} \expectationBig{ \trbig{\mathcal{Q}_{\hat{B},\hat{L}}[\rho, \tau]} } \nonumber \\
&= \frac{1}{2} \sum_{i=1}^{r} \expectationBig{\trBig{(\rho \otimes \tau) (\hat{L} \otimes \hat{B})  (Q_i \otimes Q_i) (\hat{L}^{\dagger} \otimes \hat{B}^{\dagger})} \nonumber \\
&\phantom{aaaaaaaaaaaaaaaaaaaa}+ \trBig{(\rho \otimes \tau) (\hat{B} \otimes \hat{L})  (Q_i \otimes Q_i) (\hat{B}^{\dagger} \otimes \hat{L}^{\dagger})}}
\label{eqn:lower_bound_2}
\end{align}
We use now Lemma~\ref{lem:trace_symmetrization_quadrature} for each pair of summands in the expectation of \eqref{eqn:lower_bound_2}.
We obtain
\begin{align}
\eqref{eqn:lower_bound_2} &\leq  \sum_{i=1}^{r} \frac{1}{4} \mathbb{E}\Bigl( \trbig{(\rho \otimes \tau)  \mathsf{Ad}_{\hat{L} + \hat{B}}[Q_i \otimes Q_i]} + \trbig{(\rho \otimes \tau)  \mathsf{Ad}_{\hat{L} - \hat{B}}[Q_i \otimes Q_i]}\Bigr) \nonumber \\
&\phantom{aaaaaaaaaaaaaaaaa} + \frac{1}{2} \mathbb{E}\Bigl( \trbig{(\rho \otimes \tau) \mathsf{Ad}_{\hat{L}}[Q_i \otimes Q_i]}\Bigr) + \frac{1}{2}\mathbb{E}\Bigl( \trbig{(\rho \otimes \tau)  \mathsf{Ad}_{\hat{B}}[Q_i \otimes Q_i]} \Bigr) \nonumber \\
&=  \frac{1}{4} \mathbb{E}\Bigl( \trbig{\rho \mathcal{Q}_{\hat{L} + \hat{B}}[\tau]} + \trbig{\rho \mathcal{Q}_{\hat{L} - \hat{B}}[\tau]} \Bigr) + \mathbb{E}\Bigl(\trbig{\rho \mathcal{Q}_{\hat{L}}[\tau]}\Bigr). \tag{$\hat{B}$ distributed like $\hat{L}$}
\end{align}
From independence and since $\hat{L}$, and $\hat{B}$ have identical symmetric Gaussian distributions, $\hat{L} + \hat{B}$ possesses the same distribution as $\hat{L} - \hat{B}$.
Moreover, from its definition in \eqref{eqn:definition_L_hat}, $\hat{L} + \hat{B}$ is a sum of i.i.d.\ centered Gaussian random matrices with double the variance than $\hat{L}$, that is, $\hat{L} + \hat{B}$ will have the same distribution as $\sqrt{2} \hat{L}$.
Using these identities we can write
\begin{align}
\frac{1}{4} \mathbb{E}\Bigl( \trbig{\rho \mathcal{Q}_{\hat{L} + \hat{B}}[\tau]} + \trbig{\rho \mathcal{Q}_{\hat{L} - \hat{B}}[\tau]} \Bigr) &= \frac{1}{2} \mathbb{E}\Bigl(\trbig{\rho \mathcal{Q}_{\sqrt{2} \hat{L}}[\tau]}\Bigr) \nonumber \\
&= 2 \mathbb{E}\Bigl(\trbig{\rho \mathcal{Q}_{\hat{L}}[\tau]}\Bigr)
\end{align}
Together with the previous inequality, this last bound completes the proof.
\qed\\

\par
In \eqref{eqn:lower_bound_expectation}, we have randomized the weighted measurements with $\hat{L}$, and $\hat{B}$, but we are still left with weighted measurements that cannot be directly compared with \glspl{PM}.
Later in Section~\ref{sec:interpolation_bound}, we show that traces of weighted measurements are upper bounded by traces of \gls{PM} up to a factor depending on the weights.

\subsection{Concentration}
\label{sec:concentration}

In the previous section, we have decoupled the dimension $n$ in the number of summands in the decomposition by weighted \gls{PM} in \eqref{eqn:decomposition_weighted_PM}.
However, the dependence may still be present in the form of large variance of $\hat{L}$.
A key value that characterizes how large $\hat{L}$ can be is its operator norm $\opnorm{\hat{L}}$.
The order of an operator norm of a random matrix has been thoroughly studied for independent series of Gaussian random matrices \cite{tropp2015introduction, tao2012topics}, which is our case.

We will use the following result that relates the magnitude of the operator norm to that of its expected variance.
\begin{lemma}\emph{(Adapted from \cite[Theorem 4.1.1]{tropp2015introduction})}\\
Let $B = \{B_1, \ldots, B_T \}$ be complex matrices of dimension $r$ and define the variance statistic by 
\begin{equation}
\nu(B) = \max \left( \pnormBig{ \sum_{t=1}^{T} B_t B_t^{\dagger}}{\mathrm{op}}, \pnormBig{ \sum_{t=1}^{T} B_t^{\dagger} B_t}{\mathrm{op}} \right).
\end{equation}
Let $g_{1}, \ldots, g_{T}$ be i.i.d.\ real $\mathrm{N}(0,1)$ random variables. 
Then for any $l > 0$, we have
\begin{equation}
\probabilityBig{ \pnormBig{\sum_{t=1}^{T} g_t B_t}{\mathrm{op}} > l } \leq 2r \exp \Bigl( -\frac{l^2}{\nu(B)} \Bigr).
\end{equation}
\label{lem:concentration_matrix_gaussian}
\end{lemma}

Recall that we assume for the time being that $T$ is large enough so that $\mathcal{L} = L_1^{\dagger} L_1 + \cdots  + L_T^{\dagger} L_T$ satisfies Proposition~\ref{prop:variance_term_bounded}.
Together with \eqref{eqn:reduction_parition_of_unity}, the variance statistic satisfies for some $c>0$
\begin{align}
\nu(\{L_t\}_{t=1}^{T}) &= \max \Bigl(\opnorm{\id{\mathcal{S}}}, \opnorm{\mathcal{L}} \Bigr) \nonumber \\
&= \max \bigl(1, c \log(r) \bigr).
\label{eqn:variance_statistic}
\end{align}
Using Lemma~\ref{lem:concentration_matrix_gaussian}, we directly obtain a bound on the operator norm of $\hat{L}$.
\begin{lemma}
Suppose $r \geq 2$ and let $\hat{L}$ be defined in \eqref{eqn:definition_L_hat} with weight matrices $\{L_t\}_{t=1}^{T}$ that satisfy Proposition~\ref{prop:variance_term_bounded}. 
There exists $c>0$ such that if $l \geq 0$, with probability at least $1-2r \exp (-l^2/(c \log(r))$,
\begin{equation}
\opnorm{\hat{L}} \leq l.
\end{equation}
\label{lemma:probabilistic_bound_norm_L}
\end{lemma}

The following moment estimate will be used later on and is key in determining $K_{n,r}$.
\begin{corollary}
In the same setting as Lemma~\ref{lemma:probabilistic_bound_norm_L}, if the weight matrices $\{L_t\}_{t=1}^{T}$ satisfy Proposition~\ref{prop:variance_term_bounded} there exists $c>0$ such that
\begin{equation}
\expectationbig{\opnorm{\hat{L}}^{4}} \leq c \log(r)^{4}.
\end{equation}
\label{cor:moment_bound_L_hat}
\end{corollary}
\begin{proof}
By Fubini's theorem for the expectation of a positive random variable $X$, we have the identity
\begin{equation}
\expectation{X^{4}} = \int_{0}^{\infty} \probability{X > s} 4s^{3} \mathrm{d}s.
\label{eqn:integration_formula}
\end{equation}
From Lemma~\ref{lemma:probabilistic_bound_norm_L}, for $r \geq 2$ we also have the bound
\begin{equation}
\probabilitybig{ \opnorm{\hat{L}} > s } \leq \min \Bigl(1, 2r \exp \Bigl( -\frac{s^2}{c \log(r)} \Bigr)\Bigr).
\label{eqn:tail_bound_formula}
\end{equation}
The value $1$ is attained by the second factor of \eqref{eqn:tail_bound_formula} at $l^{\star}$, with $c_1 \log(r) \leq l^{\star} \leq c_2 \log(r)$ for some $c_1, c_2>0$.
Hence, using \eqref{eqn:integration_formula} with the bound \eqref{eqn:tail_bound_formula} yields
\begin{align}
\expectationbig{\opnorm{\hat{L}}^{4}} &\leq \int_{0}^{l^{\star}} 4s^3 \mathrm{d}s + \int_{l^{\star}}^{\infty} 8r \exp \Bigl( -\frac{s^2}{c \log(r)} \Bigr) s^3\ \mathrm{d} s \nonumber \\
&\leq (c_2)^{4} \log(r)^{4} + 8r \int_{c_1 \log(r)}^{\infty}  \exp \Bigl( -\frac{s^2}{c \log(r)} \Bigr) s^3\ \mathrm{d} s.
\label{eqn:tail_bound_to_expectation_1}
\end{align}
We can compute the last term in  \eqref{eqn:tail_bound_to_expectation_1} by using the change of variables $ s^2 =c \log(r) y$,
\begin{align}
\int_{c_1 \log(r)}^{\infty}  \exp \Bigl( -\frac{s^2}{c \log(r)} \Bigr) s^3 \mathrm{d} s &= c_4 \log(r)^2 \int_{c_3 \log(r)}^{\infty}  \exp (-y) y\ \mathrm{d} y \nonumber \\
&= c_4 \log(r)^2 \Bigl( -\exp(-y)y - \exp(-y)\Bigr) \Big|_{c_3 \log(r)}^{\infty} \nonumber \\
&\leq c_5 \frac{\log(r)^3}{r},
\label{eqn:tail_bound_to_expectation_2}
\end{align}
for some constants $c_3, c_4, c_5 > 0$ independent of $r \geq 2$.
Finally, substitute \eqref{eqn:tail_bound_to_expectation_2} in \eqref{eqn:tail_bound_to_expectation_1}.
\end{proof}

\begin{remark}
We note that the bound in Proposition~\ref{prop:variance_term_bounded} is sufficient to obtain a polylogarithmic bound on the expected value of $\opnorm{\hat{L}}$, which will eventually determine the order of $K_{n,r}$. 
However, if the logarithmic bound of Proposition~\ref{prop:variance_term_bounded} can be improved to $\opnorm{\hat{L}} = O(1)$, we still have the intrinsic dimensional dependence on $r$ in the the concentration inequality for random matrices; the $r$ factor in Lemma~\ref{lem:concentration_matrix_gaussian}.
The polylogarithmic terms in the final bound of $K_{n,r}$ are thus unavoidable unless there is an additional low-rank structure of the matrices $\{L_t\}_{t \in [T]}$ that can be exploited.
This fact showcases the limits of this approach, as we are not biasing the probability space towards the optimum of \eqref{eqn:maximization_restricted}, e.g., by considering only Parseval frames that have already a certain low-rank structure.
\label{remark:dimensional_dependence}
\end{remark}

\subsection{Interpolation bound}
\label{sec:interpolation_bound}

In Lemma~\ref{lem:upper_bound_expectation} of Section~\ref{sec:concentration}, we have upper bounded the sum of $T$ weighted measurements by an expectation over a single random weighted measurement $\mathcal{Q}_{\hat{L}}$ that depends on the Gaussian random matrix $\hat{L}$.
In this section, we will `extract' the weight $\hat{L}$ of the weighted measurement.
Specifically, we will show the following inequality for weighted measurements; recall Definition~\ref{def:contracted_PM}.

\begin{proposition}
Let $\rho, \tau \in D(\mathcal{S})$ be density matrices, $L \in B(\mathcal{S})$, and $\mathcal{Q} \in P(\mathcal{S})$ be a \gls{PM}. 
There exists $\mathcal{Z} \in P(\mathcal{S})$ such that
\begin{equation}
\trbig{\rho \mathcal{Q}_{L}[\tau]} \leq \opnorm{L}^4 \trbig{\rho \mathcal{Z}[\tau]}
\end{equation}
\label{lem:inequality_bound_norm_fourth}
\end{proposition}
Proposition~\ref{lem:inequality_bound_norm_fourth} shows that we can upper bound traces of weighted measurements by traces of \glspl{PM}, and at most a multiplicative factor is added depending only on the weight matrices.

To show Proposition~\ref{lem:inequality_bound_norm_fourth}, we require two intermediate results.
The first lemma shows an integral representation of the conjugation by an operator that has norm less than one. 
\begin{lemma}
Let $K \in \bsa{\mathcal{S}}$ be such that $\opnorm{K} \leq 1$ and let $A \in B(\mathcal{S})$. 
There exists a probability measure $\mu$ over $\R$ and a unitary $D \in U(\mathcal{S})$ such that if $|K| = \sqrt{K K^{\dagger}}$,
\begin{equation}
KAK = \int_{\R} |K|^{-\ci t} DAD |K|^{-\ci t} \mu(dt).
\label{eqn:integral_representation}
\end{equation}
Moreover, if $K \in \nnbsa{\mathcal{S}}$, then $D = \id{\mathcal{S}}$ .
\label{lemma:integral_representation_conjugation}
\end{lemma}
\begin{proof}
We assume first that $K \in  \pbsa{\mathcal{S}}$ so that $K = \sum_{i=1}^r \lambda_i K_i = |K|$, where $1 \geq \lambda_i > 0$ and $K_i$ are rank-one projectors.
Ket $K^{-\ci t} = \sum_{i=1}^r \e{-\ci x_i t} K_i$, where $x_i = \log(\lambda_i) \leq 0$. 
We expand the terms $K$ and $K^{-\ci t}$ in both sides of \eqref{eqn:integral_representation} in terms of the projectors $K_i$ and examine at the term in the expansion corresponding to $K_iAK_j$. 
For \eqref{eqn:integral_representation} to hold, we must namely have the following equality for all $i,j \in [r]$
\begin{equation}
\lambda_i \lambda_j = \int_{\R} \e{-\ci (x_i + x_j) t} \mu(dt).
\label{eqn:integral_rep_eq1}
\end{equation} 
Observe that the right-hand side of \eqref{eqn:integral_rep_eq1} is the Fourier transform $\hat{\mu}$ of $\mu$ at $x_i + x_j$.
The following equality for all $i, j \in [r]$ must then hold.
\begin{equation}
\lambda_i \lambda_j = \hat{\mu}(x_j + x_i).
\label{eqn:fourier_transform_condition}
\end{equation}

Since $x_i \leq 0$ for all $i \in [r]$, a function $\hat{\mu}$ satisfying \eqref{eqn:fourier_transform_condition} is $\hat{\mu}(x) = \exp(-|x|)$. 
By applying the inverse Fourier transform to the Ansatz for $\hat{\mu}(x)$ yields
\begin{equation}
\mu(t) = \frac{1}{2\pi}\int_{\R} \e{\ci t x} \hat{\mu}(x) dx = \frac{1}{2\pi}\int_{\R} \e{\ci t x} \e{-|x|} dx = \frac{1}{\pi(1+t^2)}.
\end{equation}
Moreover, by integrating directly we have $\int_{\R} \mu(t)dt = 1$, so that $\mu$ is also a distribution over $\R$. 
This shows \eqref{eqn:integral_representation} for the case $K \in \pbsa{\mathcal{S}}$. 

If $K$ is full rank but possesses negative eigenvalues, we can consider the polar decomposition $K = |K| D$, where $|K| = \sqrt{K K^{\dagger}} \in \pbsa{\mathcal{S}}$ and $D \in U(\mathcal{S})$ is such that $|K| D = K$. 
The unitary $D$ will encode the signs of the eigenvalues of $K$, and satisfies $D = D^{\dagger}$. 
Apply the same arguments as in the case $K \in \pbsa{\mathcal{S}}$, now using $|K|$, and $DAD$ instead. 
After doing so, we also obtain
\begin{equation}
KAK = \int_{\R} |K|^{-\ci t} DAD |K|^{-\ci t} \mu(dt).
\end{equation}

Finally, if $K$ has null eigenvalues, there exists a sequence of operators $Z_1, \ldots, Z_j, \ldots \in \bsa{\mathcal{S}}$ such that $Z_j \in \bsa{\mathcal{S}}$, $\opnorm{Z_j} \leq 1$, $\mathrm{ker}(Z_j) = \emptyset$, and $\opnorm{Z_j  - K} \to 0$ as $j \to \infty$. 
From the fact that $\opnorm{|Z_j|^{-\ci t} DAD |Z_j|^{-\ci t}} \allowbreak \leq \opnorm{A}$ for any $j \in \N$, the integral in \eqref{eqn:integral_representation} is uniformly bounded. 
By the dominated convergence theorem, the limit as $j \to \infty$ exists and will be equal to $KAK$.
\end{proof}

The second lemma we need is a commonly used interpolation result from complex analysis:

\begin{lemma}\emph{(Hadamard Three Lines Theorem, \cite[p. 33]{reed1975ii})}
Let $\Psi(z)$ be a complex-valued function, bounded and continious on the closed strip $\mathcal{B} = \{z \in \C : 0 \leq \mathrm{Re}(z) \leq 1 \}$, analytic in the interior of $\mathcal{B}$ and satisfying
\begin{align}
|\Psi(z)| & \leq M_0 \quad \text{ if } \quad \mathrm{Re}(z) = 0 \\
|\Psi(z)| & \leq M_1 \quad \text{ if } \quad \mathrm{Re}(z) = 1.
\end{align}
Then, for $z \in \mathcal{B}$, $|\Psi(z)| \leq M_0^{\mathrm{Re}(z)} M_1^{1-\mathrm{Re}(z)}$.
\label{lemma:hadamard_three_lines}
\end{lemma}
We are now in position to show the proposition.\\

\noindent
\emph{Proof (of Proposition~\ref{lem:inequality_bound_norm_fourth}).}
We may assume that $L \in \nnbsa{\mathcal{H}}$ and $L$ is full rank. 
Indeed, by the polar decomposition, if $L \in B(\mathcal{S})$, there exists $U \in U(\mathcal{S})$ such that $L = |L| U$. 
We can then consider the \gls{PM} $\tilde{\mathcal{Q}} = \{U Q_1 U^{\dagger}, \ldots, U Q_r U^{\dagger} \} \in P(\mathcal{S})$ that satisfies the equality $\mathrm{Tr}[\rho \mathcal{Q}_{L}(\tau)] = \mathrm{Tr}[\rho \tilde{\mathcal{Q}}_{|L|}(\tau)]$.
Additionally, if $L$ has null eigenvalues, we can use a sequence $Z_1, \ldots, Z_j, \dots \subset B(\mathcal{S})$ with $Z_j$ of full rank for all $j \in \N$ such that $Z_j \to L$ in operator norm in a similar manner as used in Lemma~\ref{lemma:integral_representation_conjugation} to prove the claim.
We assume in the rest of the proof that $L \in \pbsa{\mathcal{H}}$.

We set $\tilde{L} = L/\opnorm{L}$, so that $\opnorm{\tilde{L}} \leq 1$. 
Therefore,
\begin{equation}
\trbig{\rho \mathcal{Q}_{L}[\tau]} = \opnorm{L}^4 \trbig{\rho \mathcal{Q}_{\tilde{L}}[\tau]}.
\label{eqn:reduction_1}
\end{equation}
We focus on the last term.
Define the complex-valued function 
\begin{equation}
\Psi(z) = \sum_{i=1}^r \trbig{\rho L^{z} Q_i \tilde{L}^{2-z} \tau \tilde{L}^{z} Q_i L^{2-z}}.
\end{equation}
Note that since $\tilde{L}$ is self-adjoint, $\Psi(1) = \trbig{\rho \mathcal{Q}_{\tilde{L}}(\tau)}$.
The conditions $\opnorm{\tilde{L}} \leq 1$ and $\tilde{L} \geq 0$ imply that if $0 \leq \mathrm{Re}(z) \leq 2$, then $\opnorm{\tilde{L}^z} \leq 1$ and $\opnorm{\tilde{L}^{2-z}} \leq 1$. 
Therefore, the function $\Psi$ is analytic, and also bounded and continuous on the strip $\mathcal{S} = \{z \in \C : 0 \leq \mathrm{Re}(z) \leq 2 \}$. 
Define
\begin{align}
M_0 = \sup_{\mathrm{Re}(z) = 0} |\Psi(z)|  = \sup_{t \in \R} |\Psi( 0 + \ci t)|, \nonumber \\
M_2 = \sup_{\mathrm{Re}(z) = 2} |\Psi(z)| = \sup_{t \in \R} |\Psi(2 + \ci t)|. \label{eqn:reduction_interpolation_boundary}
\end{align}
By Lemma~\ref{lemma:hadamard_three_lines}, the upper bound $\Psi(1) \leq \max( M_2, M_0)$ holds.

Without loss of generality, we assume that the maximum is attained at $M_2$---the case of $M_0$ is analogous.
For an expression of the supremum $M_2$ in \eqref{eqn:reduction_interpolation_boundary}, we can consider the continious map $f: \R \to U(\mathcal{S})$ given by $f(t) = \tilde{L}^{\ci t}$, and notice that if $\tilde{\Psi}(t) = \Psi( 2 + it)$, then
\begin{equation}
\tilde{\Psi}(t) = \sum_{i=1}^r \trbig{\rho \tilde{L}^{2} f(t) Q_i f(t)^{\dagger} \tau \tilde{L}^{2} f(t) Q_i f(t)^{\dagger}}.
\end{equation} 
By compactness of $U(\mathcal{S})$, $\sup_{t \in \R} |\tilde{\Psi}(t)|$ is attained at some unitary $U \in U(\mathcal{S})$ in the closure of $\mathrm{Im}(f)$, and satisfies
\begin{equation}
M_2 = \Bigl|\sum_{i=1}^r \trbig{\rho \tilde{L}^{2} U Q_i U^{\dagger} \tau \tilde{L}^2 U Q_i U^{\dagger}}\Bigr|.
\end{equation}
Let $\tilde{\mathcal{Q}} = (U Q_1 U^{\dagger}, \ldots, U Q_r U^{\dagger} ) = (\tilde{Q}_1, \ldots, \tilde{Q}_r ) \in P(\mathcal{S})$. 
Using the inequality $\Psi(1) \leq M_2$, the following holds
\begin{align}
\trbig{ \rho \mathcal{Q}_{\tilde{L}}(\tau)} &\leq \Bigl|    \sum_{i=1}^r \trbig{\rho \tilde{L}^{2} U Q_i U^{\dagger} \tau \tilde{L}^2 U Q_i U^{\dagger}}\Bigr| \nonumber \\
& = \Bigl|    \sum_{i=1}^r \trbig{\rho \tilde{L}^{2} \tilde{Q}_i \tau \tilde{L}^2 \tilde{Q}_i}\Bigr|. 
\label{eqn:reduction_interpolation}
\end{align}

We are now in position to apply \refLemma{lemma:integral_representation_conjugation} to \eqref{eqn:reduction_interpolation}.
\begin{align}
\Bigl| \sum_{i=1}^r \trbig{\rho \tilde{L}^{2} \tilde{Q}_i \tau \tilde{L}^2 \tilde{Q}_i }\Bigr| &= \Bigl|  \int_{\R} \sum_{i=1}^r \trbig{\rho \tilde{L}^{-2 \ci t} \tilde{Q}_i \tau \tilde{L}^{-2 \ci t} \tilde{Q}_i} \mu(dt) \Bigr| \tag{Lemma~\ref{lemma:integral_representation_conjugation}} \nonumber \\ 
& \leq   \int_{\R} \Bigl|\sum_{i=1}^r \trbig{\rho \tilde{L}^{-2 \ci t} \tilde{Q}_i \tau \tilde{L}^{-2 \ci t} \tilde{Q}_i} \Bigr|\mu(dt) \nonumber \\
& \leq   \sup_{t \in \R} \Bigl|\sum_{i=1}^r \trbig{ \rho \tilde{L}^{-2 \ci t} \tilde{Q}_i \tau \tilde{L}^{-2 \ci t} \tilde{Q}_i } \Bigr|. \label{eqn:reduction_norm_bound}
\end{align}
By compactness of $U(\mathcal{S})$, the supremum in \eqref{eqn:reduction_norm_bound} is attained at some $V \in U(\mathcal{S})$ in the closure of $\{ L^{- 2\ci t} | t \in \R\}$.

With $V$ at hand, we have a sequence of inequalities that use the following facts: (i) $\tilde{Q}_i$ are rank one for $i \in [r]$, (ii) the Cauchy-Schwartz inequality $|\tr{AB}| \leq \tr{AA^{\dagger}}^{1/2} \tr{BB^{\dagger}}^{1/2}$ with $A = \rho^{1/2} V \tilde{Q}_i^{1/2}$,  $B = \tilde{Q}_i^{1/2} \rho^{1/2}$ for each $i \in [r]$, and also in the case of $\tau$, respectively. Finally we also use (iii) Cauchy-Schwartz inequality for vectors.
\begin{align}
\Bigl|\sum_{i=1}^r \trbig{\rho V \tilde{Q}_i \tau V \tilde{Q}_i} \Bigr| & \leq \sum_{i=1}^r \Bigl|\trbig{ \rho V \tilde{Q}_i \tau V \tilde{Q}_i} \Bigr| \nonumber \\
&\leq \sum_{i=1}^r \Bigl| \trbig{\rho V \tilde{Q}_i} \tr{\tau V \tilde{Q}_i} \Bigr| 
\tag{i} \nonumber\\
&  \leq \sum_{i=1}^r \trbig{\rho V \tilde{Q}_i V^{\dagger}}^{\frac{1}{2}}\trbig{\rho \tilde{Q}_i}^{\frac{1}{2}} \trbig{\tau V \tilde{Q}_i V^{\dagger}}^{\frac{1}{2}}\trbig{\tau \tilde{Q}_i}^{\frac{1}{2}}  
\tag{ii} \nonumber \\
& \leq \left( \sum_{i=1}^r \trbig{\rho V \tilde{Q}_i V^{\dagger}} \trbig{\tau V \tilde{Q}_i V^{\dagger}} \right)^{\frac{1}{2}} \left( \sum_{i=1}^r \trbig{\rho \tilde{Q}_i)} \trbig{\tau \tilde{Q}_i} \right)^{\frac{1}{2}} \tag{iii} \nonumber \\
& \leq \Bigl( \trbig{\rho \mathcal{Y}[\tau]} \Bigr)^{\frac{1}{2}} \Bigl(\trbig{\rho \tilde{\mathcal{Q}}[\tau]} \Bigr)^{\frac{1}{2}} \tag{i} \nonumber \\
& \leq \max\Bigl( \trbig{\rho \mathcal{Y}(\tau)},  \trbig{\rho \tilde{\mathcal{Q}}(\tau)} \Bigr), \label{eqn:reduction_norm_finalstep}
\end{align}
where $\mathcal{Y} = ( V \tilde{Q}_1 V^{\dagger}, \ldots,  V \tilde{Q}_r V^{\dagger})\in P(\mathcal{S})$.
From \eqref{eqn:reduction_1}, combining the inequalites in \eqref{eqn:reduction_interpolation}, \eqref{eqn:reduction_norm_bound}, and \eqref{eqn:reduction_norm_finalstep} we conclude that there exists $\mathcal{Z} \in P(\mathcal{S})$ such that
\begin{equation}
\trbig{\rho \mathcal{Q}_{L}[\tau]} \leq \opnorm{L}^4 \trbig{\rho \mathcal{Z}[\tau]}.
\end{equation} 
\qed
\begin{remark}
We note that Proposition~\ref{lem:inequality_bound_norm_fourth} can be generalized to measurements other than \glspl{PM}. 
Indeed, we have not used the assumption that $\mathcal{Q} = (Q_{1}, \ldots, Q_{r})$ is a \gls{PM}, only that each $Q_i$ for $i \in [r]$, is symmetric and rank one.
This suggests that we can extend Proposition~\ref{lem:inequality_bound_norm_fourth} to symmetric Kraus representations of rank-one \glspl{POVM}.
\end{remark}

\subsection{Combining all steps}
\label{sec:final_steps}

For any $\mathcal{P} \in P(\mathcal{H})$, we use Lemma~\ref{lem:reduction_sum_PMs} together with Proposition~\ref{lem:inequality_bound_norm_fourth} in Lemma~\ref{lem:upper_bound_expectation}, to conclude that the following inequalities hold.
\begin{align}
\trbig{\rho \mathcal{P}[\tau]} & =  \sum_{i=1}^{T} \trbig{\rho \mathcal{Q}_{L_i}[\tau]} \tag{Lemma~\ref{lem:reduction_sum_PMs}}\nonumber \\
 & \leq  3\expectationBig{\trbig{\rho \mathcal{\mathcal{Q}}_{\hat{L}}[\tau]}}  \tag{Lemma~\ref{lem:upper_bound_expectation}}  \\
& \leq  3 \expectationBig{\opnorm{\hat{L}}^4 \trbig{\rho \mathcal{Z}[\tau]}} \tag{Proposition~\ref{lem:inequality_bound_norm_fourth}} \nonumber \\
& \leq  3 \expectationBig{\opnorm{\hat{L}}^4 \max_{\mathcal{Y} \in P(\mathcal{S})} \trbig{\rho \mathcal{Y}[\tau]}}  \nonumber \\
& \leq  3 \Bigl( \expectationbig{\opnorm{\hat{L}}^4} \Bigr) \max_{\mathcal{Y} \in P(\mathcal{S})} \trbig{\rho \mathcal{Y}[\tau]}. 
\label{eqn:proof_upper_bound_1}
\end{align}
We assume first that $T$---that is, $n$--- is large enough so that Proposition~\ref{prop:variance_term_bounded} holds.
Corollary~\ref{cor:moment_bound_L_hat} then implies that there exists $c>0$ such that for $r \geq 2$,
\begin{align}
\expectationbig{\opnorm{\hat{L}}^{4}} \leq c \log(r)^4.
\label{eqn:proof_upper_bound_4}
\end{align}
In this case, we can use \eqref{eqn:proof_upper_bound_4} in \eqref{eqn:proof_upper_bound_1} to conclude that there exists a constant $c >0$ such that for any $\mathcal{P} \in P(\mathcal{H})$,
\begin{equation}
\trbig{\rho \mathcal{P}[\tau]} \leq c \log(r)^{4} \max_{\mathcal{Y} \in P(\mathcal{S})} \trbig{\rho \mathcal{Y}[\tau]}.
\label{eqn:proof_upper_bound_5}
\end{equation}
Taking the maximum over $\mathcal{P} \in P(\mathcal{H})$ yields Theorem~\ref{thm:main} in the case $n$ is large enough.\\

If $T$---that is, $n$---is small compared to $r$, and the assumptions of Proposition~\ref{prop:variance_term_bounded} do not hold, we can still show the result of Theorem~\ref{thm:main}.
Indeed, \eqref{eqn:proof_upper_bound_5} shows that there exists a constant $C>0$, and $n(r)$ such that when $n \geq n(r)$
\begin{equation}
\max_{\mathcal{P} \in P(\mathcal{H})} \trbig{\rho \mathcal{P}[\tau]} \leq C \log(r)^{4} \max_{\mathcal{Y} \in P(\mathcal{S})} \trbig{\rho \mathcal{Y}[\tau]}.
\label{eqn:proof_upper_bound_7}
\end{equation}
We need to check that \eqref{eqn:proof_upper_bound_7} also holds whenever $r \leq n < n(r)$. 
To do so simply note that we can enlarge the space $\mathcal{H}$ into a larger Hilbert space $\tilde{\mathcal{H}}$ of dimension $n(r) > n$, and consider the same optimization problem in \eqref{eqn:maximization} over $P(\tilde{\mathcal{H}})$ instead.
In this case, there exists an inclusion $\iota: P(\mathcal{H}) \to P(\tilde{\mathcal{H}})$ that preserves the value of $\trbig{\rho \mathcal{P}[\tau]}$ as shown for the case of $P(\mathcal{S})$ in \eqref{eqn:restriction_PM}.
From this inclusion, we can conclude
\begin{align}
\max_{\substack{\mathcal{P} \in P(\mathcal{H})}} \trbig{\rho \mathcal{P}[\tau]} & = \max_{\substack{\mathcal{P} \in P(\mathcal{H})}} \trbig{\rho \iota(\mathcal{P})[\tau]} \tag{\textrm{inclusion} $\iota: P(\mathcal{H}) \to P(\tilde{\mathcal{H}})$ } \nonumber \\
& \leq \max_{\substack{\tilde{\mathcal{P}} \in P(\tilde{\mathcal{H}})}} \trbig{\rho \tilde{\mathcal{P}}[\tau]} \nonumber \\
& \eqcom{\ref{eqn:proof_upper_bound_7}}\leq C \log(r)^{4} \max_{\substack{\mathcal{Y} \in P(\mathcal{S})}}  \trbig{\rho \mathcal{Y}[\tau]}.
\end{align}
This last step shows the claim of Theorem~\ref{thm:main} for any $n \geq r \geq 2$.

\section{Numerical experiments}
\label{sec:numerics}
We numerically examine the value of $K_{n,r} = K_r$, defined in \eqref{eqn:maximization}, and observe its dependency on $r$ for different values.
To do so, in dimension $n$, we use the formulation in \eqref{eqn:maximization_unitary} by setting
\begin{equation}
S_n(U) = \trbig{\rho \mathcal{E}(U)[\tau]} \quad \quad U \in U(n),
\label{eqn:numerics_1}
\end{equation}
where $\mathcal{E}$ is the measurement in the canonical basis of $\C^{n}$, and $\mathcal{E}(U) = (U E_1 U^{\dagger}, \ldots, U E_n U^{\dagger})$.
For $U \in U(n)$, can then find the gradient of $S_n(U)$ restricted to the tangent space of $U(n)$, by projecting
\begin{equation}
P_n(U) = \mathrm{Proj}_{\mathrm{T}_{U} U(n)}(\nabla S_n(U)).
\label{eqn:numerics_2}
\end{equation}
We refer to Appendix~\ref{secappendix:numerical_detials} for additional implementation details such as for \eqref{eqn:numerics_2}. 
The projected gradient ascent on the manifold $U(n)$ is given by initializing $U_0 \in U(n)$, and setting
\begin{equation}
U_{t+1} = \mathrm{exp}_{U_t}\Bigl(\beta_t P_n(U_t) \Bigr) \quad \text{ for } \quad t \in [t_{\max}],
\label{eqn:numerics_3}
\end{equation}
where $\mathrm{exp}_{U_t}$ is the exponential map $\mathrm{exp}: \mathrm{T}_{U_{t}} U(n) \to U(n)$, and $\beta_t > 0$ is the stepsize.
Estimating \eqref{eqn:numerics_3} is computationally expensive so we will use instead the following approximation
\begin{equation}
U_{t+1} = \mathrm{Proj}_{U(n)}\Bigl(U_t + \beta_t P_n(U_t) U_t \Bigr)  \quad \text{ for } \quad t \in [t_{\max}],
\label{eqn:numerics_4}
\end{equation}
where the projection to $U(n)$ is given by the unitary obtained from the polar decomposition.

We will choose $\rho, \tau \in D(\mathcal{S})$ at random, and $U_0 \in U(n)$ such that $U_0|_{\mathcal{S}} \sim \textrm{Unif}(U(r))$, and $U_0|_{\mathcal{S}^{\perp}} = \id{\mathcal{S}^{\perp}}$.
For optimizing $S_r(V)$ we will use the initialization $V_0 \sim \textrm{Unif}(U(r))$.
After $t_{\max}$ iterations of \eqref{eqn:numerics_4}, we denote the empirical estimator of $K_r$ by
\begin{equation}
\hat{K}_r = \frac{\max_{t \in [t_{\max}]} S_n(U_{t})}{\max_{t \in [t_{\max}]} S_r(V_{t})}.
\label{eqn:K_estimator}
\end{equation}
The estimated values of $\hat{K}_r$ can be found in Figure~\ref{fig:numerics}. They seem to suggest $K_{r} \simeq 1$ for $r \leq 20$.
Note that the estimator can satisfy $\hat{K}_r < 1$, since the optimization problem over $U(n)$ for $n>r$ commonly requires a larger horizon $t_{\max}$, and smaller stepsize to achieve the same suboptimality due to larger dimension.
Similarly, the optimization of $S_r(V)$ appears to have saddle points and local maxima that allow for $\hat{K}_r > 1$ to occur when $V_t$ converges to those.
Therefore, we cannot rule out that with other initializations---also for $\rho$, and $\tau$---a different behavior of $\hat{K}_r$ may be observed.
\begin{figure}[!hbt]
\centering
\includegraphics[width=0.8\textwidth]{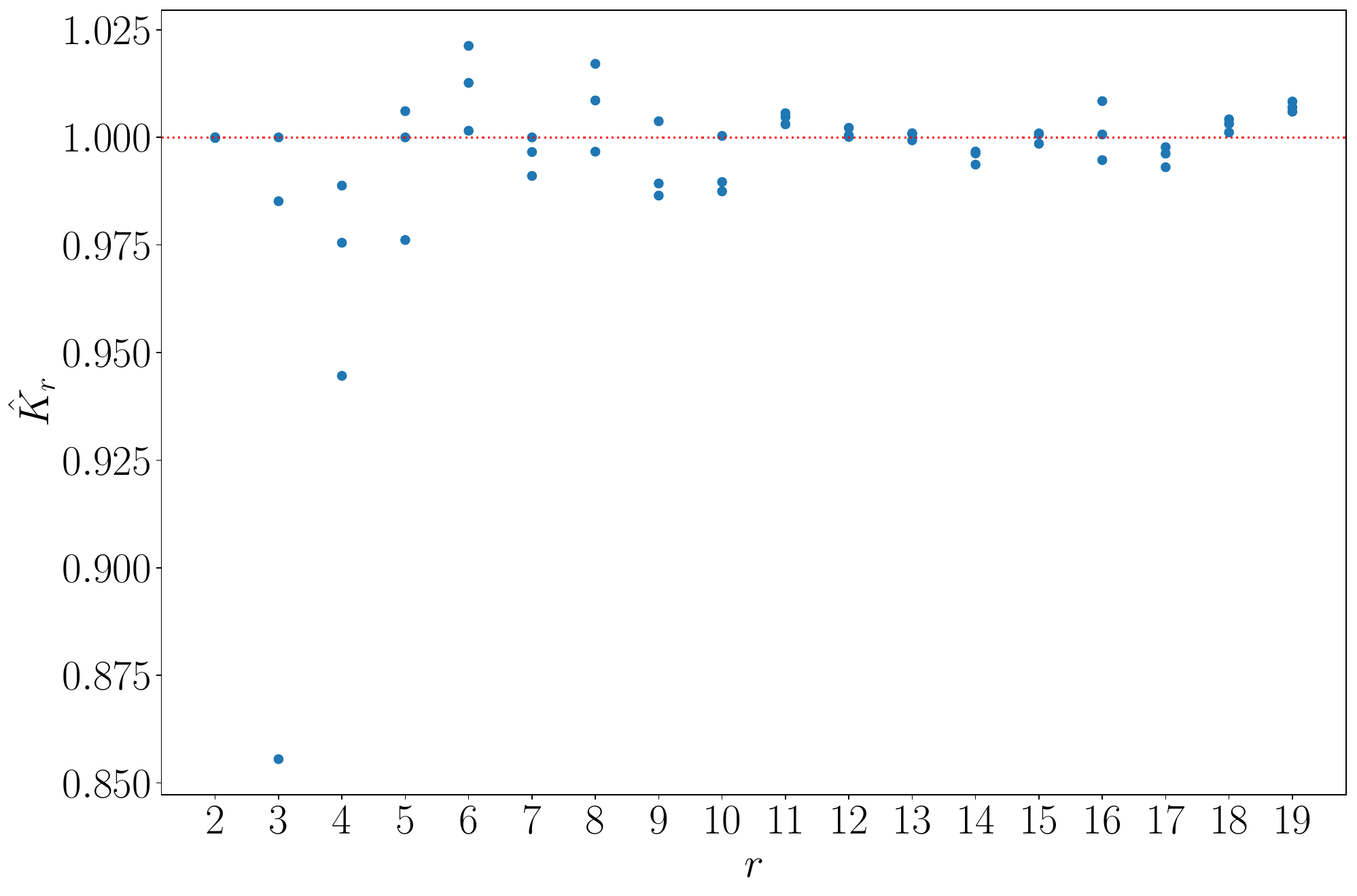}
\caption{Evaluation of $\hat{K}_r$ for $2 \leq r \leq 20$, and $n \in \{ r+1, r+2, r+3\}$ by using \eqref{eqn:K_estimator}, $t_{\max} = 10^{4}$, and $\beta_t = 10^{-2}$.}
\label{fig:numerics}
\end{figure}

It would be interesting to understand if there is no benefit in using \glspl{PM} that are not fully aligned with the subspace $\mathcal{S}$.
In this case, \eqref{eqn:intro_maximization} could be solved more efficiently when $\tau$, and $\rho$ have low-rank.
In this regard, inspired by previous remarks and numerical results we pose the following conjecture relative to \eqref{eqn:maximization}.
\begin{conjecture}
There exists $C \geq 1$ such that $K_{n,r} \leq C$ for all $n,r \geq 1$.
\label{conj:C_is_one}
\end{conjecture}

If Conjecture~\ref{conj:C_is_one} holds for $C = 1$, it would additionally distance the maximization problems in \eqref{eqn:maximization} from those in quantum information that bound maximal correlations of (entangled) quantum states, showcased by, e.g., Grothendieck's inequality \cite{blei2014grothendieck}.

\bibliography{biblio}
\bibliographystyle{plainnat}

\appendix 

\section{Proof of Lemma~\ref{lem:variance_bound_moment_bound}}
\label{secappendix:proof_of_lemma}
The proof of Lemma~\ref{lem:variance_bound_moment_bound} consists of tracking and estimating the moments of the matrix $\mathcal{L}^{\pi}(\Theta)$.
In Appendix~\ref{secappendix:preliminaries}, we introduce the problem and the objects that will help with counting, such as cycles, and factors.
In Appendix~\ref{secappendix:tracking} we estimate the chance that certain factors appear in the expansion, and how they depend on $T$.
Next in Appendix~\ref{secappendix:symmetries} we characterize symmetries in the expansion of $F(k)$ that will simplify the counting later on.
In Appendix~\ref{secappendix:combinatorial}, we bound some of the products in $F(k)$, and finally in Appendix~\ref{secappendix:moment_estimate} we use all previous work to estimate the asymptotic expansion as $T$ becomes large; see Lemma~\ref{lem:moment_bound_variance}.

\subsection{Preliminaries}
\label{secappendix:preliminaries}
For simplicity we will assume that there are $m = T r \geq n$ frame vectors exactly---we complete with the null vector, if necessary.
Recall \eqref{eqn:def_L_pi_theta}, and define for all $t \in [T]$, and $i,j \in [r]$,
\begin{align}
\mathcal{L}_{t}^{\pi}(\Theta) & = L_{t}^{\pi}(\Theta)^{\dagger}L_{t}^{\pi}(\Theta) \quad \quad  \text{ and } \nonumber \\
(\mathcal{L}_{t}^{\pi}(\Theta))_{ij} & = \braket{v_{\pi(r(t-1) + i)}(\Theta) \mid v_{\pi(r(t-1) + j}(\Theta))} \nonumber \\
& = \exp\bigl(\ci \theta_{\pi(r(t-1) + j)} - \ci \theta_{\pi(r(t-1) + i)} \bigr)\braket{v_{\pi(r(t-1) + i)} \mid v_{\pi(r(t-1) + j)}}.
\label{eqn:appendix_L_pi_theta_t}
\end{align}
In this notation, we have
\begin{align}
\mathcal{L}^{\pi}(\Theta) &= \sum_{t = 1}^{T} \mathcal{L}_{t}^{\pi}(\Theta), \quad \quad \text{ and }\nonumber \\
(\mathcal{L}^{\pi}(\Theta))_{ij} & = \sum_{t=1}^{T} \braket{v_{\pi(r(t-1) + i)}(\Theta) \mid v_{\pi(r(t-1) + j}(\Theta))} \nonumber \\
& = \sum_{t=1}^{T} \exp\bigl(\ci (\theta_{\pi(r(t-1) + j)} - \theta_{\pi(r(t-1) + i)} \bigr)\braket{v_{\pi(r(t-1) + i)} \mid v_{\pi(r(t-1) + j)}}
\label{eqn:appendix_L_pi_theta}
\end{align}

Recall that
\begin{equation}
F(k) = \expectationBigWrt{\trbig{\mathcal{L}^{\pi}(\Theta)^{k}}}{\pi, \Theta}.
\label{eqn:def_F_k}
\end{equation}
For a fixed $t \in [T]$, consider one of the factors appearing in the expansion of $\tr{\mathcal{L}^{\pi}(\Theta)^{k}}$, say
\begin{equation}
\langle v_{\pi(i_1)}(\Theta) \mid v_{\pi(i_2)}(\Theta) \rangle \cdots \langle v_{\pi(i_{k})}(\Theta) \mid v_{\pi(i_{k+1})}(\Theta) \rangle.
\label{eqn:appendix_inner_product_vectors}
\end{equation}
Corresponding to an index $\pi(i_1)$, if the vector $\ket{v_{\pi(i_1)}(\Theta)}$ does not appear the same number of times in the primal and dual variables of of inner product $\langle \cdot \mid \cdot \rangle$ in \eqref{eqn:appendix_inner_product_vectors}, the expectation over $\Theta$ will be zero.
This is due to the fact that if $\theta \sim \mathrm{Unif}([0,2\pi])$,
\begin{equation}
\expectationbigWrt{\exp(\ci l \theta)}{\theta} = \begin{cases}
0 \text{ if } l \in \Z \backslash \{0\} \\
1 \text{ if } l = 0.
\end{cases}
\end{equation}
Only vectors with the same number of appearances in the primal and dual arguments will thus avoid the cancellation by the expectation over the phase modifier $\Theta$.
For other matrix products appearing in the expansion $\trbig{\mathcal{L}^{\pi}(\Theta)^{k}}$ the same argument applies.
This limits the amount of possible factors such as \eqref{eqn:appendix_inner_product_vectors} appearing in $F(k)$.
For a fixed $\pi$, we need first to characterize the type of factors that appear in the expansion
\begin{equation}
F(k, \pi) = \expectationBigWrt{\trbig{\mathcal{L}^{\pi}(\Theta)^{k}}}{\Theta}.
\label{eqn:def_F_k_pi}
\end{equation}

A permutation $\pi \in S_{m}$ will induce a partition $\mathcal{V}^{\pi}$ of $[m]$, given by the partition sets
\begin{align}
\mathcal{V}_1^{\pi} &= \bigl\{ \pi(1), \cdots, \pi(r) \bigr\} \nonumber \\
& \cdots  \nonumber \\
\mathcal{V}_{T}^{\pi} &= \bigl\{ \pi(r(T-1) + 1), \cdots, \pi(Tr) \bigr\}.
\label{eqn:def_partitions_permuted}
\end{align}
If $i \in \mathcal{V}_t^{\pi}$, the matrix $L_t^{\pi}(\Theta)$ will contain the vector $\ket{v_i}$. 
Then, only when $i,j \in \mathcal{V}_t^{\pi}$ for some $t \in [T]$, can the inner products $\braket{v_i \mid v_j}$ or $\braket{v_j \mid v_i}$ appear in the expansion of terms in \eqref{eqn:def_F_k_pi} that contain the matrix $\mathcal{L}_t^{\pi}(\Theta)$.
Similarly, if the term $\braket{v_i \mid v_j}  \braket{v_k \mid v_l}$ appears adjacent in the trace of some factor in \eqref{eqn:def_F_k_pi}, but for that permutation $\pi$ we have that $j \in \mathcal{V}_t^{\pi}$, and $k \in \mathcal{V}_s^{\pi}$ belong to different partitions, it implies that $j$, and $k$ have same position as vectors in the matrices $\mathcal{L}_t^{\pi}(\Theta)$, and $\mathcal{L}_s^{\pi}(\Theta)$ respectively, and the matrix product $\mathcal{L}_t^{\pi}(\Theta)\mathcal{L}_s^{\pi}(\Theta)$ appears in the trace.\\

A useful representation for a partition \eqref{eqn:def_partitions_permuted} is encoded in its associated graph, depicted in Figure~\ref{fig:cycle_representation}.
If $K_{z}$ is the complete graph with $z$ vertices, the graph associated to $\pi$ is isomorphic to the Cartesian product $K_{r} \Box K_{T}$.
The Cartesian product satisfies that $(i,t), (j,s) \in K_{r} \Box K_{T}$ are adjacent if and only if $i = j$, and $t$ is adjacent to $s$ or $i$ is adjacent to $j$, and $t =s$.
In our setting, only the edges corresponding to the component $K_r$ will give factors $\braket{v_i \mid v_j}$, while each edge on the component $K_{T}$ indicate that the partition has changed, that is, the trace contains a product of different matrices, e.g., $\mathcal{L}_t^{\pi}(\Theta)\mathcal{L}_s^{\pi}(\Theta)$.
The type of factors that appear in $F(k, \pi)$ will depend on cycles on $K_{r} \Box K_{T}$ that are compatible with the permutation. 
We characterize such cycles in the following definition.

\begin{figure}[!hbt]
\centering
\begin{subfigure}{0.45\textwidth}
\centering

\tikzset{every picture/.style={line width=0.75pt}} %

\begin{tikzpicture}[x=0.75pt,y=0.75pt,yscale=-0.9,xscale=0.9]
\draw   (38.2,43.62) .. controls (38.2,35.21) and (45.01,28.4) .. (53.42,28.4) .. controls (61.83,28.4) and (68.64,35.21) .. (68.64,43.62) .. controls (68.64,52.03) and (61.83,58.84) .. (53.42,58.84) .. controls (45.01,58.84) and (38.2,52.03) .. (38.2,43.62) -- cycle ;
\draw   (38.2,85.98) .. controls (38.2,77.57) and (45.01,70.76) .. (53.42,70.76) .. controls (61.83,70.76) and (68.64,77.57) .. (68.64,85.98) .. controls (68.64,94.38) and (61.83,101.2) .. (53.42,101.2) .. controls (45.01,101.2) and (38.2,94.38) .. (38.2,85.98) -- cycle ;
\draw   (38.2,130.98) .. controls (38.2,122.57) and (45.01,115.76) .. (53.42,115.76) .. controls (61.83,115.76) and (68.64,122.57) .. (68.64,130.98) .. controls (68.64,139.38) and (61.83,146.2) .. (53.42,146.2) .. controls (45.01,146.2) and (38.2,139.38) .. (38.2,130.98) -- cycle ;
\draw   (132.2,43.38) .. controls (132.2,34.97) and (139.01,28.16) .. (147.42,28.16) .. controls (155.83,28.16) and (162.64,34.97) .. (162.64,43.38) .. controls (162.64,51.78) and (155.83,58.6) .. (147.42,58.6) .. controls (139.01,58.6) and (132.2,51.78) .. (132.2,43.38) -- cycle ;
\draw   (132.2,85.74) .. controls (132.2,77.33) and (139.01,70.52) .. (147.42,70.52) .. controls (155.83,70.52) and (162.64,77.33) .. (162.64,85.74) .. controls (162.64,94.14) and (155.83,100.96) .. (147.42,100.96) .. controls (139.01,100.96) and (132.2,94.14) .. (132.2,85.74) -- cycle ;
\draw   (132.2,130.74) .. controls (132.2,122.33) and (139.01,115.52) .. (147.42,115.52) .. controls (155.83,115.52) and (162.64,122.33) .. (162.64,130.74) .. controls (162.64,139.14) and (155.83,145.96) .. (147.42,145.96) .. controls (139.01,145.96) and (132.2,139.14) .. (132.2,130.74) -- cycle ;
\draw   (227.2,43.38) .. controls (227.2,34.97) and (234.01,28.16) .. (242.42,28.16) .. controls (250.83,28.16) and (257.64,34.97) .. (257.64,43.38) .. controls (257.64,51.78) and (250.83,58.6) .. (242.42,58.6) .. controls (234.01,58.6) and (227.2,51.78) .. (227.2,43.38) -- cycle ;
\draw   (227.2,85.74) .. controls (227.2,77.33) and (234.01,70.52) .. (242.42,70.52) .. controls (250.83,70.52) and (257.64,77.33) .. (257.64,85.74) .. controls (257.64,94.14) and (250.83,100.96) .. (242.42,100.96) .. controls (234.01,100.96) and (227.2,94.14) .. (227.2,85.74) -- cycle ;
\draw   (227.2,130.74) .. controls (227.2,122.33) and (234.01,115.52) .. (242.42,115.52) .. controls (250.83,115.52) and (257.64,122.33) .. (257.64,130.74) .. controls (257.64,139.14) and (250.83,145.96) .. (242.42,145.96) .. controls (234.01,145.96) and (227.2,139.14) .. (227.2,130.74) -- cycle ;
\draw  [dash pattern={on 4.5pt off 4.5pt}]  (38.2,43.62) .. controls (8.06,43.99) and (5.11,83.55) .. (36.25,85.88) ;
\draw [shift={(38.2,85.98)}, rotate = 181.38] [fill={rgb, 255:red, 0; green, 0; blue, 0 }  ][line width=0.08]  [draw opacity=0] (12,-3) -- (0,0) -- (12,3) -- cycle    ;
\draw    (68.64,85.98) -- (130.2,85.74) ;
\draw [shift={(132.2,85.74)}, rotate = 179.78] [fill={rgb, 255:red, 0; green, 0; blue, 0 }  ][line width=0.08]  [draw opacity=0] (12,-3) -- (0,0) -- (12,3) -- cycle    ;
\draw  [dash pattern={on 4.5pt off 4.5pt}]  (132.2,85.74) .. controls (94.42,86.78) and (96.18,130.65) .. (128.73,131.21) ;
\draw [shift={(130.24,131.2)}, rotate = 178.65] [fill={rgb, 255:red, 0; green, 0; blue, 0 }  ][line width=0.08]  [draw opacity=0] (12,-3) -- (0,0) -- (12,3) -- cycle    ;
\draw  [dash pattern={on 0.84pt off 2.51pt}] (28.64,18.54) -- (75.84,18.54) -- (75.84,153.6) -- (28.64,153.6) -- cycle ;
\draw  [dash pattern={on 0.84pt off 2.51pt}] (123.82,18.54) -- (171.02,18.54) -- (171.02,153.6) -- (123.82,153.6) -- cycle ;
\draw  [dash pattern={on 0.84pt off 2.51pt}] (218.82,18.54) -- (266.02,18.54) -- (266.02,153.6) -- (218.82,153.6) -- cycle ;
\draw    (131.9,80.9) -- (67.9,81.38) ;
\draw [shift={(65.9,81.4)}, rotate = 359.57] [fill={rgb, 255:red, 0; green, 0; blue, 0 }  ][line width=0.08]  [draw opacity=0] (12,-3) -- (0,0) -- (12,3) -- cycle    ;
\draw  [dash pattern={on 4.5pt off 4.5pt}]  (162.64,130.74) .. controls (200.42,129.69) and (196.76,86.27) .. (164.15,85.73) ;
\draw [shift={(162.64,85.74)}, rotate = 358.65] [fill={rgb, 255:red, 0; green, 0; blue, 0 }  ][line width=0.08]  [draw opacity=0] (12,-3) -- (0,0) -- (12,3) -- cycle    ;
\draw  [dash pattern={on 4.5pt off 4.5pt}]  (68.64,85.98) .. controls (106.42,84.93) and (102.76,41.51) .. (70.15,40.97) ;
\draw [shift={(68.64,40.98)}, rotate = 358.65] [fill={rgb, 255:red, 0; green, 0; blue, 0 }  ][line width=0.08]  [draw opacity=0] (12,-3) -- (0,0) -- (12,3) -- cycle    ;

\draw (41.11,36.64) node [anchor=north west][inner sep=0.75pt]  [font=\scriptsize]  {$\pi ( 1)$};
\draw (41.11,79) node [anchor=north west][inner sep=0.75pt]  [font=\scriptsize]  {$\pi ( 2)$};
\draw (41.11,124) node [anchor=north west][inner sep=0.75pt]  [font=\scriptsize]  {$\pi ( 3)$};
\draw (135.11,36.4) node [anchor=north west][inner sep=0.75pt]  [font=\scriptsize]  {$\pi ( 4)$};
\draw (135.11,78.76) node [anchor=north west][inner sep=0.75pt]  [font=\scriptsize]  {$\pi ( 5)$};
\draw (135.11,123.76) node [anchor=north west][inner sep=0.75pt]  [font=\scriptsize]  {$\pi ( 6)$};
\draw (230.11,36.4) node [anchor=north west][inner sep=0.75pt]  [font=\scriptsize]  {$\pi ( 7)$};
\draw (230.11,78.76) node [anchor=north west][inner sep=0.75pt]  [font=\scriptsize]  {$\pi ( 8)$};
\draw (230.11,123.76) node [anchor=north west][inner sep=0.75pt]  [font=\scriptsize]  {$\pi ( 9)$};
\draw (40.4,156.8) node [anchor=north west][inner sep=0.75pt]    {$\mathcal{V}_{1}^{\pi }$};
\draw (133.6,156.8) node [anchor=north west][inner sep=0.75pt]    {$\mathcal{V}_{2}^{\pi }$};
\draw (230.8,156.8) node [anchor=north west][inner sep=0.75pt]    {$\mathcal{V}_{3}^{\pi }$};

\end{tikzpicture}
 \caption{}
\end{subfigure}
\hfill
\begin{subfigure}{0.45\textwidth}
\centering

\tikzset{every picture/.style={line width=0.75pt}} %

\begin{tikzpicture}[x=0.75pt,y=0.75pt,yscale=-0.9,xscale=0.9]
\draw   (58.2,63.62) .. controls (58.2,55.21) and (65.01,48.4) .. (73.42,48.4) .. controls (81.83,48.4) and (88.64,55.21) .. (88.64,63.62) .. controls (88.64,72.03) and (81.83,78.84) .. (73.42,78.84) .. controls (65.01,78.84) and (58.2,72.03) .. (58.2,63.62) -- cycle ;
\draw   (58.2,105.98) .. controls (58.2,97.57) and (65.01,90.76) .. (73.42,90.76) .. controls (81.83,90.76) and (88.64,97.57) .. (88.64,105.98) .. controls (88.64,114.38) and (81.83,121.2) .. (73.42,121.2) .. controls (65.01,121.2) and (58.2,114.38) .. (58.2,105.98) -- cycle ;
\draw   (58.2,150.98) .. controls (58.2,142.57) and (65.01,135.76) .. (73.42,135.76) .. controls (81.83,135.76) and (88.64,142.57) .. (88.64,150.98) .. controls (88.64,159.38) and (81.83,166.2) .. (73.42,166.2) .. controls (65.01,166.2) and (58.2,159.38) .. (58.2,150.98) -- cycle ;
\draw   (152.2,63.38) .. controls (152.2,54.97) and (159.01,48.16) .. (167.42,48.16) .. controls (175.83,48.16) and (182.64,54.97) .. (182.64,63.38) .. controls (182.64,71.78) and (175.83,78.6) .. (167.42,78.6) .. controls (159.01,78.6) and (152.2,71.78) .. (152.2,63.38) -- cycle ;
\draw   (152.2,105.74) .. controls (152.2,97.33) and (159.01,90.52) .. (167.42,90.52) .. controls (175.83,90.52) and (182.64,97.33) .. (182.64,105.74) .. controls (182.64,114.14) and (175.83,120.96) .. (167.42,120.96) .. controls (159.01,120.96) and (152.2,114.14) .. (152.2,105.74) -- cycle ;
\draw   (152.2,150.74) .. controls (152.2,142.33) and (159.01,135.52) .. (167.42,135.52) .. controls (175.83,135.52) and (182.64,142.33) .. (182.64,150.74) .. controls (182.64,159.14) and (175.83,165.96) .. (167.42,165.96) .. controls (159.01,165.96) and (152.2,159.14) .. (152.2,150.74) -- cycle ;
\draw   (247.2,63.38) .. controls (247.2,54.97) and (254.01,48.16) .. (262.42,48.16) .. controls (270.83,48.16) and (277.64,54.97) .. (277.64,63.38) .. controls (277.64,71.78) and (270.83,78.6) .. (262.42,78.6) .. controls (254.01,78.6) and (247.2,71.78) .. (247.2,63.38) -- cycle ;
\draw   (247.2,105.74) .. controls (247.2,97.33) and (254.01,90.52) .. (262.42,90.52) .. controls (270.83,90.52) and (277.64,97.33) .. (277.64,105.74) .. controls (277.64,114.14) and (270.83,120.96) .. (262.42,120.96) .. controls (254.01,120.96) and (247.2,114.14) .. (247.2,105.74) -- cycle ;
\draw   (247.2,150.74) .. controls (247.2,142.33) and (254.01,135.52) .. (262.42,135.52) .. controls (270.83,135.52) and (277.64,142.33) .. (277.64,150.74) .. controls (277.64,159.14) and (270.83,165.96) .. (262.42,165.96) .. controls (254.01,165.96) and (247.2,159.14) .. (247.2,150.74) -- cycle ;
\draw  [dash pattern={on 4.5pt off 4.5pt}]  (58.2,63.62) .. controls (28.06,63.99) and (25.11,103.55) .. (56.25,105.88) ;
\draw [shift={(58.2,105.98)}, rotate = 181.38] [fill={rgb, 255:red, 0; green, 0; blue, 0 }  ][line width=0.08]  [draw opacity=0] (12,-3) -- (0,0) -- (12,3) -- cycle    ;
\draw    (88.64,105.98) -- (150.2,105.74) ;
\draw [shift={(152.2,105.74)}, rotate = 179.78] [fill={rgb, 255:red, 0; green, 0; blue, 0 }  ][line width=0.08]  [draw opacity=0] (12,-3) -- (0,0) -- (12,3) -- cycle    ;
\draw  [dash pattern={on 4.5pt off 4.5pt}]  (152.2,105.74) .. controls (114.42,106.78) and (116.18,150.65) .. (148.73,151.21) ;
\draw [shift={(150.24,151.2)}, rotate = 178.65] [fill={rgb, 255:red, 0; green, 0; blue, 0 }  ][line width=0.08]  [draw opacity=0] (12,-3) -- (0,0) -- (12,3) -- cycle    ;
\draw    (182.64,150.74) -- (245.2,150.74) ;
\draw [shift={(247.2,150.74)}, rotate = 180] [fill={rgb, 255:red, 0; green, 0; blue, 0 }  ][line width=0.08]  [draw opacity=0] (12,-3) -- (0,0) -- (12,3) -- cycle    ;
\draw  [dash pattern={on 4.5pt off 4.5pt}]  (247.2,150.74) .. controls (181.26,152.15) and (187.8,65.15) .. (245.44,63.4) ;
\draw [shift={(247.2,63.38)}, rotate = 180] [fill={rgb, 255:red, 0; green, 0; blue, 0 }  ][line width=0.08]  [draw opacity=0] (12,-3) -- (0,0) -- (12,3) -- cycle    ;
\draw  [dash pattern={on 4.5pt off 4.5pt}]  (247.2,63.38) .. controls (210.08,43.5) and (251.27,15.42) .. (261.95,46.69) ;
\draw [shift={(262.42,48.16)}, rotate = 253.62] [fill={rgb, 255:red, 0; green, 0; blue, 0 }  ][line width=0.08]  [draw opacity=0] (12,-3) -- (0,0) -- (12,3) -- cycle    ;
\draw  [dash pattern={on 0.84pt off 2.51pt}] (48.64,38.54) -- (95.84,38.54) -- (95.84,173.6) -- (48.64,173.6) -- cycle ;
\draw  [dash pattern={on 0.84pt off 2.51pt}] (143.82,38.54) -- (191.02,38.54) -- (191.02,173.6) -- (143.82,173.6) -- cycle ;
\draw  [dash pattern={on 0.84pt off 2.51pt}] (238.82,38.54) -- (286.02,38.54) -- (286.02,173.6) -- (238.82,173.6) -- cycle ;
\draw  (262.42,48.16) .. controls (237.84,29.2) and (207.84,27.6) .. (168.64,28) .. controls (130.22,28.39) and (99.09,31.47) .. (74.89,47.41) ;
\draw [shift={(73.42,48.4)}, rotate = 325.45] [fill={rgb, 255:red, 0; green, 0; blue, 0 }  ][line width=0.08]  [draw opacity=0] (12,-3) -- (0,0) -- (12,3) -- cycle    ;

\draw (61.11,56.64) node [anchor=north west][inner sep=0.75pt]  [font=\scriptsize]  {$\pi ( 1)$};
\draw (61.11,99) node [anchor=north west][inner sep=0.75pt]  [font=\scriptsize]  {$\pi ( 2)$};
\draw (61.11,144) node [anchor=north west][inner sep=0.75pt]  [font=\scriptsize]  {$\pi ( 3)$};
\draw (155.11,56.4) node [anchor=north west][inner sep=0.75pt]  [font=\scriptsize]  {$\pi ( 4)$};
\draw (155.11,98.76) node [anchor=north west][inner sep=0.75pt]  [font=\scriptsize]  {$\pi ( 5)$};
\draw (155.11,143.76) node [anchor=north west][inner sep=0.75pt]  [font=\scriptsize]  {$\pi ( 6)$};
\draw (250.11,56.4) node [anchor=north west][inner sep=0.75pt]  [font=\scriptsize]  {$\pi ( 7)$};
\draw (250.11,98.76) node [anchor=north west][inner sep=0.75pt]  [font=\scriptsize]  {$\pi ( 8)$};
\draw (250.11,143.76) node [anchor=north west][inner sep=0.75pt]  [font=\scriptsize]  {$\pi ( 9)$};
\draw (60.4,176.8) node [anchor=north west][inner sep=0.75pt]    {$\mathcal{V}_{1}^{\pi }$};
\draw (153.6,176.8) node [anchor=north west][inner sep=0.75pt]    {$\mathcal{V}_{2}^{\pi }$};
\draw (250.8,176.8) node [anchor=north west][inner sep=0.75pt]    {$\mathcal{V}_{3}^{\pi }$};

\end{tikzpicture}
 \caption{}
\end{subfigure}
\caption{The graph associated to a partition $\mathcal{V}^{\pi}$ with
examples of cycles starting at $\pi(1)$. 
Solid lines represent the transitions between partition sets, and can only be horizontal. 
Dashed lines represent transitions within partitions and add a factor $\langle v_i, v_j \rangle$ to a product such as \eqref{eqn:appendix_inner_product_vectors}. 
(a) An example of a representation of a nonvanishing cycle $\gamma = \pi(1, 2, 5, 6, 6, 5, 2, 1)$. 
(b) Example of a cycle $\delta$ that will not appear in $F(k)$.
Indeed, there are edges within partitions that are visited only once. Hence, the product in \eqref{eqn:appendix_inner_product_vectors} associated to the cycle will have expectation zero over $\Theta$.}
\label{fig:cycle_representation}
\end{figure}
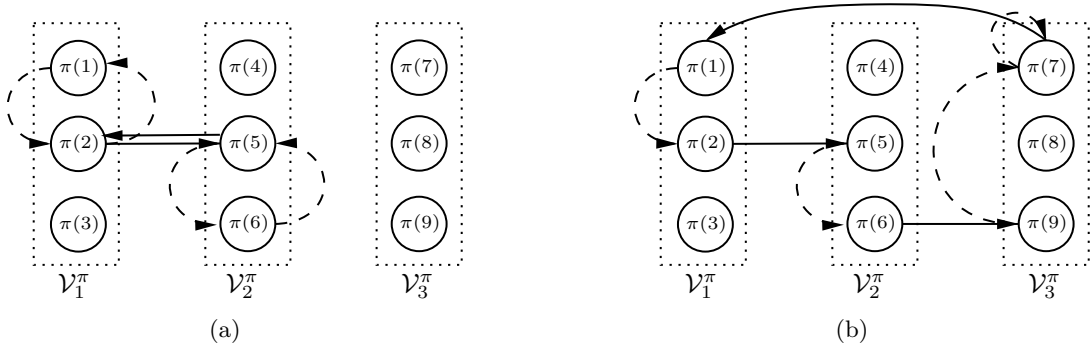

\begin{definition}
For a permutation $\pi$, a valid \emph{cycle} $\gamma$ (in our context) is a concatenation of $2k$ indices in $[m]$ up to permutation that generate a cycle in the graph induced by the partition $\mathcal{V}^{\pi}$, and such that 
\begin{itemize}
\item[(i)] Restricted to each partition set $\mathcal{V}^{\pi}_t$ for $t \in [T]$, the cycle is union of disjoint cycles.
\item[(ii)] Transitions of the cycle $\gamma$ between partition sets only occur between indices, $i,j \in [m]$ such that $\pi^{-1}(i) \equiv  \pi^{-1}(j) \mod r$.
\end{itemize}
\label{def:cycle}
\end{definition}
\noindent
If $\gamma$ is a valid cycle for $\pi$, the expectation over $\Theta$ of its associated factor
\begin{equation}
V_{\gamma}(\Theta) = \langle v_{\gamma_1}(\Theta), v_{\gamma_2}(\Theta) \rangle \cdots \langle v_{\gamma_k}(\Theta), v_{\gamma_{k+1}}(\Theta) \rangle,
\label{eqn:appendix_V_gamma}
\end{equation}
does not vanish because of Definition~\ref{def:cycle} (i). That is, $V_{\gamma}(\Theta) = V_{\gamma}$ is independent of $\Theta$.
Moreover, the transitions between partitions of $\gamma$ will constrain the matrix product in \eqref{eqn:def_F_k_pi} that generated $V_{\gamma}$.

\begin{figure}[!hbt]
\centering

\tikzset{every picture/.style={line width=0.75pt}} %

\begin{tikzpicture}[x=0.75pt,y=0.75pt,yscale=-1,xscale=1]
\draw   (38.2,43.62) .. controls (38.2,35.21) and (45.01,28.4) .. (53.42,28.4) .. controls (61.83,28.4) and (68.64,35.21) .. (68.64,43.62) .. controls (68.64,52.03) and (61.83,58.84) .. (53.42,58.84) .. controls (45.01,58.84) and (38.2,52.03) .. (38.2,43.62) -- cycle ;
\draw   (38.2,85.98) .. controls (38.2,77.57) and (45.01,70.76) .. (53.42,70.76) .. controls (61.83,70.76) and (68.64,77.57) .. (68.64,85.98) .. controls (68.64,94.38) and (61.83,101.2) .. (53.42,101.2) .. controls (45.01,101.2) and (38.2,94.38) .. (38.2,85.98) -- cycle ;
\draw   (38.2,130.98) .. controls (38.2,122.57) and (45.01,115.76) .. (53.42,115.76) .. controls (61.83,115.76) and (68.64,122.57) .. (68.64,130.98) .. controls (68.64,139.38) and (61.83,146.2) .. (53.42,146.2) .. controls (45.01,146.2) and (38.2,139.38) .. (38.2,130.98) -- cycle ;
\draw   (132.2,43.38) .. controls (132.2,34.97) and (139.01,28.16) .. (147.42,28.16) .. controls (155.83,28.16) and (162.64,34.97) .. (162.64,43.38) .. controls (162.64,51.78) and (155.83,58.6) .. (147.42,58.6) .. controls (139.01,58.6) and (132.2,51.78) .. (132.2,43.38) -- cycle ;
\draw   (132.2,85.74) .. controls (132.2,77.33) and (139.01,70.52) .. (147.42,70.52) .. controls (155.83,70.52) and (162.64,77.33) .. (162.64,85.74) .. controls (162.64,94.14) and (155.83,100.96) .. (147.42,100.96) .. controls (139.01,100.96) and (132.2,94.14) .. (132.2,85.74) -- cycle ;
\draw   (132.2,130.74) .. controls (132.2,122.33) and (139.01,115.52) .. (147.42,115.52) .. controls (155.83,115.52) and (162.64,122.33) .. (162.64,130.74) .. controls (162.64,139.14) and (155.83,145.96) .. (147.42,145.96) .. controls (139.01,145.96) and (132.2,139.14) .. (132.2,130.74) -- cycle ;
\draw   (227.2,43.38) .. controls (227.2,34.97) and (234.01,28.16) .. (242.42,28.16) .. controls (250.83,28.16) and (257.64,34.97) .. (257.64,43.38) .. controls (257.64,51.78) and (250.83,58.6) .. (242.42,58.6) .. controls (234.01,58.6) and (227.2,51.78) .. (227.2,43.38) -- cycle ;
\draw   (227.2,85.74) .. controls (227.2,77.33) and (234.01,70.52) .. (242.42,70.52) .. controls (250.83,70.52) and (257.64,77.33) .. (257.64,85.74) .. controls (257.64,94.14) and (250.83,100.96) .. (242.42,100.96) .. controls (234.01,100.96) and (227.2,94.14) .. (227.2,85.74) -- cycle ;
\draw   (227.2,130.74) .. controls (227.2,122.33) and (234.01,115.52) .. (242.42,115.52) .. controls (250.83,115.52) and (257.64,122.33) .. (257.64,130.74) .. controls (257.64,139.14) and (250.83,145.96) .. (242.42,145.96) .. controls (234.01,145.96) and (227.2,139.14) .. (227.2,130.74) -- cycle ;
\draw  [dash pattern={on 4.5pt off 4.5pt}]  (38.2,43.62) .. controls (8.06,43.99) and (5.11,83.55) .. (36.25,85.88) ;
\draw [shift={(38.2,85.98)}, rotate = 181.38] [fill={rgb, 255:red, 0; green, 0; blue, 0 }  ][line width=0.08]  [draw opacity=0] (12,-3) -- (0,0) -- (12,3) -- cycle    ;
\draw    (68.64,85.98) -- (130.2,85.74) ;
\draw [shift={(132.2,85.74)}, rotate = 179.78] [fill={rgb, 255:red, 0; green, 0; blue, 0 }  ][line width=0.08]  [draw opacity=0] (12,-3) -- (0,0) -- (12,3) -- cycle    ;
\draw  [dash pattern={on 4.5pt off 4.5pt}]  (132.2,85.74) .. controls (94.61,86.78) and (103.04,115.23) .. (129.27,123.89) ;
\draw [shift={(130.9,124.4)}, rotate = 196.22] [fill={rgb, 255:red, 0; green, 0; blue, 0 }  ][line width=0.08]  [draw opacity=0] (12,-3) -- (0,0) -- (12,3) -- cycle    ;
\draw  [dash pattern={on 0.84pt off 2.51pt}] (28.64,19.2) -- (75.84,19.2) -- (75.84,153.6) -- (28.64,153.6) -- cycle ;
\draw  [dash pattern={on 0.84pt off 2.51pt}] (123.82,18.54) -- (171.02,18.54) -- (171.02,152.94) -- (123.82,152.94) -- cycle ;
\draw  [dash pattern={on 0.84pt off 2.51pt}] (218.82,18.54) -- (266.02,18.54) -- (266.02,152.94) -- (218.82,152.94) -- cycle ;
\draw    (131.9,82.9) -- (67.9,83.38) ;
\draw [shift={(65.9,83.4)}, rotate = 359.57] [fill={rgb, 255:red, 0; green, 0; blue, 0 }  ][line width=0.08]  [draw opacity=0] (12,-3) -- (0,0) -- (12,3) -- cycle    ;
\draw  [dash pattern={on 4.5pt off 4.5pt}]  (162.64,130.74) .. controls (200.42,129.69) and (196.76,86.27) .. (164.15,85.73) ;
\draw [shift={(162.64,85.74)}, rotate = 358.65] [fill={rgb, 255:red, 0; green, 0; blue, 0 }  ][line width=0.08]  [draw opacity=0] (12,-3) -- (0,0) -- (12,3) -- cycle    ;
\draw  [dash pattern={on 4.5pt off 4.5pt}]  (68.64,85.98) .. controls (106.42,84.93) and (102.76,41.51) .. (70.15,40.97) ;
\draw [shift={(68.64,40.98)}, rotate = 358.65] [fill={rgb, 255:red, 0; green, 0; blue, 0 }  ][line width=0.08]  [draw opacity=0] (12,-3) -- (0,0) -- (12,3) -- cycle    ;
\draw    (132.9,127.9) -- (68.9,128.38) ;
\draw [shift={(66.9,128.4)}, rotate = 359.57] [fill={rgb, 255:red, 0; green, 0; blue, 0 }  ][line width=0.08]  [draw opacity=0] (12,-3) -- (0,0) -- (12,3) -- cycle    ;
\draw    (68.64,130.98) -- (130.2,130.74) ;
\draw [shift={(132.2,130.74)}, rotate = 179.78] [fill={rgb, 255:red, 0; green, 0; blue, 0 }  ][line width=0.08]  [draw opacity=0] (12,-3) -- (0,0) -- (12,3) -- cycle    ;
\draw  [dash pattern={on 4.5pt off 4.5pt}]  (53.42,115.76) .. controls (46.01,86.84) and (4.67,117.74) .. (36.69,130.41) ;
\draw [shift={(38.2,130.98)}, rotate = 199.43] [fill={rgb, 255:red, 0; green, 0; blue, 0 }  ][line width=0.08]  [draw opacity=0] (12,-3) -- (0,0) -- (12,3) -- cycle    ;

\draw (41.11,36.64) node [anchor=north west][inner sep=0.75pt]  [font=\scriptsize]  {$\pi ( 1)$};
\draw (41.11,79) node [anchor=north west][inner sep=0.75pt]  [font=\scriptsize]  {$\pi ( 2)$};
\draw (41.11,124) node [anchor=north west][inner sep=0.75pt]  [font=\scriptsize]  {$\pi ( 3)$};
\draw (135.11,36.4) node [anchor=north west][inner sep=0.75pt]  [font=\scriptsize]  {$\pi ( 4)$};
\draw (135.11,78.76) node [anchor=north west][inner sep=0.75pt]  [font=\scriptsize]  {$\pi ( 5)$};
\draw (135.11,123.76) node [anchor=north west][inner sep=0.75pt]  [font=\scriptsize]  {$\pi ( 6)$};
\draw (230.11,36.4) node [anchor=north west][inner sep=0.75pt]  [font=\scriptsize]  {$\pi ( 7)$};
\draw (230.11,78.76) node [anchor=north west][inner sep=0.75pt]  [font=\scriptsize]  {$\pi ( 8)$};
\draw (230.11,123.76) node [anchor=north west][inner sep=0.75pt]  [font=\scriptsize]  {$\pi ( 9)$};
\draw (40.4,156.8) node [anchor=north west][inner sep=0.75pt]    {$\mathcal{V}_{1}^{\pi }$};
\draw (133.6,156.8) node [anchor=north west][inner sep=0.75pt]    {$\mathcal{V}_{2}^{\pi }$};
\draw (230.8,156.8) node [anchor=north west][inner sep=0.75pt]    {$\mathcal{V}_{3}^{\pi }$};

\end{tikzpicture}
\caption{For the permutation $\pi$, the depicted cycle appears in the expansion of the trace of (omitting $\pi$ and $\Theta$) $\mathcal{L}_1^2 \mathcal{L}_2 \mathcal{L}_1 \mathcal{L}_2$, but not in the trace of $\mathcal{L}_1^3 \mathcal{L}_2^2$. Indeed, there must be 4 transitions between partition sets.}
\label{fig:cycle_representation_2}
\end{figure}
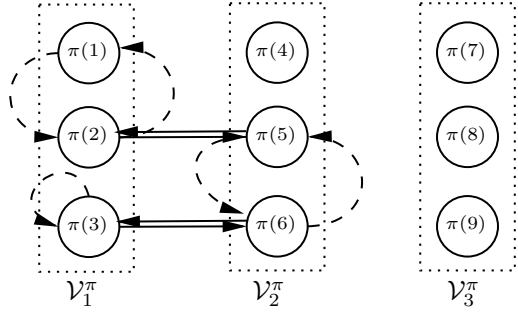
The order of the product elements in a multinomial product of matrices in \eqref{eqn:def_F_k} will determine the constraints of the cycles---the transition between partitions---that can appear from this product.
For a fixed permutation $\pi$, depending on the matrix multinomials in the expansion of \eqref{eqn:def_F_k_pi} and how `commutative' they are, a cycle $\gamma$ may not be in the trace of that product.
See Figure~\ref{fig:cycle_representation_2} for an example.
Alternatively, if the permutation $\pi$ is random, the more a matrix product is `mixed', the larger amount of constraints imposes to the cycles appearing in its trace, and so generally less cycles will appear.
We will thus need to track the exponents of the matrices in the expansion of $F(k)$ to determine how often a cycle appears for a (random) permutation.

\subsection{Tracking factors and their chance of occurring}
\label{secappendix:tracking}

We examine the noncommutative expansion of $F(k)$, and we track the exponents depending on the matrix indices.
We define the set $\mathfrak{L}(k)$ of possible indices for the $k$ different matrices appearing on each term in the expansion of \eqref{eqn:def_F_k_pi}. 
Concretely, the set of \textsf{exp}-indices is defined by
\begin{equation}
\mathfrak{L}(k) = [T]^{k}.
\end{equation}
and each $\ell = (\ell_1, \ldots, \ell_{k}) \in \mathfrak{L}(k)$ will track the matrix indices in the expansion of
\begin{equation}
\expectationBigWrt{\trBig{\prod_{j=1}^{k} \mathcal{L}_{\ell_j}^{\pi}(\Theta)}}{\Theta}.
\label{eqn:def_F_pi_intro}
\end{equation}
For $\ell \in \mathfrak{L}(k)$, we will denote the set of \emph{different} indices appearing in $\ell$ by $N(\ell) = \cup_{i=1}^{k} \{ \ell_i \}$, and $n(\ell) = |N(\ell)|$ its amount.
Recall \eqref{eqn:appendix_V_gamma}. 
Given $\ell \in \mathfrak{L}(k)$, for a fixed permutation $\pi \in S_m$, we will denote the set of factors in \eqref{eqn:def_F_pi_intro} as
\begin{align}
B(\ell, \pi) &= \Bigl\{ V_{\gamma} ~\Big|~  \text{ for some valid } \gamma,  V_{\gamma} \text{ appears in } \expectationbigWrt{\tr{\prod_{j=1}^{k} \mathcal{L}_{\ell_j}^{\pi}(\Theta)}}{\Theta} \Bigr\}, \quad \text{ and define } \nonumber \\
 B(\ell) &= \cup_{\pi \in S_m} B(\ell, \pi),
\label{eqn:def_B_set}
\end{align}
as well as the set
\begin{equation}
\mathcal{G}(k)= \bigsqcup_{\ell \in \mathfrak{L}(k)} \{ \ell, B(\ell) \}.
\label{eqn:def_mathcal_G}
\end{equation}
\begin{definition}
We will say that a factor $V$ \emph{appears} in $\Gamma = \{ \ell, B(\ell) \} \in \mathcal{G}(k)$ for a permutation $\pi \in S_m$, and denote so by $V \in B(\Gamma, \pi)$, if $V \in B(\ell, \pi)$.
We will also simply say that a factor $V$ appears in $\Gamma$ if $V \in B(\Gamma)$.
\label{def:L_index_set}
\end{definition}

We denote the set of different indices in $\Gamma = \{\ell, B(\ell)\} \in \mathcal{G}(k)$ and their number by $N(\Gamma) = N(\ell)$, and $n(\Gamma) = n(\ell)$, respectively.
Similarly, we will denote the vectors appearing in a factor $V \in B(\Gamma)$ by $N(V)$, and their number by $n(V)$.
If a factor $V$ appears in $\Gamma$, the probability that $V$ appears for a random permutation $\pi \in S_{m}$ is positive, and will depend on the constraints imposed by $\Gamma$, and its valid cycles.

With the notation of Definition~\ref{def:L_index_set}, if $V \in B(\Gamma, \pi)$ for some permutation $\pi$, each index in $N(V) \subset [m]$ must belong to only one of the partition sets from \eqref{eqn:def_partitions_permuted} with label in $N(\Gamma) \subset [T]$.
Moreover, in this case a cycle $\gamma$ such that $V_{\gamma} = V$ will satisfy the constraints that come from the order in the multiplication of the matrices $\{ L^{\pi}_{t}\}_{t \in [T]}$ associated to $\Gamma$, encoded by the \textsf{exp}-index of $\Gamma$. 
With this information, in the following Lemma~\ref{lem:constraints_cycle} we obtain necessary conditions for $\gamma$ to satisfy $V_{\gamma} \in B(\Gamma, \pi)$, which we will use to bound the chance that this happens when $\pi$ is random later on.
\begin{lemma}
For a permutation $\pi \in S_m$ and $\Gamma \in \mathcal{G}(k)$, if a cycle $\gamma$ satisfies $V_{\gamma} \in B(\Gamma, \pi)$, then the following conditions hold.
\begin{itemize}
\item[\emph{(\textsf{A})}] The cycle $\gamma$ restricted to each partition set of $\mathcal{V}^{\pi}$ in \eqref{eqn:def_partitions_permuted}, will be a disjoint union of cycles (or the emptyset). 
Only partitions with index in $N(\Gamma)$ may have nonempty intersection with $\gamma$, and any index in $N(V_{\gamma}) \cap \mathcal{V}_t^{\pi}$ for $t \in N(\Gamma)$ will appear as many times in the primal as in the dual arguments of the inner products of $V_{\gamma}$.
\item[\emph{(\textsf{B})}] The cycle $\gamma$ changes partitions in the graph associated to the permutation $\pi$ only between indices $i,j \in N(V_{\gamma})$ that satisfy $\pi^{-1}(i) \equiv \pi^{-1}(j) \mod r$, that is, if $i \in \mathcal{V}_t^{\pi}$, $j \in \mathcal{V}_{t^{\prime}}^{\pi}$, both indices appear in the same relative position within their respective partitions $t$, and $t^{\prime}$.
\end{itemize}
\label{lem:constraints_cycle}
\end{lemma}
\begin{proof}
For a fixed $\pi$ and $\Gamma \in \mathcal{G}(k)$, condition (\textsf{A}) is implied by the cycle being valid, and not vanishing in $F(k, \pi)$.
Thus, in each partition set $\mathcal{V}_t^{\pi}$ the condition must also hold.

For condition (\textsf{B}), note that transitions between partitions correspond to changes of matrix index $l \in N(\Gamma)$ as they appear in the noncommutative product corresponding to $\prod_{j=1}^{k} \mathcal{L}_{\ell_j}^{\pi}(\Theta)$ for $\Gamma = \{ \ell, N(\ell)\}$. 
Indices of $\gamma$ will change partitions in the graph associated to $\pi$ if their corresponding vectors belong to different matrices that are consecutive in the matrix product. 
For two indices $i,j \in N(V_{\gamma})$ that change between partitions in the graph, from the multiplication rule of matrices, the locations of these indices in the matrices share row and column number. 
This implies that the vectors corresponding to these indices are in the same position within their respective partitions.
If $i= \pi(d_i)$, and $j=\pi(d_j)$, for some $d_i, d_j \in [m]$, then $d_i$ must be the same as $d_j$ modulo $r$.
\end{proof}

The rules in Lemma~\ref{lem:constraints_cycle} yield necessary conditions for $\gamma$ to satisfy $V_{\gamma} \in B(\Gamma, \pi)$.
For example, if a product $\langle v_{i_1}, v_{i_2} \rangle \langle v_{i_3}, v_{i_4} \rangle$ appears in $V$ and $i_2$ and $i_3$ belong to different partition sets in the graph of $\gamma$, then the factor $\langle v_{i_4}, v_{i_1}\rangle$ cannot appear in $V_{\gamma}$ by condition (\textsf{A}).
In the following lemmas, we use the previous necessary conditions to bound the probability that a factor $V$ appears for a random permutation $\pi$ and $\Gamma \in \mathcal{G}(k)$.
We take especial interest in the contributions by the diagonal terms of the matrices $\mathcal{L}_t^{\pi}(\Theta)$---factors $V$ that contain only norms---in the expansion of $F(k)$; they will become leading terms later on.
\begin{lemma}
Suppose that $V$ appears in $\Gamma$ and is such that $n(V) = n(\Gamma)$. 
Then $V$ is composed only of diagonal elements of $\{\mathcal{L}_t^{\pi}(\Theta) \}_{i\in [T]}$ and thus $V = \prod_{i \in N(V)} |v_i|^{2k_i}$ where $\mathbf{k} = (k_1, \ldots, k_{n(\Gamma)})$ are the multiplicities of the vectors with indices in $N(V)$.
\label{lem:diagonal_elements}
\end{lemma}
\begin{proof}
If $n(V) = n(\Gamma)$ for $\Gamma = \{\ell, N(\ell)\}$, the product $\prod_{i=1}^{k} \mathcal{L}_{\ell_i}^{\pi}(\Theta)$ that has $n(\Gamma)$ distinct matrices---not counting multiplicies---can only have a cycle $\gamma$ with $n(V)$ different indices if and only if for each index in $N(\Gamma)$, there is exactly one index in $N(V_{\gamma})$. 
This happens only when diagonal elements of $\mathcal{L}_t^{\pi}(\Theta)$ appear in $V_{\gamma}$. 
Therefore, the identity $V = V_{\gamma} = \prod_{i \in N(V)} |v_i|^{2k_i}$ must hold for some $\gamma$, where $\mathbf{k} = (k_1, \ldots, k_{n(\Gamma)})$ are the multiplicities of the indices $t \in N(\Gamma)$ in $\ell$, as well as of $N(V)$.
\end{proof}

\begin{lemma}\emph{(Adapted from \cite[Sec. 1.9]{stanley2011enumerative})}
Let $d \geq n$. The number of ways to partition $d$ labeled items into $n$ labeled classes such that no class is empty is given by
\begin{equation}
n! \stirling{d}{n} = \sum_{j=0}^{n-1} (-1)^{j} (n - j)^{d} \binom{d}{j},
\end{equation}
where $\stirling{d}{n}$ is the Stirling number of the second kind that satisfies
$\stirling{d}{d} = \stirling{d}{1} = 1$, and $\stirling{d}{n} \leq \frac{n^{d}}{n!}$.
\label{lem:stirling}  
\end{lemma}

A cycle $\gamma$ that satisfies condition (\textsf{A}), on each partition it is decomposed in smaller disjoint cycles $\delta_1, \ldots, \delta_{s}$ for some $s \geq 1$.
We can choose the decomposition such that each cycle cannot be further decomposed, and the indices of each subcycle belong to a single partition.
We denote the maximal number of disjoint cycles that decompose $\gamma$ by $d(\gamma)$.
This maximal decomposition is independent of partition set $\mathcal{V}^{\pi}$, and satisfies
\begin{equation}
V_{\gamma} = V_{\delta_1} \cdots V_{d(\gamma)}.
\label{eqn:maximal_decomposition}
\end{equation}
Note that the subcycles $\delta_1, \ldots, \delta_{d(\gamma)}$ do not keep track of the transitions between the partitions of $\gamma$; this is instead in $\Gamma$. 
Crucially, for any $\Gamma \in \mathcal{G}(k)$ such that condition (\textsf{A}) holds for some $\gamma$, we must then have $n(\Gamma) \leq d(\gamma) \leq n(V_{\gamma})$.

The number of maximal subcycles allows us to bound in the next lemma how many permutations $\pi$ allow the expansion term in \eqref{eqn:def_F_k_pi} corresponding to $\Gamma$ to have a factor $V$ appearing.
For the following Lemma, if $\mathbf{k}(\Gamma) = (k_1, \ldots, k_{n(\Gamma)})$ are the multiplicities of the \textsf{exp}-index $\ell$ of $\Gamma$, we let $z_{1}$ be the different number of times that $k_1$ appears in $\mathbf{k}(\Gamma)$, $z_2$ the number of times a different multiplicity than $k_1$ appears in $\mathbf{k}(\Gamma)$, and so on.
Denote the number different multiplicities of $\mathbf{k}(\Gamma)$ in this way by $Z(\Gamma)$.
The vector 
\begin{equation}
\mathbf{z}(\Gamma) = (z_1, \ldots, z_{Z(\Gamma)})
\label{eqn:def_z_Gamma}
\end{equation}
will contain the multiplicities of $\mathbf{k}(\Gamma)$, and
\begin{align}
\sum_{i= 1}^{Z(\Gamma)} z_i = n(\Gamma).
\end{align}
With slight abuse of notation, the vectors $\mathbf{k}(\Gamma)$, or $\mathbf{z}(\Gamma)$ will also be considered as lists whenever used in a multinomial number, e.g, $\binom{k}{\mathbf{k}(\Gamma)}$.
\begin{lemma}
Suppose that $\Gamma \in \mathcal{G}(k)$, and let $V \in B(\Gamma)$. 
If $\pi \sim \mathrm{Unif}(S_{m})$, then

\noindent
\emph{(i)} If $n(V) > n(\Gamma)$,
\begin{equation}
C_{V, \Gamma} = \probabilitybig{V \in B(\Gamma, \pi)} \leq \frac{n(\Gamma)! \stirling{d(V)}{n(\Gamma)} r^{n(V) - n(\Gamma) + 1}}{\binom{Tr}{n(V)} n(V)!}.
\end{equation}
\emph{(ii)} If $n(V) = n(\Gamma)$, let $\mathbf{z}(\Gamma)$ be the multiplicities of $\mathbf{k}(\Gamma)$ defined in \eqref{eqn:def_z_Gamma}, then
\begin{equation}
C_{V, \Gamma} = \probabilitybig{V \in B(\Gamma, \pi)} = \frac{r}{\binom{Tr}{n(V)} \binom{n(\Gamma)}{\mathbf{z}(\Gamma)}}.
\end{equation}
\label{lem:bound_prob_cycle}
\end{lemma}
\begin{proof}

We count first the number of ways that permutations of $[m]$ can rearrange the indices of $V$.
Out of $m = Tr$ possible total index locations, $n(V)$ must be chosen.
By symmetry, there are 
\begin{equation}
\frac{(Tr)!}{(Tr - n(V))!} = \binom{Tr}{n(V)} n(V)!
\label{eqn:proof_lemma_bound_prob_cycle_5}
\end{equation}
possible ordered choices of locations for the $n(V)$ indices that can be attained by a permutation $\pi \in S_m$.
We bound how many of these permutations allow the existance of a cycle $\delta$ compatible with $\Gamma$ and the permutation $\pi$, such that $V_{\delta} = V$.
We show first part (ii).

\textbf{Case (ii)} If $n(V) = n(\Gamma)$, by Lemma~\ref{lem:diagonal_elements}, we only need to look at the diagonal elements of the matrices $\mathcal{L}^{\pi}_t$ for $t \in N(\Gamma)$ and see where the indices of $V$ are sent by $\pi$.
From the assumption $V \in B(\Gamma)$, by Lemma~\ref{lem:diagonal_elements} we must have $V = \prod_{i \in N(V)} |v_i|^{2k_i}$, where $k_i$ coincide with the multiplicities of the $\ell$ index of $\Gamma$.
For any permutation $\pi$ such that $V \in B(\Gamma, \pi)$, each $\ket{v_i}$ with $i \in N(V)$ must belong to a unique partition $\mathcal{V}_t^{\pi}$ with $t \in N(\Gamma)$, and also satisfy condition (\textsf{B}) from Lemma~\ref{lem:constraints_cycle}.
If there exists a cycle $\delta$ such that $V_{\delta} = V$, then all indices of $V$ must satisfy condition (\textsf{B}) pairwise since they are all adjacent to each other.
There are $r$ different values modulo $r$ such that condition (\textsf{B}) can be satisfied simultaneously for all $i \in N(V)$.
Furthermore, for each of the choices, we can try to assign differently the indices of $N(\gamma)$ within the partitions available in $N(\Gamma)$. 
However, the multiplicities of $\Gamma$ have to match the multiplicity of the vectors in $V$ for the factor to appear for that index assignment.
Only if multiplicities coincide will the factor $V$ appear for the permutation.
For $i,j \in N(\gamma)$, we can only exchange $i$ with the position of $j$, if $k_i = k_j$.
Let $\mathbf{z}(\Gamma)$ be the vector of multiplicities of $\{k_i\}_{i \in N(\Gamma)}$.
From the previous discussion, for a fixed value modulo $r$, there are exactly $\prod_{i=1}^{z(\Gamma)} z_i!$ ways to permute indices in $N(V)$ within the different partitions corresponding to $N(\Gamma)$ such that the product $V$ appears for that permutation.
There are at most $r \prod_{i=1}^{z(\Gamma)} z_i!$ permutations that satisfy $V \in B(\Gamma, \pi)$.
Conversely, any permutation satisfying these constraints will also imply that $V \in B(\Gamma, \pi)$. 
We conclude
\begin{equation}
\probabilitybig{V \in B(\Gamma, \pi)} = \frac{r \prod_{i=1}^{z(\Gamma)} z_i! }{\binom{Tr}{n(V)} n(V)!} = \frac{r}{\binom{Tr}{n(V)} \binom{n(\Gamma)}{\mathbf{z}(\Gamma)}}.
\end{equation}

\textbf{Case (i)} We assume that $n(V) > n(\Gamma)$.
To upper bound the amount of permutations that allow  $V \in B(\Gamma, \pi)$, we consider cycles that satisfy the necessary conditions (\textsf{A}) and (\textsf{B}) from Lemma~\ref{lem:constraints_cycle}.
Let $\Gamma =\{\ell, B(\ell)\} \in \mathcal{G}(k)$, and $\gamma$ a cycle satisfying $V = V_{\gamma}$.
The chance that $V \in B(\Gamma, \pi)$ is at most the chance that there exists $\gamma$ such that $V = V_{\gamma}$ and $\gamma$ satisfies the necessary conditions (\textsf{A}) and (\textsf{B}) in Lemma~\ref{lem:constraints_cycle} for $\pi$ and $\Gamma$,
\begin{align}
\probability{V \in B(\Gamma, \pi)} &\leq \probability{~\exists \gamma \text{ s.t. } V_{\gamma}=V, \gamma \text{ satisfies (\textsf{A}) and (\textsf{B}) for } \pi \text{ and } \Gamma  } %
\label{eqn:proof_lemma_bound_prob_cycle}
\end{align}
We count how many cycles $\gamma$ may satisfy all conditions simultaneously.
For a permutation $\pi$, if a cycle $\gamma$ exists that satisfies $V = V_{\gamma}$, for each index $t \in N(\Gamma)$ corresponding to a partition set $\mathcal{V}_t^{\pi}$, there exist at least one constraint for the position of the indices of $N(V) \cap \mathcal{V}_t^{\pi}$ when the cycle $\gamma$ either arrives from, or leaves to another partition.
This constraint is imposed by the fact that the positions must be up to $\pi$ equal modulo $r$---see Lemma~\ref{lem:constraints_cycle}.

Denote by $\gamma|_{\mathcal{V}_{t}^{\pi}}$ for each $t \in N(\Gamma)$ the maximal set of subcycles of $\gamma$ that are contained in the partition set $\mathcal{V}_{t}(\pi)$ .
Since there are at least two transitions---arrival and departure---of $\gamma$ for each partition set, we show that there are at least $n(\Gamma) - 1$ indices of $V$ whose positions in the associated graph are fixed by the positions of other indices in $N(V)$ in the graph. 
Indeed, for each $t \in N(\Gamma)$ at least one index $u_t \in N(V) \cap \mathcal{V}_{t}^{\pi}$ must be adjacent in the graph of $\gamma$ to at least another index $u_{t^{\prime}}$ belonging to another partition $t^{\prime} \in N(\Gamma)$. 
If this was not the case, then for at least one index $t \in N(\Gamma)$, $\gamma$ would arrive to the partition corresponding to $t$, but not leave, which would contradict the fact that $\gamma$ is a cycle.
A connected graph with the $n(\Gamma)$ partitions is generated by $\gamma$ tracking arrivals and departures of partition sets.
It is well-known that the the connected graph with minimal amount of edges that we can construct with $n(\Gamma)$ vertices has $n(\Gamma) - 1$ edges.
Hence, there will be at least $n(\Gamma) - 1$ indices in $N(V)$ whose positions are determined by other indices of $N(V)$ whenever there is a cycle $\gamma$ such that $V = V_{\gamma}$.

We now upper bound how many permutations $\pi$ allow for a cycle $\gamma$ to exist with $V = V_{\gamma}$ satisfying conditions (\textsf{A}), and (\textsf{B}).
Conditional on $V = V_{\gamma}$, consider the maximal decomposition of $V$ by disjoint subcycles described in \eqref{eqn:maximal_decomposition}. 
Let $d(V)$ be the number of disjoint subcycles for this decomposition and denote the set of subcycles by $D(V) = \{ \delta_1, \ldots, \delta_{d(V)}\}$. 
Condition (\textsf{A}) requires that the indices of any subcycle $\delta_i$, $i \in [d(V)]$ belong to the same partition.
Any $\gamma$ satisfying $V = V_{\gamma}$, and (\textsf{A}) will have these $d(V)$ labeled maximal subcycles distributed among $n(\Gamma)$ labeled partition sets, with at least one subcycle per partition set.
By Lemma~\ref{lem:stirling}, the different number of ways to distribute the subcycles is
\begin{equation}
n(\Gamma)! \stirling{d(V)}{n(\Gamma)} = \sum_{i=0}^{n(\Gamma)-1} (-1)^{j} (n(\Gamma) - j)^{d(V)} \binom{d(V)}{j},
\end{equation}
where $\stirling{d(V)}{n(\Gamma)}$ is the Stirling number of second kind.

For each fixed distribution of the subcycles, each index $i \in N(V)$ within a partition $t \in N(\Gamma)$ can be assigned to at most $r$ positions in its corresponding partition.
However, from the previous discussion, in order to generate a connected cycle $\gamma$ such that $V_{\gamma} = V$, there are at least $n(\Gamma) - 1$  indices of $V$ whose locations are determined by other indices in $N(V)$.
Therefore, there are at most $r^{n(V) - n(\Gamma) + 1}$ possible ways to distribute the indices for each given subcycle assignment, and satisfy (\textsf{B}).
Combining the previous bounds, the number of different ways to allow a cycle $\gamma$ to exists with $V = V_{\gamma}$, while also satisfying (\textsf{B}) and (\textsf{A}) for $\Gamma$ is at most 
\begin{equation}
n(\Gamma)! \stirling{d(V)}{n(\Gamma)} r^{n(V) - n(\Gamma) + 1}.
\label{eqn:proof_lemma_bound_prob_cycle_3}
\end{equation}
Combining \eqref{eqn:proof_lemma_bound_prob_cycle_3} with \eqref{eqn:proof_lemma_bound_prob_cycle_5}, we obtain that
\begin{equation}
\probability{V \in B(\Gamma, \pi)} \leq \frac{n(\Gamma)! \stirling{d(V)}{n(\Gamma)}  r^{n(V) - n(\Gamma) + 1}}{\binom{Tr}{n(V)} n(V)!}.
\label{eqn:proof_lemma_bound_prob_cycle_2}
\end{equation}
\end{proof}

Lemma~\ref{lem:bound_prob_cycle} bounds the probability $C_{V, \Gamma}$ of a factor $V$ appearing for $\Gamma \in \mathcal{G}(k)$ and a random permutation.
In the next section, we examine the symmetries of $\mathcal{G}(k)$ that will simplify the computations later, and are key in obtaining a leading order of $F(k)$ independent of $T$.

\subsection{Characterizing symmetries in the expansion}
\label{secappendix:symmetries}
The set $\mathcal{G}(k)$ possesses several additional symmetries that we can exploit by using groups.
There are namely two group actions of $S_k$ and $S_{T}$ on $\mathcal{G}(k)$.
Both actions will allow us to decouple the combinatorial problem of counting \textsf{exp}-indices to counting only factors.

\paragraph{First action by $S_k$.} We define the action $\mathfrak{h}: S_{k} \to \mathrm{Aut}(\mathcal{G}(k))$ that permutes the location of the indices of $\Gamma = \{\ell, B(\ell)\} \in \mathcal{G}(k)$, that is, for $g \in S_{k}$,
\begin{align}
g(\ell_1, \ldots, \ell_{k}) &= (\ell_{g(1)}, \ldots, \ell_{g(k)}) \nonumber \\
\mathfrak{h}[g] \bigl( \{\ell, B(\ell)\} \bigr) &= \{g(\ell), B(g(\ell)) \}.
\end{align}
This group action tracks how many $\Gamma^{\prime} \in \mathcal{G}(k)$ in the expansion of \eqref{eqn:def_F_k} have similar indices to those of $\Gamma$, up to reordering.
For example, they track if
\begin{equation}
\expectationBigWrt{\trBig{\bigl(\mathcal{L}_{1}^{\pi}(\Theta)\bigr)^2 \bigl(\mathcal{L}_{2}^{\pi}(\Theta)\bigr)^2}}{\Theta, \pi} \text{ and } \expectationBigWrt{\trBig{\mathcal{L}_{1}^{\pi}(\Theta) \mathcal{L}_{2}^{\pi}(\Theta) \mathcal{L}_{1}^{\pi}(\Theta) \mathcal{L}_{2}^{\pi}(\Theta)}}{\Theta, \pi},
\label{eqn:example_S_T}
\end{equation}
may share some of the factors.
We define the set of orbits of this action on $\mathcal{G}(k)$ by
\begin{equation}
\bar{\mathcal{G}}(k) = \frac{\mathcal{G}(k)}{S_{k}}. 
\label{eqn:def_mathcal_G_bar}
\end{equation}
We bound the number of elements of an orbit $\bar{\Gamma} \in \bar{\mathcal{G}}(k)$.
Given $\Gamma = \{ \ell, B(\ell)\} \in \mathcal{G}(k)$ we denote the multiplicities of $\ell$ by $\mathbf{k}(\Gamma) = (k_1, \ldots, k_{n(\Gamma)})$.
A group element $g \in S_{k}$ will permute the entries of the vector $\ell$ but not their multiplicities.
The orbit of $\Gamma \in \mathcal{G}(k)$ thus contains at most all possible different ways of choosing out of $k$ locations in the \textsf{exp}-indices $\ell$, $n(\Gamma)$ labels with their respective multiplicities in $\mathbf{k}(\Gamma)$.
If $\bar{\Gamma}$ denotes the orbit of $\Gamma$ under this action, the number of elements is given by the multinomial number
\begin{equation}
|\bar{\Gamma}| = \binom{k}{k_1, \ldots, k_{n(\Gamma)}} = \binom{k}{\mathbf{k}(\Gamma)}.
\label{eqn:proof_variance_bounded_combinatorial_num}
\end{equation}
The number, and index labels of $\Gamma \in \mathcal{G}(k)$ are invariant under $\mathfrak{h}$.
We can define $n(\bar{\Gamma}) = n(\Gamma)$, and $N(\bar{\Gamma})) = N(\Gamma)$ for any representative $\Gamma \in \bar{\Gamma}$.

Note that the action of $S_{k}$ is in general not well-defined for factors of $\Gamma$ since $B(\ell) \neq B(g(\ell))$, as the following example shows.
\begin{example}
For $k=4$, let $\ell_1 = (1,1,2,2)$, and $\ell_2=(1,2,1,2)$---in \eqref{eqn:example_S_T}---be \textsf{exp}-indices in the same equivalence class under action $\mathfrak{h}$, and consider the following factor corresponding to some cycle $\gamma$
\begin{equation}
V_{\gamma} = \langle v_1, v_2 \rangle \langle v_2, v_1 \rangle \langle v_3, v_4 \rangle  \langle v_4, v_3 \rangle \in B(\ell_1),
\end{equation}
where $\{1,2\} \subset \mathcal{V}_1$,  $\{3,4\} \subset \mathcal{V}_2$ belong to different partition sets.
In order for $V_{\gamma}$ to also appear in $B(\ell_2)$, only the following product or its conjugate could appear in the expansion corresponding to $B(\ell_2)$ according to Definition~\ref{def:cycle},
\begin{equation}
\langle v_1, v_2 \rangle \langle v_3, v_4 \rangle \langle v_2, v_1 \rangle  \langle v_4, v_3 \rangle.
\end{equation}
Here, the indices $2$ and $3$ satisfy constrain \textsf{(B)} of Lemma~\ref{lem:constraints_cycle}, and appear in the same position for both partitions.
However, the same should occur with indices $2$ and $4$. 
The constraints \textsf{(B)} cannot be satisfied and so $V_{\gamma} \notin B(\ell_2)$.
\label{example:action_not_well_defined}
\end{example}

Nonetheless, the action of $\mathfrak{h}$ will be well-defined for some factors; see Remark~\ref{remark:action_commutes_diagonal_elements} below. We will use this symmetry to estimate their contribution to $F(k)$, which will become the leading term.
\begin{remark}
For a factor $V \in B(\ell)$ composed only of diagonal elements, that is, $V$ is a product of norms of vectors $\ket{v_i}$, the action $\mathfrak{h}$ is well-defined. 
Indeed, the diagonal parts of the matrices $\{ \mathcal{L}_{t}^{\pi}(\Theta)\}_{t \in [T]}$ commute for any \textsf{exp}-index, that is, $V \in B(g(\ell))$ for any $g \in S_{k}$.
\label{remark:action_commutes_diagonal_elements}
\end{remark}

\paragraph{Second action by $S_T$.} We can also define an action $\mathfrak{g}: S_{T} \to \mathrm{Aut}(\bar{\mathcal{G}}(k))$. 
A permutation $h \in S_{T}$ acts on $\ell \in [T]^{k}$ by permuting its indices
\begin{equation}
h(\bar{\ell}_1, \ldots, \bar{\ell}_{k}) = (h(\ell_1), \ldots, h(\ell_{k})).
\end{equation}
The action on a representative $[\{\bar{\ell}, B(\bar{\ell}) \}] \in \bar{\Gamma}$ is
\begin{align}
\mathfrak{g}[h]\bigl( [\{\bar{\ell}, B(\bar{\ell}) \}] \bigr) = [\{h(\bar{\ell}), B(h(\bar{\ell})) \}].
\end{align}
Differently to $\mathfrak{h}$, the action $\mathfrak{g}$ is well defined for factors due to symmetry, e.g., we expect the same factors $V$ occurring in 
\begin{equation}
\expectationBigWrt{\trBig{\bigl(\mathcal{L}_{1}^{\pi}(\Theta)\bigr)^2 \bigl(\mathcal{L}_{2}^{\pi}(\Theta)\bigr)^2}}{\Theta, \pi} \text{ and } \expectationBigWrt{\trBig{\bigl(\mathcal{L}_{3}^{\pi}(\Theta)\bigr)^2 \bigl(\mathcal{L}_{4}^{\pi}(\Theta)\bigr)^2}}{\Theta, \pi}.
\label{eqn:def_F_pi_intro}
\end{equation}
Indeed, if a factor $V \in B(\Gamma, \pi)$, the effect of permuting the indices of $\Gamma$ with $h$ for a fixed $\pi$ is equivalent to permuting the partition sets with label $t \in N(\Gamma)$, while leaving $V$ and $\pi$ fixed.
We can then compensate this change by using another permutation $\pi^{\prime} \in S_{m}$ such that $V \in B(h(\Gamma), \pi^{\prime})$ instead. 
From the definition of $B(\ell)$ in \eqref{eqn:def_B_set} we obtain that the factors appearing for a representative of $\bar{\Gamma}$ are invariant under the action $\mathfrak{g}$.
We denote the set of orbits under this action by
\begin{equation}
G(k) = \frac{\bar{\mathcal{G}}(k)}{S_{T}}
\label{eqn:def_mathcal_G}
\end{equation}

Let $\Delta \in G(k)$ be the orbit of $\bar{\Gamma}$ under $\mathfrak{g}$.
The number of different \textrm{exp}-indices of any of its representatives in $\bar{\Gamma}$ is $n(\Delta) = n(\bar{\Gamma})$, and their multiplicities are $\mathbf{k}(\Delta) = \mathbf{\bar{\Gamma}}$.
We will also use the notation $\mathbf{z}(\Delta)$ to denote the multiplicities of $\mathbf{k}(\Delta)$ similarly to \eqref{eqn:def_z_Gamma}.

We compute the size of $\Delta$. 
Recall that $\bar{\Gamma}$ contains all permutations of the \textsf{exp}-index $\ell$ of a representative $\Gamma \in \bar{\Gamma}$. 
We find a representative $\Gamma$ such that its \textrm{exp}-indices are equal adjacent to each other
\begin{equation}
(\overbrace{\ell_1, \cdots, \ell_1}^{k_1 \text{times}}, \cdots, \overbrace{\ell_{n(\Gamma)}, \ldots, \ell_{n(\Gamma)}}^{k_{n(\Gamma)} \text{times}})
\end{equation}
In this manner, we can identify $\bar{\Gamma}$ with a commutative multinomial of order $k$ in $[T]$ variables
\begin{equation}
\prod_{i=1}^{n(\Gamma)} (x_{\ell_i})^{k_i}.
\label{eqn:counting_polys}
\end{equation} 
A permutation $h \in S_{T}$ sends the labels $i \in N(\Gamma)$ to any other different labels in $[T]$ while leaving the multiplicities invariant.
If all multiplicities are different, then the amount of different multinomials under this action will be equivalent to the number of injective maps from the set $N(\Gamma)$ to $[T]$, which is given by the multinomial coefficient
\begin{equation}
\frac{T!}{(T - n(\Gamma))!} = \binom{T}{n(\Gamma)} n(\Gamma)! = \binom{T}{n(\Delta)} n(\Delta)!.
\label{eqn:order_orbit_delta_distinct_multiplicities}
\end{equation}
However, if some of the multiplicities in $\mathbf{k}(\Gamma)$ are equal, some of the previous maps for the labels in $N(\Gamma)$ will be sent to the same multinomial.
Let $\mathbf{z}(\Gamma)$ be the vector of size $n(\Gamma)$ with the multiplicities of $\mathbf{k}(\Gamma)$.
This vector satisfies
\begin{equation}
n(\Gamma) = \sum_{i=1}^{Z(\Gamma)} z_i.
\end{equation}
For a fixed multinomial of the type \eqref{eqn:counting_polys} with multiplicities $\mathbf{k}(\Gamma)$ there are $\prod_{i=1}^{Z(\Gamma)} z_i!$ different relabeling of its indices that leave the multinomial invariant.
Therefore, the number of commutative monomials of order $k$ in $[T]$ variables with multiplicities $\mathbf{k}(\Gamma)$ is given by
\begin{equation}
|\Delta| = \binom{T}{n(\Gamma)} \binom{n(\Gamma)}{\mathbf{z}(\Gamma)} = \binom{T}{n(\Delta)} \binom{n(\Delta)}{\mathbf{z}(\Delta)}.
\label{eqn:order_orbit_delta}
\end{equation}

\subsection{Combinatorial inequalities}
\label{secappendix:combinatorial}
We have examined \textsf{exp}-indices and their symmetries. 
We now define sets that will help us with the counting of factors. 
\begin{align}
\mathcal{M}(k, l)=\Bigl\{ (k_1, \ldots, k_{m}) \in \bigl( [k] \cup \{0\} \bigr)^{\times m} \Big| \sum_{i=1}^{m} k_i = k, k_i > 0 \text{ for \emph{exactly} } l \text{ indices } i \in [m] \Bigr\}.
\label{eqn:def_commutative_polinomials_boundeds}
\end{align}
The set $\mathcal{M}(k, l)$ can be identified with the set of different multinomials of order $k$ with $l$ variables out of $m$.
The set of different multinomials of order $k$ with $m$ variables $\mathcal{M}(k)$ is then one-to-one with
\begin{align}
\mathcal{M}(k) &= \bigsqcup_{\substack{1 \leq l \leq k}} \mathcal{M}(k, l). \label{eqn:def_commutative_polinomials_all}
\end{align}
For a vector $\mathbf{k} \in \mathcal{M}(k)$ (of dimension $m$), we denote its number of nonzero entries by $n(\mathbf{k})$.
\begin{definition}
Let a vector $\mathbf{k} \in \mathcal{M}(k)$, we denote its frequency by $f(\mathbf{k}) = (k_{[1]}, \ldots, k_{[k]})$, where $k_{[i]}$ is the $i$th largest entry of $\mathbf{k}$, and we complete with zeros if necessary.
Similarly, for any $\Gamma  \in \bar{\Gamma} \in \Delta \in G(k)$, we define $f(\Delta) = f(\bar{\Gamma}) = f(\Gamma) = (\mathbf{k}(\Gamma), 0_{k - n(\Gamma)})$, which is the vector of multiplicities of $\Gamma$ in order with added zeros if necessary.
We denote the set of frequencies with at most $k$ different items by $\mathfrak{F}(k) = f(\mathcal{M}(k))$.
\label{def:frequency}
\end{definition}
Note that any $f \in \mathfrak{F}(k)$ can be represented by a decreasing vector $(f_1, \ldots, f_k)$, where we complete with zeros if necessary. 
If $\Delta \in \mathcal{G}(k)$ and $n(\Delta)=l$, we can define the subset of \eqref{eqn:def_commutative_polinomials_boundeds} that have the same frequencies as $\Delta$.
\begin{equation}
\mathcal{M}(k, l)[\Delta]=\Bigl\{ \mathbf{k} \in \mathcal{M}(k, l) ~\Big|~ f(\mathbf{k}) = f(\Delta) \Bigr\}.
\label{eqn:def_commutative_polinomials_boundeds_delta}
\end{equation}
The set $\mathcal{M}(k, l)[\Delta]$ is one-to-one with the set of multinomials of order $k$ with $l$ variables out of $m$, such that their vectors of multiplicities is $\mathbf{k}(\Delta)$ up to reordering.
Crucially, the frequencies of different elements in $G(k)$ are also different as shown in the following lemma.
\begin{lemma}
Let $\Delta_1, \Delta_2 \in G(k)$.
Then $f(\Delta_1) = f(\Delta_2)$, or equivalently $\mathbf{k}(\Delta_1)$ is $\mathbf{k}(\Delta_2)$ up to reordering, if and only if $\Delta_1 = \Delta_2$.
\label{lem:disjoint_union_frequencies}
\end{lemma}
\begin{proof}
If $f(\Delta_1) = f(\Delta_2)$, then $\Delta_1$ has a representative $\Gamma_1 = \{\ell_1, B(\ell_1)\}$ that satisfies $i_1 \in \ell_1$ $k_1$ times, $i_2 \in \ell_1$ $k_2$ times, \ldots, and  $i_{n(\Gamma_1)} \in \ell_1$ $k_{n(\Gamma_1)}$ times for some $i_1, \ldots, i_{n(\Gamma_1)} \in N(\Gamma_1)$.
Similarly for $\Delta_2$ with a representative $\Gamma_2 = \{\ell_2, B(\ell_2)\}$ that satisfies $j_1 \in \ell_2$ $k_1$ times, $j_2 \in \ell_2$ $k_2$ times, \ldots, and  $j_{n(\Gamma_1)} \in \ell_2$ $k_{n(\Gamma_1)}$ times for some $j_1, \ldots, j_{n(\Gamma_2)} \in N(\Gamma_2)$.
Since $n(\Delta_1) = n(\Gamma_1) = n(\Gamma_2) = n(\Delta_2)$, there exists a permutation $h \in S_{T}$ that takes $i_1, \ldots, i_{n(\Gamma_1)}$ to $j_1, \ldots, j_{n(\Gamma_2)}$. 
In particular, $\Delta_1$, and $\Delta_2$ have the same representatives, and so they must be equal.
\end{proof}

If we consider the disjoint union of \eqref{eqn:def_commutative_polinomials_boundeds_delta} over all  $\Delta \in G(k)$, we recover \eqref{eqn:def_commutative_polinomials_boundeds}
\begin{equation}
\mathcal{M}(k, l) = \bigsqcup_{\substack{\Delta \in G(k)\\ n(\Delta)=l}} \mathcal{M}(k, l)[\Delta].
\label{eqn:def_M_k_l}
\end{equation}
We introduce a partial order over frequencies in $\mathfrak{F}(k)$.
\begin{definition}
A frequency $f=(f_1, \ldots, f_k)$ is dominated by $g=(g_1, \ldots, g_k)$ in $\mathfrak{F}(k)$, denoted by $g \prec f$, if for a partition $\{ \mathcal{A}_i\}_{i=1}^{l}$ of $[k]$ with some $l \geq 1$, we have for all $j \in [k]$ such that $f_j \neq 0$,
\begin{align}
f_j  = \sum_{i \in \mathcal{A}_j} g_i.
\end{align}
\label{def:domination}
\end{definition}
\noindent
For $\Delta \in G(k)$ such that $d > n(\Delta)$, we define
\begin{align}
\mathcal{M}(k, d)[\Delta] =\Bigl\{ \mathbf{k} \in \mathcal{M}(k,d) ~\Big|~ f(\mathbf{k}) \prec f(\Delta) \Bigr\},
\label{eqn:def_commutative_polinomials_boundeds_delta_larger}
\end{align}
We use the previous definitions to simplify the computations with the following results. 
\begin{lemma}
In the notation of this section, let $\Gamma \in \bar{\Gamma} \in \Delta \in G(k)$ be such that $n(\Gamma) = l$.
\begin{equation}
\sum_{\substack{V \in B(\Gamma)\\ n(V) = l}} C_{V, \Gamma} V = \frac{r}{\binom{Tr}{n(\Delta)} \binom{n(\Delta)}{\mathbf{z}(\Delta)}} \sum_{\mathbf{k} \in \mathcal{M}(k,l)[\Delta]} \prod_{i=1}^{m} |v_{i}|^{2k_i},
\label{eqn:counting_cycles_to_polynomials_same_length}
\end{equation}
and
\begin{equation}
\sum_{\substack{V \in B(\Gamma)\\ n(V) > l}} C_{V, \Gamma} V \leq 
\sum_{d=l+1}^{k} \frac{n(\Gamma)! \stirling{d}{n(\Gamma)}  r^{d - n(\Gamma) + 1}}{\binom{Tr}{d} d!} \sum_{\substack{\mathbf{g} \in \mathcal{M}(k,d)[\Delta]}} \binom{k}{\mathbf{g}} \prod_{i=1}^{m} |v_i|^{2g_i}.
\label{eqn:counting_cycles_to_polynomials}
\end{equation}
\label{lem:counting_cycles}
\end{lemma}
\begin{proof}
We show \eqref{eqn:counting_cycles_to_polynomials_same_length} first.
Since we consider the case $n(V)= n(\Gamma)=l$, by Lemma~\ref{lem:diagonal_elements}, $V$ is product of norms of the vectors $\{ \ket{v_i}\}_{i \in [m]}$, and there will be exactly $n(\Gamma)$ different norms in the product, up to multiplicities.
The indices $t \in N(\Gamma)$ are one to one with indices in $N(V)$, and the multiplicites will be $\mathbf{k}(\Gamma) = (k_1, \ldots, k_{l})$.
Thus, we can find as many factors $V \in B(\Gamma)$ with exactly $n(V) = n(\Gamma)$ different norms as there are multinomials of the type
\begin{equation}
 x_{i_1}^{k_1} \cdots x_{i_l}^{k_l},
\label{eqn:proof_counting1}
\end{equation}
with $i_1, \ldots, i_l \in [m]$ different, that is, vectors $\mathbf{k} \in \mathcal{M}(k)$, with the frequencies equal to $f(\Gamma) = f(\Delta)$.
This is exactly the definition of $\mathcal{M}(k,l)[\Delta]$ in $\eqref{eqn:def_commutative_polinomials_boundeds_delta}$.
If we denote the factor with only norms with indices and multiplicities corresponding to $\mathbf{k}$ by $V(\mathbf{k})$, then
\begin{align}
\sum_{\substack{V \in B(\Gamma)\\ n(V) = l}} C_{V, \Gamma} V &= \sum_{\mathbf{k} \in \mathcal{M}(k,l)[\Delta]}  C_{V(\mathbf{k}), \Gamma} \prod_{i=1}^{m} |v_{i}|^{2k_i} \tag{$f(\Gamma) = f(\mathbf{k})$} \nonumber \\
& = \frac{r}{\binom{Tr}{n(\Delta)} \binom{n(\Delta)}{\mathbf{z}(\Delta)}} \sum_{\mathbf{k} \in \mathcal{M}(k,l)[\Delta]} \prod_{i=1}^{m} |v_{i}|^{2k_i} \tag{Lemma~\ref{lem:bound_prob_cycle}, and $n(\Gamma) = n(\Delta) = l$}
\end{align}

We show now \eqref{eqn:counting_cycles_to_polynomials}. 
Suppose $\Gamma$ has multiplicities $\mathbf{k}(\Gamma) = (k_1, \ldots, k_{l})$, and $V$ is such that $N(V) = \{1, \ldots, n(V)\} \subset [m]$, where $n(V) > n(\Gamma) = l$.
Let $\mathbf{g} = (g_1, \ldots, g_{n(V)})$ be the multiplicities of the indices in $V$, which satisfy
\begin{equation}
k = \sum_{i=1}^{n(V)} g_i.
\end{equation}
Using Cauchy-Schwartz in all inner products of $V$ we obtain the bound
\begin{equation}
|V| \leq \prod_{i=1}^{n(V)} |v_{i}|^{2g_i}.
\label{eqn:proof_counting2}
\end{equation}
Differently to the previous case, elements $t \in N(\Gamma)$ are not one to one with indices in $N(V)$. 
However, if $V \in B(\Gamma, \pi)$ for some permutation $\pi$ there is an assignment that maps each $t \in N(\Gamma)$ to a set $\mathcal{V}_t^{\pi} \cap N(V)$. 
In this case, the multiplicity $k_t$ of each $t \in N(\Gamma)$ must be equal to the sum of the multiplicities of the indices of $V$ in $\mathcal{V}_t^{\pi} \cap N(V)$.
That is, there exists a partition $\{\mathcal{A}_i\}_{i=1}^{n(\Gamma)}$ of $N(V)$ such that
\begin{equation}
k_t = \sum_{i \in \mathcal{A}_{t}} g_i \quad \text{ for each } \quad t \in N(\Gamma).
\label{eqn:proof_counting3}
\end{equation}
This is the domination condition from Definition~\ref{def:domination}.
We can thus upper bound the sum in \eqref{eqn:counting_cycles_to_polynomials} by a sum over multinomials such as \eqref{eqn:proof_counting2} whose multiplicities are dominated by $f(\Gamma) = f(\Delta)$, which is  $\mathcal{M}(k, d)[\Delta]$ by definition in \eqref{eqn:def_commutative_polinomials_boundeds_delta_larger}.
We are left to count how many $V \in B(\Gamma)$ have the same multinomial in \eqref{eqn:proof_counting2} that upper bounds the factor $V$.
Let $\mathcal{N}$ be this number.
Any other factor $V^{\prime}$ with \eqref{eqn:proof_counting2} as upper bound will have the same indices $i \in N(V^{\prime}) = N(V)$ appearing in $V^{\prime}$ with some other order.
We can upper bound the number of ways to order the indices by the number of ways to partition $k$ items in $n(V)$ labeled classes with multiplicities $\mathbf{g} = (g_1, \ldots, g_{n(V)})$, that is
\begin{equation}
\mathcal{N} \leq \mathcal{N}(\mathbf{g}) = \binom{k}{\mathbf{g}}.
\label{eqn:proof_counting5}
\end{equation}
The previous argument implies that if $\Gamma \in \bar{\Gamma} \in \Delta$, 
\begin{align}
\sum_{\substack{V \in B(\Gamma)\\ n(V) > l}} 
C_{V, \Gamma} V
&\leq 
\sum_{d=l+1}^{k} \sum_{\substack{V \in B(\Gamma)\\ n(V) = d}} C_{V, \Gamma} \prod_{i \in N(V)} |v_i|^{2g_i} \tag{From \eqref{eqn:proof_counting2}} \nonumber \\
&\leq 
\sum_{d=l+1}^{k} \sum_{\substack{V \in B(\Gamma)\\ n(V) = d}} \frac{n(\Gamma)! \stirling{n(V)}{n(\Gamma)}  r^{n(V) - n(\Gamma) + 1}}{\binom{Tr}{n(V)} n(V)!} \prod_{i \in N(V)} |v_i|^{2g_i}  \tag{Lemma~\ref{lem:bound_prob_cycle}} \nonumber \\
&\leq 
\sum_{d=l+1}^{k} \frac{n(\Gamma)! \stirling{d}{n(\Gamma)}  r^{d - n(\Gamma) + 1}}{\binom{Tr}{d} d!} \sum_{\substack{\mathbf{g} \in \mathcal{M}(k,d)[\Delta]}} \mathcal{N}(\mathbf{g}) \prod_{i=1}^{m} |v_i|^{2g_i}  \nonumber \\
&\leq 
\sum_{d=l+1}^{k} \frac{n(\Delta)! \stirling{d}{n(\Delta)}  r^{d - n(\Delta) + 1}}{\binom{Tr}{d} d!} \sum_{\substack{\mathbf{g} \in \mathcal{M}(k,d)[\Delta]}} \binom{k}{\mathbf{g}} \prod_{i=1}^{m} |v_i|^{2g_i} \tag{From \eqref{eqn:proof_counting5}}
\label{eqn:proof_counting4}
\end{align}
\end{proof}

Let 
\begin{align}
W_{\Delta}(T) &= \frac{r}{\binom{Tr}{n(\Delta)} \binom{n(\Delta)}{\mathbf{z}(\Delta)}} \label{eqn:def_W_Delta}\\
W_{d, \Delta}(T) &= \frac{n(\Delta)! \stirling{d}{n(\Delta)}  r^{d - n(\Delta) + 1}}{\binom{Tr}{d} d!} \quad \text{ if } \quad d > n(\Delta).
\label{eqn:def_W_d_Delta}
\end{align}
the following lemma characterizes the fact that the leading term of $F(k)$ is independent of $T$. 
The following asymptotic in $T$ will be key in the computation.
\begin{lemma}
Let $\Delta \in G(k)$, and recall that $|\Delta| = |\Delta|(T)$ is \eqref{eqn:order_orbit_delta} and depends on $T$. 
As $T \to \infty$:
\begin{itemize}
\item[(a)]
$
\begin{aligned}[t]
|\Delta|(T) W_{\Delta}(T) &= \frac{1 + O(T^{-1})}{r^{n(\Delta)-1}}
\end{aligned}
$
\item[(b)]
$
\begin{aligned}[t]
|\Delta|(T)  W_{d,\Delta}(T) &= \BigO{\frac{k^{k}}{r^{n(\Delta) -1 } T^{d - n(\Delta)}}}.
\end{aligned}
$

\end{itemize}
\label{lem:asymptotic_bound_T}
\end{lemma}
\begin{proof}
We show (a) first. 
From Lemma~\ref{lem:bound_prob_cycle}, and the bound in \eqref{eqn:order_orbit_delta}, we immediately have
\begin{align}
|\Delta|(T) W_{\Delta}(T) &= \binom{T}{n(\Delta)} \frac{r}{\binom{Tr}{n(\Delta)}} = \frac{r T(T-1) \cdots (T - n(\Delta) + 1)}{rT(rT-1) \cdots (rT - n(\Delta) + 1)} \nonumber \\
& = \frac{1 + O(T^{-1})}{r^{n(\Delta)-1}}.
\end{align}
For (b), we again use the bound of Lemma~\ref{lem:bound_prob_cycle}, and \eqref{eqn:order_orbit_delta} to obtain
\begin{align}
|\Delta|(T) W_{d, \Delta}(T) &= \binom{T}{n(\Delta)} \binom{n(\Delta)}{\mathbf{z}(\Delta)} \frac{n(\Delta)! \stirling{d}{n(\Delta)}  r^{d - n(\Delta) + 1}}{\binom{Tr}{d} d!} \nonumber \\
& \leq  \binom{T}{n(\Delta)} \frac{n(\Delta)! \stirling{d}{n(\Delta)}  r^{d - n(\Delta) + 1}}{\binom{Tr}{d}} \tag{Since $\binom{n(\Delta)}{\mathbf{z}(\Delta)}   \leq d!$, and $d > n(\Delta)$} \nonumber \\
& \leq \frac{T(T-1) \cdots (T - n(\Delta) + 1)}{rT(rT-1) \cdots (rT - d + 1)} d! \frac{n(\Delta)^{d}}{n(\Delta)!}r^{d - n(\Delta) + 1} \tag{Lemma~\ref{lem:stirling}, $\stirling{d}{n(\Delta)} \leq \frac{n(\Delta)^{d}}{n(\Delta)!}$} \nonumber \\
& = \BigO{\frac{k^{k}}{r^{n(\Delta) -1 } T^{d - n(\Delta)}}} \tag{From $\frac{n(\Delta)^{d}}{n(\Delta)!} \leq \frac{d^{d}}{d!}$ }.
\end{align}
\end{proof}

\subsection{Estimating the k-th moment}
\label{secappendix:moment_estimate}
We can now estimate the moments. 
Expanding $F(k)$, we obtain from the definition in \eqref{eqn:def_mathcal_G} that
\begin{align}
F(k) = \sum_{\Gamma = \{\ell, B(\ell)\} \in \mathcal{G}(k)} \expectationBig{\trBig{\prod_{i=1}^{k} \mathcal{L}_{\ell_i}^{\pi}(\theta)}}.
\end{align}

We present the asymptotic expansion of $F(k)$ as $T \to \infty$ in the following lemma.
\begin{lemma}
Let the expectation of $F(k)$ in \eqref{eqn:def_F_k} be with respect to a random permutation $\pi \sim \mathrm{Unif}(S_m)$ and phases $\exp(\ci \Theta_i) \sim \mathrm{Unif}(S^{1})$ for all $i \in [m]$. 
Then as $T \to \infty$ we have
\begin{equation}
F(k) =   
\sum_{l=1}^{k}
\frac{1}{r^{l-1}} \sum_{\mathbf{k} \in \mathcal{M}(k,l)} \binom{k}{\mathbf{k}} \prod_{i=1}^{m} |v_i|^{2k_i} + \BigO{\frac{k! k^{2k} r^{k}}{T}}.
\end{equation}
\label{lem:moment_bound_variance}
\end{lemma}
\begin{proof}
We expand
\begin{align}
F(k) &= \sum_{\Gamma \in \mathcal{G}(k)}
\sum_{V \in B(\Gamma)} C_{V, \Gamma} V \nonumber \\
&= \sum_{l=1}^{k} 
\sum_{\substack{\Gamma \in \mathcal{G}(k)\\ n(\Gamma) = l}} 
\left(\sum_{\substack{V \in B(\Gamma)\\ n(V) = l}} C_{V, \Gamma}  V  
+ \sum_{\substack{\gamma \in B(\Gamma)\\ n(V) > l}}  C_{V, \Gamma}  V  \right) \nonumber \\
& = I_1 + I_2.
\label{eqn:proof_variance_bound_two_terms}
\end{align}
We are left to bound both terms $I_1$ and $I_2$ from \eqref{eqn:proof_variance_bound_two_terms}.

\noindent
\emph{Bound of $I_1$.} 
The following inequalities hold:
\begingroup
\allowdisplaybreaks
\begin{align}
I_1 & = \sum_{l=1}^{k}
\sum_{\substack{\Delta \in G(k)\\ n(\Delta) = l}} 
\sum_{\bar{\Gamma} \in \Delta}
\sum_{\Gamma \in \bar{\Gamma}}
\sum_{\substack{V \in B(\Gamma)\\ n(V) = l}} C_{V, \Gamma} V \tag{From \eqref{eqn:def_mathcal_G_bar} and \eqref{eqn:def_mathcal_G}} \nonumber \\
& = \sum_{l=1}^{k}
\sum_{\substack{\Delta \in G(k)\\ n(\Delta) = l}} 
\sum_{\bar{\Gamma} \in \Delta}
\sum_{\Gamma \in \bar{\Gamma}} \sum_{\mathbf{k} \in \mathcal{M}(k,l)[\Delta]} W_{\Delta} \prod_{i=1}^{m} |v_i|^{2k_i} \tag{Lemma~\ref{lem:counting_cycles}, and \eqref{eqn:def_W_Delta}}\nonumber \\
& = \sum_{l=1}^{k}
\sum_{\substack{\Delta \in G(k)\\ n(\Delta) = l}} 
|\Delta| |\bar{\Gamma}| \sum_{\mathbf{k} \in \mathcal{M}(k,l)[\Delta]} W_{\Delta} \prod_{i=1}^{m} |v_i|^{2k_i} \nonumber \\
& = \sum_{l=1}^{k}
\sum_{\substack{\Delta \in G(k)\\ n(\Delta) = l}} 
|\Delta| W_{\Delta} \binom{k}{\mathbf{k}(\Delta)} \sum_{\mathbf{k} \in \mathcal{M}(k,l)[\Delta]}   \prod_{i=1}^{m} |v_i|^{2k_i} \tag{From \eqref{eqn:proof_variance_bounded_combinatorial_num}}\nonumber \\
& = \sum_{l=1}^{k}
\sum_{\substack{\Delta \in G(k)\\ n(\Delta) = l}} 
|\Delta| W_{\Delta}  \sum_{\mathbf{k} \in \mathcal{M}(k,l)[\Delta]} \binom{k}{\mathbf{k}}  \prod_{i=1}^{m} |v_i|^{2k_i} \tag{From $f(\Delta) = f(\mathbf{k})$} \nonumber \\
& = \sum_{l=1}^{k} \frac{1 + O(T^{-1})}{r^{l-1}}
\sum_{\substack{\Delta \in G(k)\\ n(\Delta) = l}} 
\sum_{\mathbf{k} \in \mathcal{M}(k,l)[\Delta]} \binom{k}{\mathbf{k}} \prod_{i=1}^{m} |v_i|^{2k_i} \tag{Lemma~\ref{lem:asymptotic_bound_T}} \nonumber \\ \allowbreak
& = \sum_{l=1}^{k}
\frac{1}{r^{l-1}} \sum_{\mathbf{k} \in \mathcal{M}(k,l)} \binom{k}{\mathbf{k}} \prod_{i=1}^{m} |v_i|^{2k_i} + \BigO{\frac{r^k}{T}}, \label{eqn:proof_variance_bound_main_term_bound}
\end{align}
\endgroup
where in \eqref{eqn:proof_variance_bound_main_term_bound} we have used \eqref{eqn:def_M_k_l}, and the bound
\begin{equation}
\sum_{\mathbf{k} \in \mathcal{M}(k,l)} \binom{k}{\mathbf{k}} \prod_{i=1}^{m} |v_i|^{2k_i} \leq \sum_{\mathbf{k} \in \mathcal{M}(k)} \binom{k}{\mathbf{k}} \prod_{i=1}^{m} |v_i|^{2k_i} = \Bigl( \sum_{i=1}^{m} |v_i|^2 \Bigr)^k \eqcom{\ref{eqn:proof_variance_example_normalization}}= r^k.
\label{eqn:upper_bound_M_k}
\end{equation}

\noindent
\emph{Bound of $I_2$.}
Recall that if $\Gamma \in \bar{\Gamma} \in \Delta$, we have $n(\Delta) = n(\Gamma)$.
\begingroup
\allowdisplaybreaks
\begin{align}
I_2  &  =   \sum_{l=1}^{k} 
\sum_{\substack{\Delta \in G(k)\\ n(\Delta) = l}} 
\sum_{\bar{\Gamma} \in \Delta}
\sum_{\Gamma \in \bar{\Gamma}}
\sum_{\substack{V \in B(\Gamma)\\ n(V) > l}}  C_{V, \Gamma} V 
\tag{From \eqref{eqn:def_mathcal_G_bar}, and \eqref{eqn:def_mathcal_G}}\nonumber \\
&  \leq   \sum_{l=1}^{k} 
\sum_{\substack{\Delta \in G(k)\\ n(\Delta) = l}}  
\sum_{\bar{\Gamma} \in \Delta}
\sum_{\Gamma \in \bar{\Gamma}}
\sum_{d=l+1}^{k} W_{d, \Delta} \sum_{\substack{\mathbf{g} \in \mathcal{M}(k,d)[\Delta]}} \binom{k}{\mathbf{g}} \prod_{i=1}^{m} |v_i|^{2g_i} 
\tag{Lemma~\ref{lem:counting_cycles}}
\nonumber \\
&  \leq   \sum_{l=1}^{k} 
\sum_{\substack{\Delta \in G(k)\\ n(\Delta) = l}}  
|\Delta| |\bar{\Gamma}|
\sum_{d=l+1}^{k}  W_{d, \Delta} \sum_{\substack{\mathbf{g} \in \mathcal{M}(k,d)[\Delta]}} \binom{k}{\mathbf{g}} \prod_{i=1}^{m} |v_i|^{2g_i} 
\nonumber \\
&  \leq   \sum_{l=1}^{k} 
\sum_{\substack{\Delta \in G(k)\\ n(\Delta) = l}}  
\sum_{d=l+1}^{k} |\Delta| W_{d, \Delta} \binom{k}{\mathbf{k}(\Delta)}\sum_{\substack{\mathbf{g} \in \mathcal{M}(k,d)[\Delta]}} \binom{k}{\mathbf{g}} \prod_{i=1}^{m} |v_i|^{2g_i}
\tag{From \eqref{eqn:proof_variance_bounded_combinatorial_num}} 
\nonumber \\
&  \leq   \sum_{l=1}^{k} 
\sum_{\substack{\Delta \in G(k)\\ n(\Delta) = l}}  
\sum_{d=l+1}^{k}   \BigO{\frac{k^{k}}{r^{l - 1 } T^{d - l}}}   \binom{k}{\mathbf{k}(\Delta)}\sum_{\substack{\mathbf{g} \in \mathcal{M}(k,d)[\Delta]}} \binom{k}{\mathbf{g}} \prod_{i=1}^{m} |v_i|^{2g_i}
\tag{Lemma~\ref{lem:asymptotic_bound_T}} 
\nonumber \\
&  \leq  \BigO{\frac{k! k^{k}}{T}}  
\sum_{l=1}^{k} 
\sum_{\substack{\Delta \in G(k)\\ n(\Delta) = l}}  
\sum_{d=l+1}^{k} \sum_{\substack{\mathbf{g} \in \mathcal{M}(k,d)[\Delta]}} 
\binom{k}{\mathbf{g}} \prod_{i=1}^{m} |v_i|^{2g_i} 
\tag{From $\binom{k}{\mathbf{k}(\Delta)} \leq k!$}
 \nonumber \\
 &  \leq  \BigO{\frac{k! k^{k}}{T}}  
\sum_{l=1}^{k}
\sum_{\substack{\Delta \in G(k)\\ n(\Delta) = l}} \sum_{\substack{\mathbf{g} \in \mathcal{M}(k)}} 
\binom{k}{\mathbf{g}} \prod_{i=1}^{m} |v_i|^{2g_i} \tag{$\sqcup_{d=l+1}^{k} \mathcal{M}(k,d)[\Delta] \subset \mathcal{M}(k)$}
 \nonumber \\
  &  \leq  \BigO{\frac{k! k^{k}}{T}}  
\sum_{l=1}^{k}
\sum_{\substack{\Delta \in G(k)\\ n(\Delta) = l}} r^{k} 
 \tag{From \eqref{eqn:upper_bound_M_k}} \nonumber \\
   &  \leq  \BigO{\frac{k! k^{k} r^{k}}{T}} |G(k)|  \nonumber \\
    &  =  \BigO{\frac{k! k^{2k} r^{k}}{T}}, \label{eqn:proof_variance_bound_error_term}
\end{align}
\endgroup
where in \eqref{eqn:proof_variance_bound_error_term} we have used that $|G(k)| \leq k^{k}$.
This follows from the fact that each $\Delta \in G(k)$ is in one-to-one correspondance with $f(\Delta) \in \mathfrak{F}(k)$ by Lemma~\ref{lem:disjoint_union_frequencies}, and the number of multiplicities that can be attained is far less than $k^{k}$.
Combining the terms \eqref{eqn:proof_variance_bound_main_term_bound} and \eqref{eqn:proof_variance_bound_error_term} yields the claim.
\end{proof}

We are left to estimate the leading term in Lemma~\ref{lem:moment_bound_variance}, which will finally show Lemma~\ref{lem:variance_bound_moment_bound}.
\begin{lemma}
Let $\{v_i\}_{i=1}^{m} \subset \C^{r}$ be a frame in dimension $r$ satisfying \eqref{eqn:proof_variance_example_normalization}. There exists a constant $c>0$ independent of $k$ and $r$ such that
\begin{equation}
\sum_{l=1}^{k}
\frac{1}{r^{l-1}} \sum_{\mathbf{k} \in \mathcal{M}(k,l)} \binom{k}{\mathbf{k}} \prod_{i=1}^{m} |v_i|^{2k_i} \leq c r k^{k}.
\end{equation}
\end{lemma}
\begin{proof}
Define
\begin{equation}
M_l(k) = \sum_{\mathbf{k} \in \mathcal{M}(k,l)} \binom{k}{\mathbf{k}} \prod_{i=1}^{m} |v_i|^{2k_i}.
\label{eqn:proof_variance_bound_M}
\end{equation}
In particular, from \eqref{eqn:proof_variance_example_normalization} and the fact that $|v_i| \leq 1$ we have that for any $k \geq 1$,
\begin{equation}
M_1(k) = \sum_{i=1}^{m} |v_i|^{2k} \leq r.
\label{eqn:proof_variance_bound_M1}
\end{equation}
In \eqref{eqn:proof_variance_bound_M}, note that up to multiplicities, there are exactly $l$ different vector norms appearing in each multinomial. 
From comparing coefficients, and since all are positive, we can readily see that the following inequality holds for all $l \leq k$,
\begin{equation}
M_l(k) \leq \sum_{\substack{\mathbf{b} = (b_1, \ldots, b_{l})\\ \sum_{i=1}^{l} b_i = k, b_i \geq 0}} \binom{k}{\mathbf{b}} \prod_{i=1}^{l} M_1(b_i).
\label{eqn:proof_variance_bound_M2}
\end{equation}
Using \eqref{eqn:proof_variance_bound_M1} in \eqref{eqn:proof_variance_bound_M2}, and the multinomial coefficient formula
\begin{equation}
\sum_{\substack{\mathbf{b} = (b_1, \ldots, b_{l})\\ \sum_{i=1}^{l} b_i = k, b_i \geq 0}} \binom{k}{\mathbf{b}} \leq l^{k},
\end{equation}
yields the bound
\begin{equation}
M_l(k) \leq r^{l} \sum_{\substack{\mathbf{b} = (b_1, \ldots, b_{l})\\ \sum_{i=1}^{l} b_i = k, b_i > 0}} \binom{k}{\mathbf{b}} \leq r^{l} l^{k}.
\end{equation}
Using this last bound, we have the following inequality
\begin{align}
\sum_{l=1}^{k}
\frac{1}{r^{l-1}} \sum_{\mathbf{k} \in \mathcal{M}(k,l)} \binom{k}{\mathbf{k}} \prod_{i=1}^{k} |v_i|^{2k_i} & \leq \sum_{l=1}^{k}
\frac{1}{r^{l-1}} r^{l} l^{k} \nonumber \\
& \leq r \sum_{l=1}^{k} l^{k} \nonumber \\
& \leq c  r k^{k},
\end{align}
where in the last step we have used that there exists a constant $c>0$ independent of $k \geq 1$ such that 
\begin{equation}
\sum_{l=1}^{k} l^{k} \leq c \int_{1}^{k} x^{k} dx \leq c \frac{k^{k+1}}{k+1} \leq c k^k.
\end{equation}
\end{proof}

\section{Appendix on numerical experiments}
\label{secappendix:numerical_detials}

Our implementation code in \texttt{Python} can be shared upon request.
The following are other comments about the implementation.
\begin{itemize}
\item \emph{Projection on the tangent space:} Consider $U(\delta) = \exp(\delta A) U$ in \eqref{eqn:numerics_1} and differentiate with respect to $\delta$. If we let $U E_i U^{\dagger} = P_i(U)$, by evaluating the gradient at $\delta = 0$ we obtain its projection onto the tangent space of $U(n)$ at $U$, which is
\begin{equation}
P_n(U) = \sum_{i=1}^{n} \Bigl( \rho P_i(U) \tau P_i(U) - P_i(U) \tau P_i(U) \rho \Bigr) + 
\Bigl( \tau P_i(U) \rho P_i(U) - P_i(U) \rho P_i(U) \tau \Bigr).
\end{equation}
This matrix is skew-hermitian and belongs to the Lie algebra of $U(n)$.

\item \emph{Initialization of quantum states:} We initialize $\rho$ by evaluating $W \sim \mathrm{Unif}(U(r))$, the vector $d = (d_1, \ldots, d_r)$ uniformly distributed in the simplex, and setting $\rho = W~\mathrm{Diag}(d)~
W^{\dagger}$, where $\mathrm{Diag}(d)$ is the matrix with $d$ in the diagonal and zeros elsewhere. We perform a similar, and independent initialization of $\tau$.

\item \emph{Landscape complexity:} From trying different settings, our simulations suggests that there are several local maxima or saddle points in the optimization landscape of \eqref{eqn:maximization_unitary}, and convergence to those is dependent on the initialization.
\end{itemize}

\end{document}